\documentclass[pdflatex]{sn-jnl}

\makeatletter
\@twosidefalse
\@mparswitchfalse
\makeatother
\usepackage{fullpage}

\usepackage{amsmath,amssymb,amsfonts}
\usepackage{amsthm}

\usepackage{placeins}

\usepackage{booktabs}

\newcounter{tablerow}
\newcommand{\RowNumber}{\refstepcounter{tablerow}(\arabic{tablerow})}

\usepackage{pifont}
\newcommand{\vmark}{\ding{51}}%
\newcommand{\xmark}{\ding{55}}%

\usepackage{multirow}

\usepackage{algpseudocode}
\usepackage{algorithm}
\algtext*{EndWhile} 
\algtext*{EndIf} 
\algtext*{EndFor} 

\algblockdefx{Event}{EndEvent}
{\textbf{event:}~}
{\textbf{end event}}
\algtext*{EndEvent} 

\newtheorem{GlobalTheorem}{Theorem}
\newtheorem{guarantee}{Guarantee}
\newtheorem{theorem}{Theorem}[section]
\newtheorem{lemma}[theorem]{Lemma}
\newtheorem{observation}[theorem]{Observation}
\newtheorem{corollary}[theorem]{Corollary}
\newtheorem{proposition}[theorem]{Proposition}

\newcommand{\Integers}{\mathbb{Z}}
\newcommand{\Reals}{\mathbb{R}}
\providecommand{\Pr}{}\renewcommand{\Pr}{\mathbb{P}}
\providecommand{\Ex}{}\renewcommand{\Ex}{\mathbb{E}}
\providecommand{\argmax}{}\renewcommand{\argmax}{\operatorname{argmax}}
\newcommand{\MaxTwo}{\operatorname{max2}}
\newcommand{\MaxOne}{\operatorname{max}}

\newcommand{\Busy}{\mathcal{B}}
\newcommand{\Channel}{\operatorname{channel}}
\newcommand{\Peer}{\operatorname{peer}}
\newcommand{\Channels}{\mathcal{C}}
\newcommand{\OperationalNodes}{\mathcal{OP}}
\newcommand{\OperationalChannels}{\mathcal{OC}}
\newcommand{\Next}{\Phi}
\newcommand{\AS}{\operatorname{AS}}

\providecommand{\A}{}\renewcommand{\A}{\mathcal{A}}
\providecommand{\T}{}\renewcommand{\T}{\mathcal{T}}
\providecommand{\M}{}\renewcommand{\M}{\mathcal{M}}

\newcommand{\Alg}{\ensuremath{\mathtt{Alg}}}
\newcommand{\AlgTF}{\ensuremath{\Alg_{\mathrm{TF}}}}
\newcommand{\ProcCreateChannel}{\mathtt{CreateChannel}}
\newcommand{\Adv}{\ensuremath{\mathtt{Adv}}}

\newcommand{\ProcCheckEndPhase}{\mathtt{CheckEndPhase}}

\newcommand{\ProcFormTeams}{\mathtt{FormTeams}}
\newcommand{\ProcBeginNewPhase}{\mathtt{BeginNewPhase}}

\newcommand{\Object}{\mathtt{obj}}
\newcommand{\varTokens}{\mathtt{tok}}
\newcommand{\varBusyTokens}{\mathtt{busy}\text{-}\mathtt{toks}}
\newcommand{\varMediators}{\mathtt{meds}}
\newcommand{\varChannel}{\mathtt{chan}}
\newcommand{\varBusyAcked}{\mathtt{busy}\text{-}\mathtt{acked}}
\newcommand{\varDiff}{\mathtt{diff}}
\newcommand{\varPhaseType}{\mathtt{phase}\text{-}\mathtt{type}}
\newcommand{\varAwaitingResponse}{\mathtt{awaiting}\text{-}\mathtt{resp}}
\newcommand{\varDelayingResponse}{\mathtt{delaying}\text{-}\mathtt{resp}}

\newcommand{\msgTokensPlease}{\mathsf{TokensPlease}}
\newcommand{\msgWaiting}{\mathsf{Waiting}}
\newcommand{\msgTransport}{\mathsf{Transport}}
\newcommand{\msgNoTransport}{\mathsf{NoTransport}}
\newcommand{\msgGoOn}{\mathsf{GoOn}}

\newcommand{\msgBusy}{\mathsf{Busy}}
\newcommand{\msgBusyAck}{\mathsf{BusyAck}}
\newcommand{\msgTokensUpdate}{\mathsf{TokensUpdate}}
\newcommand{\msgNotBusy}{\mathsf{NotBusy}}
\newcommand{\msgChannel}{\mathsf{Channel}}
\newcommand{\msgNoChannel}{\mathsf{NoChannel}}
\newcommand{\msgChannelAck}{\mathsf{ChannelAck}}

\newcommand{\msgShortB}{\mathsf{B}}
\newcommand{\msgShortBA}{\mathsf{BA}}
\newcommand{\msgShortTU}{\mathsf{TU}}
\newcommand{\msgShortNB}{\mathsf{NB}}
\newcommand{\msgShortC}{\mathsf{C}}
\newcommand{\msgShortNC}{\mathsf{NC}}

\providecommand{\True}{}\renewcommand{\True}{\mathsf{true}}
\providecommand{\False}{}\renewcommand{\False}{\mathsf{false}}
\newcommand{\Center}{\mathsf{center}}
\newcommand{\Arm}{\mathsf{arm}}

\newcommand{\Paragraph}[1]{\vskip3mm\par\noindent\textbf{#1} }
\newcommand{\Subparagraph}[1]{\vskip1mm\par\textit{#1} }

\newcommand{\Sect}{Sec.}
\newcommand{\Thm}{Thm.}
\newcommand{\Prop}{Prop.}
\newcommand{\Lem}{Lem.}
\newcommand{\Obs}{Obs.}
\newcommand{\Grnt}{Grnt.}

\title{Team Formation and Applications}

\author[1]{\fnm{Yuval} \sur{Emek} \orcid{https://orcid.org/0000-0002-3123-3451}}\email{yemek@technion.ac.il}

\author[1]{\fnm{Shay} \sur{Kutten} \orcid{https://orcid.org/0000-0003-2062-6855}}\email{kutten@technion.ac.il}

\author*[1]{\fnm{Ido} \sur{Rafael} \orcid{https://orcid.org/0000-0003-4923-4660}}\email{ido.rafael@campus.technion.ac.il}

\author[2]{\fnm{Gadi} \sur{Taubenfeld} \orcid{https://orcid.org/0000-0003-3070-5370}}\email{tgadi@runi.ac.il}

\equalcont{All authors contributed equally to this work.}

\affil[1]{\orgdiv{Faculty of Data and Decision Sciences}, \orgname{Technion --- Israel Institute of Technology}, \orgaddress{\street{Technion City}, \city{Haifa}, \postcode{3200003}, \country{Israel}}}

\affil[2]{\orgdiv{School of Computer Science}, \orgname{Reichman University}, \orgaddress{\street{8 University Street}, \city{Herzliya}, \postcode{4610101}, \country{Israel}}}

\abstract{A novel long-lived distributed problem, called \textit{Team Formation (TF)},
is introduced together with a message- and time-efficient randomized algorithm.
The problem is defined over the asynchronous model with a complete
communication graph, using bounded size messages, where a certain fraction of
the nodes may experience a generalized, strictly stronger, version of initial
failures.
The goal of a TF algorithm is to assemble tokens injected by the environment,
in a distributed manner, into \emph{teams} of size $\sigma$, where $\sigma$ is
a parameter of the problem.

The usefulness of TF is demonstrated by using it to derive efficient algorithms for many distributed problems.
Specifically, we show that various (one-shot as well as long-lived) distributed problems reduce to TF.
This includes well-known (and extensively studied) distributed problems such as several versions of leader election and threshold detection.
For example, we are the first to break the linear message complexity bound for asynchronous implicit leader election.
We also improve the time complexity of message-optimal algorithms for asynchronous explicit leader election.
Other distributed problems that reduce to TF are new ones, including matching players in online gaming platforms, a generalization of gathering, constructing a perfect matching in an induced
subgraph of the complete graph, and more.
To complement our positive contribution, we establish a tight lower bound on the message complexity of TF algorithms.}


\begin{document}

\maketitle

\section*{Funding}

\begin{description}
\item[Yuval Emek:] Partially supported by an Israel Science Foundation (ISF) grant 730/24 and by the Grand Technion Energy Program (GTEP).

\item[Shay Kutten:] Partially supported by an Israel Science Foundation (ISF) grant 1346/22,
by the Grand Technion Energy Program (GTEP),
and by the Technion Gordon Center.
\end{description}

\section*{Acknowledgments}
We would like to thank Eviatar Procaccia for his help with the proof of Lemma~\ref{lem:average-exponentially-decaying-rvs}.

\setcounter{tocdepth}{3}
\tableofcontents


\section{Introduction}
\label{sec:intro}
Consider the following three problems, defined over an asynchronous
message-passing system with a complete communication graph, where a constant
fraction of the nodes may be faulty:
(P1)
electing a leader among the non-faulty nodes;
(P2)
constructing a perfect matching in the subgraph induced by a subset of the
non-faulty nodes, specified (distributively) by the input;
(P3)
assembling Dungeons \& Dragons parties from characters that arrive over time
at the non-faulty nodes so that each party includes
$3$ wizards,
$2$ paladins,
$2$ rogues,
and
$1$ monk.
On the face of it, problems (P1)--(P3) have little in common.
In particular, (P1) and (P2) are one-shot problems, whereas (P3) is long-lived;
(P1) requires global symmetry breaking, whereas in (P2) and (P3) the
symmetry breaking is local in essence.
Therefore, it may come as a surprise that all three problems reduce to a
single new problem.

In this paper, we introduce the aforementioned new problem, called \emph{team
formation}, and develop an efficient algorithm for it.
Following that, we show how our algorithm leads to algorithmic solutions for
various other problems (including (P1)--(P3)), improving their
state-of-the-art.
For example, we are the first to break the linear message complexity bound for asynchronous implicit leader election, see \Sect{}~\ref{sec:implicit}. 
We also improve the time complexity of message-optimal algorithms for asynchronous explicit leader election, see \Sect{}~\ref{sec:explicit}.
%
On the negative side, we establish lower bounds on the communication demands
of TF algorithms, see \Sect{}~\ref{sec:lower-bound}.

\subsection{The Basic Setting and Problem Definition}
\label{sec:intro:problem}
Consider an \emph{asynchronous} message-passing system that consists of $n$
nodes (a.k.a.\ processors).
The communication structure is assumed to be a complete undirected graph over
the node set $V$ so that every two nodes may exchange messages with each
other.
These messages carry
$O (\log n)$
bits of information (cf.\ the CONGEST model \cite{Peleg2000book}) and are
delivered with a finite delay, where the delay of each message is determined
(individually) by an adversary.
Unless stated otherwise, it is assumed that the nodes are anonymous and that
each node distinguishes between its
$n - 1$
neighbors by means of (locally) unique \emph{port numbers}
\cite{DBLP:conf/stoc/Angluin80,
DBLP:journals/dc/HellaJKLLLSV15,
DBLP:conf/podc/YamashitaK88}.
The adversary may cause some nodes to become faulty, preventing them from
participating in the execution in accordance with a generalization of the
\emph{initial failures} model
\cite{DBLP:journals/cacm/Dijkstra65,
DBLP:journals/jacm/FischerLP85,
DBLP:journals/ijpp/TaubenfeldKM89},
as long as a sufficiently large fraction of the nodes remain
non-faulty throughout the execution.\footnote{%
The failure model considered in this paper is identical to the classic model of initial failures when restricted to one-shot problems.
As explained in the sequel, the generalization comes into play only in the context of long-lived problems.}
A precise definition of the computational model is presented in
\Sect{}~\ref{sec:preliminaries:model}.

\Paragraph{Team Formation.}
We introduce a long-lived distributed problem called \emph{team formation
(TF)}, defined over an integral \emph{team size} parameter
$2 \leq \sigma \leq n$.\footnote{%
The TF problem is well defined also for
$\sigma > n$,
however, the restriction to
$\sigma \leq n$
simplifies the discussion and is consistent with the applications presented in
\Sect{}~\ref{sec:intro:applications}.}
In a TF instance, abstract \emph{tokens} are injected into the nodes over
time, where the adversary determines the timing and location of these token
injections.
Tokens held by a node
$v \in V$
can be transported to a node
$v' \in V$
over a message sent from $v$ to $v'$.
Each token remains in the system until it is deleted, as explained below.

The correctness criterion of the TF problem is captured by the following two
conditions:
\begin{description}

\item[Safety:]
a token can be deleted from the system only as part of a \emph{team} that
consists of exactly $\sigma$ tokens, all of which are held by the same node
and deleted simultaneously;
the operation of deleting a team of tokens is referred to as \emph{forming a
team}.

\item[Liveness:]
if the system contains at least $\sigma$ tokens, then a team must be formed in
a finite time.

\end{description}
Assume without loss of generality that the system contains at most
$n^{O (1)}$
tokens at any given time.\footnote{%
Conditioned on the assumption that
$\sigma \leq n$,
the nodes themselves should hold (all together) less than
$n^{2}$
tokens since a node that holds at least $\sigma$ tokens can perform team
formation operations.
Regarding the tokens in transit, the algorithm developed in the current paper
is designed so that there are less than $\sigma$ tokens in transit over any
edge at any given time, so the total number of tokens in transit is
$O (n^{3})$.}
The tokens are assumed to be indistinguishable and any number of tokens can be
transported over a single message that simply encodes their number.

\subsection{Our Contribution}
\label{sec:intro:contribution}
Our main contribution regarding the TF problem is cast in
\Thm{}~\ref{thm:main-informal} (refer to \Thm{}\
\ref{thm:analysis:safety-and-liveness}, \ref{thm:analysis:message-load}, and
\ref{thm:reaction-time} for the formal statement).

\begin{GlobalTheorem}[slightly informal]
\label{thm:main-informal}
For every constant
$\epsilon > 0$,
there exists a randomized algorithm that solves any $n$-node TF instance whp
if at least
$\epsilon n$
nodes remain non-faulty indefinitely.\footnote{%
An event $A$ occurs \emph{with high probability (whp)} if
$\Pr (A) \geq 1 - n^{-c}$
for any (arbitrarily large) constant $c$.}
The algorithm sends
$O (\sqrt{n \log n})$
messages per token in expectation and
$O (\sqrt{n \log n} \cdot \log n)$
messages per token whp.
Moreover, if the system contains at least $\sigma$ tokens, then the algorithm
is guaranteed to form a team in
$O (\sigma + \log n)$
(asynchronous) time whp.
\end{GlobalTheorem}

We emphasize that the whp guarantees promised in
\Thm{}~\ref{thm:main-informal} apply for any TF instance, regardless of the
total number of tokens injected into the system that may be arbitrarily large
(and not necessarily bounded as a function of $n$);
in particular, the correctness probability remains (arbitrarily) close to $1$
and does not deteriorate as the number of tokens increases.
Moreover, the bounds on the number of messages per token hold already for
``the first tokens'', that is, even if the adversary decides to inject a few
tokens in total, regardless of $\sigma$.
It is also interesting to point out that our TF algorithm withstands a
$(1 - \epsilon)$-fraction
of faulty nodes for an arbitrarily small constant
$\epsilon > 0$;\footnote{%
Refer to \Sect{}~\ref{sec:preliminaries:model} for a precise definition of the failure model adopted in the current paper, including the assumptions made on the faulty nodes.}
this is in contrast to many other problems (in message-passing with initial
failures), including leader election and consensus, where a majority of faulty
nodes leads to impossibility results.
On the negative side, we establish the following theorem (refer to
Corollary~\ref{cor:lowerbound-max} for the formal statement).

\begin{GlobalTheorem}[slightly informal]\label{thm:lowerbound-informal}
For
$2 \leq \sigma \leq n / 2$,
consider the simplified (one-shot) version of the TF problem, where the
schedule is synchronous and it is guaranteed that exactly
$\sigma$ tokens are
injected into the system, all at the beginning of the execution.
Any algorithm that solves this problem whp must send a total of
$\Omega ( \max \{ \sqrt{n \log n}, \sqrt{n \sigma} \} )$
messages in expectation.
This holds even if the messages are of unbounded size.
\end{GlobalTheorem}

\subsection{Extra Features}
\label{sec:intro:extra-features}
Beyond the safety and liveness conditions presented in
\Sect{}~\ref{sec:intro:problem}, it may be advantageous for TF algorithms to
satisfy additional features.
In this section, we introduce three such desirable features that turn out to
be very useful (see \Sect{}~\ref{sec:intro:applications}), all
three of them are readily satisfied by the TF algorithm promised in
\Thm{}~\ref{thm:main-informal}.
Refer to \Sect{}~\ref{sec:discussion} for further discussions of the qualities
of these features.

\Paragraph{The Forgetful Feature.}
In the scope of this feature, we distinguish between two types of coin tosses of a node
$v \in V$:
(1)
the finitely many ``factory coin tosses'' generated by $v$ upon the first activation event of $v$, before any other action of $v$ and
hence, could have been generated when $v$ was ``manufactured'';
(2)
all other coin tosses of $v$, generated along the execution.
In contrast to the latter, the ``factory coin tosses'' are considered to be
hardwired into $v$'s memory and, in particular, constitute part of $v$'s
initial state.

We say that a randomized algorithm satisfies the \emph{forgetful feature} if
it is guaranteed that every node resides in its initial state at any quiescent
time (formally defined in \Sect{}~\ref{sec:preliminaries:model}).
Thus, at quiescent times, the nodes may not record any information about their
past events and actions.
We view the notion of forgetful algorithms as a natural extension of the
important notion of \emph{memoryless} algorithms
\cite{DBLP:conf/podc/StyerP89,
gadi-book}
from deterministic to randomized algorithms and believe that it should be
studied further, regardless of the TF problem.

\Paragraph{The Trace-Tree Feature.}
A token $\tau$ injected into node
$v \in V$
follows a tour in the communication graph from $v$ to the node $r$ at which
$\tau$ is deleted (if at all) as part of a team formation event $F$.
Somewhat informally, this tour forms a simple path $\pi$ in the temporal graph
that reflects the communication graph along the execution's time axis (a node
$v' \in V$
may appear in $\pi$ several times, but at different times).
The collection of the paths $\pi$ of all the tokens $\tau$ that participate in
$F$ form a temporal tree $T$ rooted at $r$.
A TF algorithm satisfies the \emph{trace-tree feature} if the algorithm
maintains a distributed data structure that supports broadcast and echo
processes over $T$ (it is the responsibility of ``the user'' to delete this
data structure once it is no longer used).
Since the basic description of our TF algorithm, as presented in
\Sect{}\ \ref{sec:algorithm} and \ref{sec:channel-layer-algorithm}, does not
address this trace-tree feature, we provide the implementation details in
\Sect{}~\ref{sec:trace-tree-mechanism-details}.

\Paragraph{The Accumulation Feature.}
Consider a TF instance and assume that the total number $\ell$ of tokens
injected into the system satisfies
$\ell \bmod \sigma = k > 0$.
A TF algorithm satisfies the \emph{accumulation feature} if it is guaranteed
that the $k$ tokens that remain in the system forever are eventually held by
the same node.

\subsection{Applications}
\label{sec:intro:applications}
We present several interesting problems that are reducible to TF.
These include the classic leader election problem, for which we improve the
state-of-the-art, as well as various problems that received less attention
from the community so far (if at all).
For each problem, we provide a sketchy description of the reduction to TF;
further technical details as well as precise statements of the results
obtained through these reductions are deferred to
\Sect{}~\ref{sec:applications-additional-information}.


\Paragraph{Leader Election (LE).}
This is the fundamental one-shot problem of designating a single non-faulty
node of a communication network as a leader
\cite{DBLP:books/mk/Lynch96,
DBLP:books/daglib/0017536,
Peleg2000book}.
We improve upon the previous results for both
(1)
the \textit{implicit} version, where each node is required to know whether it
is the elected leader;
and
(2)
the \emph{explicit} version, where, in addition to the above requirement, each
non-leader knows which of its (internal) ports leads to the leader.
Moreover, our algorithms are fault-tolerant (see \Sect{}~\ref{sec:preliminaries:model}), while the improved upon prior art
\cite{DBLP:conf/wdag/KuttenMPP21,
kutten2020singularly}
assumes fault freedom.

Intuitively, the reduction works as follows:
A logarithmic number of nodes are chosen probabilistically to serve as
candidates, injecting a token into each one of them.
The team size parameter is adjusted so that the number of injected tokens
suffices for the formation of exactly one team whp.
The node at which the team is formed is elected.
In the explicit LE version, on top of the above, the elected leader notifies
all other nodes.

We assume that there are at most
$n (1 / 2 - \epsilon)$
faulty nodes for any constant
$\epsilon > 0$.
The run-time of our LE algorithms is
$O (\log n)$
whp. 
Our \emph{explicit} LE algorithm sends
$O (n)$
messages in expectation and whp.
This improves upon the time complexity of the best previous algorithms with $O(n)$ messages for asynchronous explicit LE
\cite{DBLP:conf/wdag/KuttenMPP21,
kutten2020singularly};
their time complexity was $(\log^2 n)$ even though they assume fault-freedom.
Our \emph{implicit} LE algorithm sends
$O (\sqrt{n \log n} \cdot \log n)$
messages in expectation and
$O (\sqrt{n \log n} \cdot \log ^2 n)$
messages whp.
This is the first algorithm with message complexity $o(n)$ for asynchronous LE (even for fault-free algorithms).

\Paragraph{Vector Team Formation.}
Consider a (long-lived) generalization of the TF problem, where each token
comes with a \emph{color} from a certain color palette and a team should
include a pre-specified number of tokens from each color.
This problem, called \emph{vector team formation (vTF)}, is formulated as
follows:
Instead of a scalar team size parameter
$\sigma \in \Integers_{> 0}$,
the vTF problem is defined over a \emph{team size vector}
$\vec{\sigma} \in \Integers_{> 0}^{m}$
whose dimension $m$ corresponds to the size of the color palette;
a team now consists of (exactly)
$\vec{\sigma}(i)$
tokens of color $i$ for each
$i \in [m]$.
As in the TF problem, the safety condition states that tokens can be deleted
only as part of a team, whereas the liveness condition states that if the
system contains at least
$\vec{\sigma}(i)$
tokens of color $i$ for each
$i \in [m]$,
then a team must be formed in a finite time.
(The aforementioned Dungeons \& Dragons example is a special case of vTF.)

To solve vTF, we shall invoke a separate copy of a TF algorithm for each color
$i \in [m]$,
setting $\sigma = \vec{\sigma}(i)$,
and generate a ``super-token'' of color $i$ whenever the $i$-th copy forms a
team.
The super-tokens can then be collected into a vTF team, exploiting the fact
that the symmetry is now broken due to the different colors.

\Paragraph{Agreement with Failures.}
The problem of (explicit) agreement with initial failures was defined in the
seminal paper of \cite{DBLP:journals/jacm/FischerLP85} as a contrast to the setting
of asynchronous crash failures proved to be impossible.
A weaker version of the problem, known as \emph{implicit} agreement, is
defined in \cite{DBLP:conf/podc/AugustineMP18}.
Informally, the latter problem requires that
(1)
at least one node must decide;
(2)
the decided value must be the same for all deciding nodes;
and
(3)
the decided value must be the input of some node.
This version of the agreement problem was addressed so far only in
\emph{synchronous} networks, considering both fault-free environments
\cite{DBLP:conf/podc/AugustineMP18} and crash faults \cite{DBLP:conf/icdcn/KumarM23}.
As implicit agreement is directly reducible to the aforementioned
implicit leader election problem, we obtain an \emph{asynchronous} implicit
agreement algorithm that withstands
$n (1 / 2 - \epsilon)$
faults in the initial failures model.

\Paragraph{Online Gaming Platforms.}
This long-lived problem addresses an online gaming platform (see, e.g.,
\cite{DBLP:conf/stoc/EmekKW16}) in a ``one-versus-one'' game mode (two players compete
against each other directly like in chess), where players arrive over time and
should be matched with opponents.
The reduction to TF works as follows:
Each player is represented by a single token, which is injected into a
relevant node, e.g., a nearby server, and the TF algorithm is invoked with
$\sigma = 2$.
When a team is formed, the trace tree is used to establish a connection
between the servers of the matched players.

\Paragraph{Distributed Matching with Failures.}
In this one-shot problem, the nodes are equipped with unique IDs and an (even
size) subset of the (non-faulty) nodes is marked as part of the input.
The goal is to form a perfect matching in the subgraph induced by the marked
nodes so that each node knows the ID of the node it is matched to.
To reduce this problem to TF, intuitively, each marked node generates a token
and the TF algorithm is invoked with
$\sigma = 2$;
the trace tree is employed to allow the matched nodes to exchange IDs.

\Paragraph{Robot Team Gathering.}
The one-shot \emph{gathering} problem requires robots, initially positioned
arbitrarily, to meet at a single point within a finite time, where the meeting
point is not fixed initially
\cite{DBLP:journals/siamcomp/AgmonP06,
DBLP:conf/wdag/LunaUVY20,
DBLP:journals/eatcs/MhamdiGMT19,
DBLP:journals/tcs/FlocchiniPSW05,
DBLP:journals/siamcomp/SuzukiY99}.
We propose and solve a new problem called \textit{robot team gathering} which
is a generalization of the gathering problem.
Robots are required to form teams of a given (known) size, and each team must
meet at a single point within a finite time, where the meeting points are not
fixed initially.
In the reduction of robot team gathering (in a complete graph) to TF, each
robot is represented as a token, and the TF algorithm is used to form
teams.\footnote{%
It is not difficult to have the robots (each playing a token) to simulate the
algorithms of the nodes, using whiteboards
\cite{DBLP:conf/ciac/DobrevKSS06,
DBLP:journals/networks/FraigniaudGKP06}
at the nodes.}

\Paragraph{Distributed Trigger Counting.}
Consider a network of interconnected devices where the devices count triggers
from an external source.
In the one-shot \emph{distributed trigger counting (DTC)} problem, the
algorithm is required to raise an alarm when the total number of triggers
counted by all the devices reaches a predefined threshold
\cite{DBLP:journals/wicomm/ChangT16,
DBLP:conf/podc/EmekK10,
DBLP:journals/ict-express/KimFP24}.
Intuitively, to solve DTC, each trigger generates a token, and the TF
algorithm is invoked after setting the team size parameter to the threshold
value.

\Paragraph{Team Formation with Associated Information.}
The application is for the case where a token $\tau$ arrives with some
associated information $\AS(\tau)$.
It is required that $\AS(\tau)$ be available at the node where the team, that
includes $\tau$, is formed at the time of the formation.
For example, $\AS(\tau)$ may be a piece of a secret, and $\sigma$ pieces are
needed and are enough to recover the secret.
We further assume that the information associated with any two tokens that
exist in the network simultaneously is different.
This means that the size of $\AS(\tau)$ may be large, so only a constant number
of such pieces may fit in one message (otherwise, the implementation may be
straightforward even if the message size is $O(\log n)$).
The reduction to (the standard version of) TF is slightly more technical, see
\Sect{}~\ref{sec:applications-additional-information}.

\subsection{Additional Related Work}
\label{sec:related-work}
The problem of reaching consensus in a complete graph with initial faults was introduced in the seminal paper of \cite{DBLP:journals/jacm/FischerLP85} since they showed that reaching consensus was impossible in complete asynchronous networks when faults could occur at any time. The same algorithm could be used for leader election. Its message complexity in the CONGEST model was $O(n^3)$ if $O(n)$ nodes could fail. This was improved to $O(n^2)$ in 
\cite{DBLP:conf/stacs/Bar-YehudaKWZ87} that also addressed general graphs and \cite{DBLP:journals/tse/ItaiKWZ90} that solved leader election (and thus also consensus) in a complete graph where the algorithm is given $k < n/2$ that is the maximum possible number of initial faults. The message complexity was $O(n \log n + n k)$) which is $O(n^2)$ in complete networks. For general graphs, a further improvement was obtained 
in \cite{DBLP:conf/podc/AfekS87}
by combining their counting algorithm with an algorithm for election (with no termination detection) in the presence of initial faults, such as the algorithm of 
 \cite{DBLP:journals/ipl/Kutten88}. Their message complexity was $O(|E|+n \log^5{n})$. The above mentioned counting was improved in \cite{DBLP:journals/jacm/AfekAPS96}, implying a message complexity of 
 $O( |E| + n \log^3{n} )$ for deterministic leader election and consensus under initial faults in general graphs. This means $O(n^2)$ in complete graphs.

Randomization was used to enable consensus in a complete graph in case nodes could fail at any time. Various algorithms were introduced, with various promises on resilience (e.g., also addressing Byzantine faults, varying the number of faults), time complexity, probabilistic assumptions, etc. 
\cite{DBLP:conf/podc/Ben-Or83,DBLP:conf/podc/BrachaT83,DBLP:conf/podc/Bracha84,DBLP:conf/focs/Rabin83}. They send $\Omega(n^2)$ messages.  

Election and consensus in \emph{non-faulty} complete graphs have been studied extensively; see, e.g.,  \cite{DBLP:conf/podc/KorachMZ84,DBLP:journals/siamcomp/AfekG91,humblet1984selecting,DBLP:journals/iandc/AfekM94,DBLP:conf/wdag/KuttenMPP21,DBLP:conf/podc/Kutten0T023}. The message complexity in most of them is $\Omega(n \log n)$, with some being $O(n)$, e.g., in the randomized synchronous case. Surprisingly, $o(n)$ (specifically, $O(\log n \sqrt{n}$) was eventually shown for synchronous fault-free networks for the implicit version of the problem (that is, not every node needs to know the output, but no two outputs contradict) \cite{DBLP:journals/tcs/KuttenPP0T15}.
This has been generalized to algorithms that can withstand crash faults \cite{DBLP:conf/icdcn/KumarM23} or even byzantine faults (assuming authentication) \cite{DBLP:conf/icdcn/KumarM24} 
and sublinear message complexity of agreement was also studied \cite{DBLP:conf/podc/AugustineMP18}
but all these was still done in \emph{synchronous} networks.
 The birthday paradox \cite{mckinney1966generalized}
 is used in the above papers;
 it is used here too, for building a more general structure that supports more general results (such as asynchrony and initial failures).
 
   The \emph{explicit} synchronous case of consensus with crash faults was addressed in \cite{DBLP:conf/soda/GilbertK10}, using $\Theta(n)$ messages and $\Theta( \log{n} )$ time. 
 
For the \emph{asynchronous} fault-free case, too, $\Omega(n)$ was considered to be optimal (see, e.g. \cite{DBLP:conf/wdag/KuttenMPP21}). It may indeed be for the explicit case, but recall that one of the results in the current paper is an $o(n)$ message result for the \emph{asynchronous} case, even for networks with (generalized) initial faults. The time complexity in
\cite{kutten2020singularly} for the \emph{non-faulty} synchronous case (and in
\cite{DBLP:conf/wdag/KuttenMPP21} for the \emph{non-faulty} asynchronous case) is $O(\log^2 n)$, improved here to $O(\log n)$ with at least as good as those of 
\cite{DBLP:conf/wdag/KuttenMPP21,kutten2020singularly} (the message complexity as good as that of the previous synchronous solution and potentially better than the previous asynchronous one).\footnote{The previous result for the asynchronous case was for general graphs.}

Various papers addressed 
problems that are variations of the application we call here distributed trigger counting (DTC).
(Some such variations are called threshold detection, threshold sensing, and controller.)
The one-shot case is the one where the number of triggers (called ``events'' in some other papers) passes some threshold \cite{DBLP:journals/jacm/AfekAPS96,DBLP:conf/podc/KormanK07,DBLP:conf/ics/GargGS06,DBLP:journals/mst/ChakaravarthyCGS12,DBLP:conf/icdcs/HuangGJT07,DBLP:journals/dc/EmekK11}.
The message complexity obtained was not always analyzed (in the more practical papers), but when analyzed, it was at least $O(\log n \log \sigma)$ even in the case of complete graphs \cite{DBLP:conf/ics/GargGS06}.
In the latter case a message could be sent directly from a site (node) to a leader (a coordinator node) that is known in advance. (That model allows sending messages between sites, but the algorithms do not use that.) 

The ongoing case is treated, e.g., in
\cite{DBLP:conf/sigmod/KeralapuraCR06}. 
The message complexity per counted event is about the same as in the one-shot case, and some errors in the counting are allowed.

The notion of forgetful algorithms extends the notion of 
\textit{memoryless} algorithms first defined for deterministic shared-memory mutual exclusion algorithms in \cite{DBLP:conf/podc/StyerP89}. A memoryless mutual exclusion algorithm is an algorithm in which, when all the processes are in their remainder sections, the registers' values are the same as their initial values. 
 This means that a process that tried to enter its critical section did not use any information about its previous attempts (such as the fact that it has entered its critical section five times so far).
Almost all known mutual exclusion algorithms are memoryless \cite{gadi-book}, including Lamport’s famous Bakery algorithm \cite{DBLP:journals/cacm/Lamport74a}.

\subsection{Paper's Outline}

\Sect{}~\ref{sec:technical-challenges} presents the main technical challenges arising in the study of the TF problem.
\Sect{}~\ref{sec:preliminaries} introduces the model, initial failures generalization, complexity measures, and basic definitions.
\Sect{}~\ref{sec:algorithm} presents our TF algorithm, first through a high-level overview and then through a detailed description of its two-layer structure and the interface between the layers.
\Sect{}~\ref{sec:analysis} contains the main analysis of the algorithm, establishing safety and liveness, and analyzing the message load.
\Sect{}~\ref{sec:channel-layer-algorithm} provides the technical details of the algorithm's lower layer and its analysis.
\Sect{}~\ref{sec:reaction-time} presents an analysis of the algorithm's reaction time.

\Sect{}~\ref{sec:lower-bound} establishes our lower bound, stated in \Thm{}~\ref{thm:lowerbound-informal}.
\Sect{}~\ref{sec:trace-tree-mechanism-details} provides further details on the trace-tree mechanism mentioned in \Sect{}~\ref{sec:intro:applications}.
\Sect{}~\ref{sec:applications-additional-information} presents additional technical details, omitted from \Sect{}~\ref{sec:intro:applications}, concerning LE as an application of our TF algorithm, as well as additional applications (beyond LE) of the TF problem.

\Sect{}~\ref{sec:discussion} concludes with a discussion of the usefulness of the TF problem, the results presented in this paper, and possible future directions.
\Sect{}~\ref{sec:pseudocode} contains the pseudocode of the algorithm, together with accompanying tables of variables and messages.

\section{Technical Challenges}
\label{sec:technical-challenges}
Before we move on to presenting our solution for the TF problem, the reader
may wonder about the following simple scheme:
Elect a leader $v$ (once) and following that, transport all injected tokens to
$v$ so that $v$ handles all team formation operations.
We argue that this scheme suffers from various severe shortcomings:
(1)
under our failure model (formally presented in
\Sect{}~\ref{sec:preliminaries:model}), the leader $v$ may become faulty at
any quiescent time, causing all subsequently injected tokens to get lost;
(2)
since our model does not include ``spontaneous wake-ups'', it is not clear
which nodes participate in the leader election stage;
(3)
leader election is impossible if we allow a
$(1 - \epsilon)$-fraction
of the nodes to be faulty;
and
(4)
the scheme does not satisfy the forgetful feature (as defined in
\Sect{}~\ref{sec:intro:extra-features}).

The aforementioned shortcomings highlight some of the technical challenges we
had to overcome when devising our TF algorithm.
To discuss those, assume for now that the team size parameter $\sigma$ is a
large constant and consider a token $\tau$ injected into some node $v$.
The first task faced by $v$ is to look for other (nodes that hold) tokens ---
if there are at least
$\sigma - 1$
of those, then a team should be formed.
Doing so in a communication efficient manner is a challenge, especially under
an asynchronous scheduler when a large fraction of the nodes may be faulty (and
thus, never respond to incoming messages).
We tackle this challenge by constructing a ``probabilistic quorum system''
with quorums of size
$O (\sqrt{n \log n})$
(using the ``factory coin tosses''), thus ensuring that any two token holding
nodes have at least one non-faulty node in the intersection of their quorums.

Once node $v$ has identified other token holding nodes, the next task is to
determine who transports the tokens to whom, raising a symmetry breaking
challenge:
a naive approach may cause $v$ to transport its token to $v'$ while $v'$
transports its own token to $v''$ and so on.
To prevent long token transportation chains, we adopt a ``star-shaped''
transportation pattern:
the token holding nodes run consecutive phases and in each phase, assume a
center or an arm role at random;
each arm node $v_{a}$ then tries to transport its token to a center node
$v_{c}$ that is ``ready to accept'' $v_{a}$'s token.
For this idea to succeed, one has to lower-bound the probability for an arm
node to identify an appropriate center node in each phase.
However, as we cannot synchronize between the phases of different nodes,
ensuring such a probability lower bound turns out to be eminently challenging.
We resolve this challenge by incorporating a ``selective waiting mechanism''
inspired by a technique introduced in \cite{DBLP:conf/stoc/AwerbuchCS94} (for the
entirely different task of constructing a maximal independent set in a general
graph);
refer to \Sect{}~\ref{sec:algorithm} for details.

\section{Preliminaries}
\label{sec:preliminaries}

\subsection{The Computational Model}
\label{sec:preliminaries:model}
Fix some randomized TF algorithm \Alg{}.
Recall our assumption that the adversary assigns a finite delay to each
message of \Alg{}.
The only constraint we impose on the adversary in this regard is that the
messages sent from a node
$v \in V$
to a node
$v' \in V$
are delivered in FIFO order.
For the sake of the runtime analysis (more on that later), we scale the time
axis so that the message delays are up-bounded by $1$ time unit;
we emphasize that the nodes themselves have no notion of time.

\Paragraph{Event Driven Executions.}
The execution of \Alg{} advances through the continuous time axis
$\Reals_{> 0}$
in an event driven fashion so that a node
$v \in V$
is activated at time
$t > 0$
if (and only if) one of the following two \emph{activation events}
occurs at time $t$:
(1)
a message is delivered to $v$;
or
(2)
one or more tokens are injected into $v$.
When activated, the actions of $v$ under \Alg{} include
(I)
reading the new incoming message (if any);
(II)
local computation that may involve (private) random coin tosses;
(III)
any number of team formation operations;
and
(IV)
sending messages to (a subset of) $v$'s neighbors.\footnote{%
We allow $v$ to send multiple messages to the same neighbor (in practice, those
can be piggybacked).}
It is assumed that actions (I)--(IV) are performed atomically (i.e., they
take zero time).
We emphasize that the nodes never act unless activation events (1) or (2)
occur (in particular, in the scope of the TF problem, there are no
``spontaneous wake-ups'').

Assume without loss of generality that the activation events are isolated so
that no two of them occur at the same time
(indeed, simultaneous activations can be handled locally by processing them sequentially in an arbitrary order).
The times at which these activation events occur are referred to as
\emph{activation times}.
Time
$t > 0$
is regarded as \emph{quiescent} if
(1)
$t$ is not an activation time;
(2)
there are no messages in transit at time $t$;
and
(3)
the system contains $0$ tokens at time $t$.

For a time
$t > 0$,
let
$t- = t - d t$
and
$t+ = t + d t$
be the times immediately before and immediately after $t$, respectively, where
$d t$
is chosen so that the time interval
$[t-, t+]$
contains at most one activation event.
Given an object $\Object$, let $\Object^{t}$ denote the state of $\Object$ at
time $t$, adhering to the convention that $\Object^{t} = \Object^{t-}$
for all
$t > 0$;\footnote{%
Objects in this regard include local variables as well as global objects
defined for the analysis.}
in particular, if $t$ is an activation time and the state of $\Object$ changes
at time $t$, then $\Object^{t}$ is taken to be the state before the change.
Notice that our convention implies that instantaneous changes to $\Object$,
that are done as part of the local computation and do not persist beyond time
$t$, are not reflected in $\Object^{t}$.

\Paragraph{Fault Tolerance and Adversarial Policies.}
We adopt a generalization of the \emph{initial failures} model
\cite{DBLP:journals/cacm/Dijkstra65,
DBLP:journals/jacm/FischerLP85,
DBLP:journals/ijpp/TaubenfeldKM89}.
Under the initial failures model, the adversary selects a node subset
$F \subset V$
at time $0$ (i.e., right before the execution commences) and turns the status
of all the nodes in $F$ to \emph{faulty}, thus preventing them from
participating in the execution.
In particular, faulty nodes do not send messages and in the context of the TF
problem, no tokens are injected into faulty nodes.
A message $\mu$ delivered to a faulty node is lost (i.e., it is ignored and
does not trigger an activation event).
If tokens are transported over a lost message $\mu$, then these tokens are
assumed to remain in ``limbo'' forever, which may prevent the algorithm from
satisfying the liveness condition;
it is the responsibility of the algorithm designer to avoid such scenarios.

In the (new) generalized version of the initial failures model, considered in
the current paper, the adversary may \emph{toggle} the status of a node
$v \in F$
from faulty to non-faulty and vice versa, so that $v$ may participate in the
execution in certain time intervals and not participate in others.
The decisions of the adversary in this regard are subject to the following
constraint:
the faulty/non-faulty status of $v$ may be toggled only at quiescent times.
For a precise description of the failure model, if $v$ becomes faulty at time
$t \geq 0$
and remains faulty until time
$t' > t$,
at which it becomes non-faulty, then the memory image of $v$ at time $t'$ is
assumed to be identical to that of time $t$.
To avoid confusion, we subsequently refer to the nodes in $F$ as
\emph{fragile} (regardless of their faulty/non-faulty status at a specific
time of the execution), emphasizing that the non-fragile nodes remain non-faulty throughout the
execution.

Notice that due to the asynchronous nature of the system, it may be impossible
for a non-faulty node to distinguish its faulty neighbors from the ``slow
responding'' ones.
Moreover, the fragile nodes do not know that they are fragile and when a
fragile node becomes non-faulty, it runs the same algorithm as the non-fragile
nodes.
The generalization of the classic initial failures model is different from the crash failures model, where a node may fail at any time.
As pointed out in \cite{DBLP:journals/jacm/FischerLP85}, considering initial failures is known to open the gate for asynchronous algorithms that are impossible under crash failures.
Note that the TF problem is impossible under crash failures, since
safety may already be violated by a single such failure if a node holding
(or about to receive)
tokens crashes.

We emphasize that quiescent-time toggles do not, by themselves, imply a restart from an initial configuration.
In general, an algorithm may reach a quiescent time with arbitrary node states.
In our case, quiescent times resemble ``fresh starts'' only because our algorithm satisfies the forgetful feature from \Sect{}~\ref{sec:intro:extra-features};
this is a property of our algorithm, not of the model itself.

With this failure model, the adversary has the following roles:
(1)
At time $0$, the adversary determines the subset
$F \subset V$
of fragile nodes and (without loss of generality) turns their status to
faulty.
(2)
For time
$t = 0$
and for every activation time
$t > 0$,
at time $t+$, the adversary determines the next activation event (if any)
including its type (message delivery or token injection), location, and time
$t' > t$.
(3)
For every quiescent time
$t > 0$
and for every node
$v \in F$,
the adversary determines (at time $t$) if the faulty/non-faulty status of $v$
is toggled at time $t$ (notice that this is \emph{not} an activation event).
We emphasize that tokens can be injected into node $v$ at time $t$ only if $v$ is non-faulty at time $t$;
tokens are never injected into faulty nodes.

The adversary is \emph{adaptive}, namely, its decisions at time $t$ may be based on the nodes' actions before time $t$
(however, without knowing the nodes' future coin tosses;
in the terminology of \cite{DBLP:conf/atal/DatarNA23}, our adversary is \emph{weakly adaptive}, in contrast to their \emph{strongly adaptive} adversary, which can exploit for each node the random bits used in its next event).
Consequently, the execution of \Alg{} can be represented as an extensive form
game with incomplete information \cite{Hart1992games}, where the adversary is
the only strategic player;
that is, a rooted tree whose (internal) vertices are partitioned into a set of
adversarial moves, in which the adversary makes its decisions, and a set of
chance moves, each corresponding to the coin tosses of an activated node in
some activation event.
Terminology-wise, a mechanism that makes the decision in each adversarial move
of the aforementioned extensive form game is referred to as an
\emph{adversarial policy}.
Notice that by fixing an adversarial policy, \Alg{}'s execution becomes a
random variable, fully determined by the nodes' coin tosses.

\Paragraph{Performance Measures.}
The quality of \Alg{} is evaluated by means of two performance measures.
First and foremost, we wish to bound the number of messages that \Alg{} sends
per token in expectation and whp.
Care is needed in this regard since the adaptive nature of the adversary
implies that the number of tokens injected into the system may depend, by
itself, on the nodes' coin tosses.
Thus, the number of messages per token is, in general, the ratio of two random
variables that may exhibit complex dependencies, making it difficult to reason
about.
To resolve this difficulty, we formulate our main performance measure through
the notion of \emph{message load}, defined as follows.

For
$\ell \in \Integers_{> 0}$,
let $\A_{\ell}$ be the family of adversarial policies that inject at most
$\ell$ tokens into the system (throughout the execution).
We say that \Alg{}'s \emph{message load} is $M$ in expectation (resp.,
whp) if $M$ is the smallest real such that
for every
$\ell \in \Integers_{> 0}$
and for every adversarial policy in $\A_{\ell}$, it is guaranteed that \Alg{}
sends at most
$M \cdot \ell$
messages in expectation (resp., whp).\footnote{%
We note that the message load provides a bound on the number of messages sent
(per token) in ``finite prefixes'' of the execution under \emph{any}
adversarial policy, including those that do not belong to
$\bigcup_{\ell > 0} \A_{\ell}$.}

We are also interested in the time it takes for \Alg{} to form a team once
sufficiently many tokens are present.
To this end, we say that \Alg{}'s \emph{reaction time} is $R$ if $R$ is the
smallest real such that for every time
$t > 0$,
if the system contains at least $\sigma$ tokens at time $t$, then a team is
formed by time
$t + R$
whp.

\subsection{The Primary-Utility Graph}
\label{sec:preliminaries:primary-utility}
The algorithms developed in this paper are designed so that each node
$v \in V$
simulates a virtual \emph{primary node} $p(v)$ and a virtual \emph{utility
node} $u(v)$;
let
$P = \{ p(v) : v \in V \}$
and
$U = \{ u(v) : v \in V \}$
denote the sets of primary and utility nodes, respectively.
Although the primary node $p(v)$ and utility node $u(v)$ are simulated by the
same ``physical'' node
$v \in V$,
it is convenient to address them as standalone computational entities that
participate in the execution by receiving and sending messages independently
of $v$.
To this end, tokens injected into $v$ are regarded as injected into the
primary node $p(v)$.

In terms of the failure model, the primary node $p(v)$ and utility node $u(v)$
are regarded as \emph{fragile} if the corresponding
``physical'' node
$v \in V$
is fragile.
Let
$F_{P} = \{ p(v) \in P : v \in F \}$
and
$F_{U} = \{ u(v) \in U : v \in F \}$
denote the sets of fragile primary and utility nodes, respectively.
The faulty/non-faulty status of the fragile primary node
$p(v) \in F_{P}$
and fragile utility node
$u(v) \in F_{U}$
is assumed to be inherited from that of the corresponding fragile ``physical''
node
$v \in F$.\footnote{%
This assumption is made only for the sake of simplicity;
as far as our algorithms and analyses go, the faulty/non-faulty status of
$p(v)$ and $u(v)$ can be decoupled.}

A key design feature of our algorithms is that the entire message exchange
is confined to an overlay bipartite graph
$G_{P U} = (P, U, E_{P U})$,
referred to as the \emph{primary-utility graph}.
The edge set $E_{P U}$ of the primary-utility graph is constructed randomly by
selecting each node pair
$(p, u) \in P \times U$
to be included in $E_{P U}$ independently with probability
$c \cdot \sqrt{\frac{\log n}{n}}$,
where
$c > 0$
is a constant to be derived from the analysis;\footnote{%
With a slight abuse of notation, we often use $p$ (resp., $u$) as a
placeholder for a general primary (resp., utility) node.}
this random selection is made by each primary node $p$ upon its first
activation event and remains fixed throughout the execution.
Notation-wise, for a primary node
$p \in P$
and a utility node
$u \in U$,
let
$U(p) = \{ u' \in U : (p, u') \in E_{P U} \}$
and
$P(u) = \{ p' \in P : (p', u) \in E_{P U} \}$.
The key probabilistic properties of the primary-utility graph are cast in the
following lemma;
in the remainder of this paper, we condition on the event that these
properties are satisfied.

\begin{lemma}
\label{lem:pu-graph}
The primary-utility graph
$G_{P U} = (P, U, E_{P U})$
satisfies the following two properties whp:
\\
(1)
$( U(p) \cap U(p') ) - F_{U} \neq \emptyset$
for every
$p, p' \in P$;
and
\\
(2)
$|U(p)|, |P(u)| \leq O (\sqrt{n \log n})$
for every
$p \in P$
and
$u \in U$.
\end{lemma}

\begin{proof}
To see that (1) holds, consider two distinct primary nodes
$p, p' \in P$
and let $A_{u}$ be the event that the non-fragile utility node
$u \in U - F_{U}$
is included in
$U(p) \cap U(p')$.
Since
$\Pr (A_{u}) = \frac{c^{2} \log n}{n}$
for every
$u \in U - F_{U}$
and since the events $A_{u}$,
$u \in U - F_{U}$,
are mutually independent, it follows that
\begin{align*}
\Pr \left( \bigvee_{u \in U - F_{U}} A_{u} \right)
\, = \, &
1 - \left( 1 - \frac{c^{2} \log n}{n} \right)^{|U - F_{U}|}
\\
\geq \, &
1 - \left( 1 - \frac{c^{2} \log n}{n} \right)^{\epsilon n}
\\
\geq \, &
1 - e^{-\epsilon c^{2} \log n}
\, = \,
1 - n^{-c'}
\end{align*}
for a constant $c'$ that can be made arbitrarily large by increasing the
constant
$c = c(\epsilon)$,
where the second transition holds as there are at least
$\epsilon n$
non-fragile utility nodes.
The assertion follows by applying the union bound over all
$\binom{n}{2}$
choices of the
primary nodes
$p, p' \in P$.

To see that (2) holds, notice that
$\Ex ( |U(p)| ) = c \cdot \sqrt{n \log n}$
(resp.,
$\Ex ( |P(u)| ) = c \cdot \sqrt{n \log n}$)
for every primary node
$p \in P$
(resp., utility node
$u \in U$).
By a standard application of Chernoff's bound, we conclude that
$|U(p)| \leq O (\sqrt{n \log n})$
(resp.,
$|P(u)| \leq O (\sqrt{n \log n})$)
whp, thus completing the proof by applying the union bound over all $n$
choices of the primary node
$p \in P$
(resp., utility node
$u \in U$).
\end{proof}

\section{The Algorithm}
\label{sec:algorithm}
In this section, we present our TF algorithm, referred to as \AlgTF{}.
We begin in \Sect{}~\ref{sec:high-level-overview} with a high-level overview that explains the main ideas behind \AlgTF{}, and
then proceed to a more detailed and formal presentation in
\Sect{}~\ref{sec:detailed-description-algorithm}.

\subsection{High-Level Overview of the Algorithm}
\label{sec:high-level-overview}

Recall that each physical node simulates a primary node and a utility node.
A primary node is called \emph{busy} if it currently holds one or more tokens.
Algorithm \AlgTF{} repeatedly tries to merge the token bundles of busy primary
nodes so that eventually some busy primary node accumulates at least
$\sigma$
tokens and can form team(s).
(This is in contrast to the fixed-coordinator approach discussed in
\Sect{}~\ref{sec:technical-challenges}, and avoids its shortcomings.)

To support such local mergers, the algorithm relies on the random bipartite
primary-utility graph introduced in
\Sect{}~\ref{sec:preliminaries:primary-utility}.
This graph is constructed by having each primary node independently and
uniformly choose
$O(\sqrt{n \log n})$
utility nodes to be its neighbors.
The key property of this graph, guaranteed by
\Lem{}~\ref{lem:pu-graph} via a birthday-paradox type effect, is that, whp,
every two primary nodes share at least one non-fragile utility neighbor.
Thus, whenever two primary nodes are simultaneously busy, they are likely to
have a shared utility neighbor through which they can interact, while each primary node has only a small number of utility neighbors.

Conceptually, \AlgTF{} is organized into two layers.
The lower, \emph{channel layer}, maintained by the utility nodes, is
responsible only for communication among busy primary nodes.
Its job is to notice when two primary nodes are simultaneously busy and,
through a shared utility node, offer them a short-lived communication link,
called a \emph{channel}; when such a link exists, we say that the two primary
nodes are \emph{paired}.
One may think of a channel as a temporary conversation opened between two busy
primary nodes and automatically closed once at least one of them is no longer busy.
The channel layer is used only to provide these temporary pairwise contacts in a reliable way;
it does not decide where tokens should move.

The actual token-gathering logic is implemented in the upper,
\emph{principal layer}, maintained by the primary nodes and using the channels
provided by the channel layer.
There, each busy primary node proceeds through consecutive \emph{phases}.
At the beginning of each phase, the primary node flips a coin and acts either
as a \emph{center} or as an \emph{arm}.
Informally, a center tries to collect tokens from arms paired with it,
whereas an arm gives up all of its currently held tokens when requested to do
so by a center.
This induces a star-shaped transport pattern: tokens are meant to move from
arms into centers, rather than through long chains of intermediate primary
nodes.
Consequently, once a node gives up its tokens in a phase, it becomes non-busy
until new tokens are injected into it.

Each phase follows a simple \emph{request-response} discipline:
at the beginning of the phase, a busy primary node sends a request to each primary node paired with it, over their current channel, 
and then waits for a response on each of these channels before moving to the next phase;
the request sent by a center is a request for tokens, whereas the request sent by an arm is only a synchronization request.
(The one exception to this discipline is when an arm receives a token request from a
center: the arm then sends all of its tokens, becomes non-busy, and its current
phase terminates abruptly rather than waiting for the remaining responses.)
The request-response mechanism is needed because, without it, under an asynchronous scheduler,
a node could race through many phases while an earlier request of its peer is still in transit, thereby gaining an ``unfair advantage'' in the local symmetry breaking.

The request-response discipline by itself is still not enough.
Even if every phase waits for all its responses, an adversary may still force many futile local interactions by causing one primary node to remain stuck waiting,
while a primary node paired with it repeatedly completes phases.
To prevent this, the principal layer uses a \emph{selective waiting
mechanism}.
Roughly speaking, when a center receives a synchronization request from an arm, it delays its response until its next phase.
This effectively ``freezes'' the arm temporarily: the arm cannot race far ahead, and centers get a fair chance to claim its tokens.
At the same time, the delay lasts for only one phase transition, and the
waiting is asymmetric: arms may wait for centers in this manner, but centers do
not wait for arms.
As a result, this mechanism does not create long waits or deadlocks.

The next section formalizes this scheme and specifies the exact interface
between the two layers.

\subsection{Detailed Description of the Algorithm}
\label{sec:detailed-description-algorithm}

The algorithm \AlgTF{} attempts to (locally) gather tokens, that may be distributed
over multiple primary nodes, into a single primary node and perform team
formation(s) if the number of gathered tokens is large enough.

Under \AlgTF{}, each primary node
$p \in P$
maintains the local variable
$p.\varTokens \in \Integers_{\geq 0}$
that counts the number of tokens held by $p$.
Recalling that $p.\varTokens^{t}$ is the value of $p.\varTokens$ at time
$t > 0$,
let
$\Busy^{t} = \{ p \in P : p.\varTokens^{t} > 0\}$,
referring to the nodes in $\Busy^{t}$ as \emph{busy}.
The algorithm is designed so that a utility node
$u \in U$
may hold tokens only instantaneously, when the tokens are transported, through
$u$, between two primary nodes in $P(u)$.

Fix some
$p \in P$
and
$t > 0$
and suppose that
$p \in \Busy^{t}$
and that $t$ is an activation time of $p$.
If
$p.\varTokens^{t} < \sigma$,
then $p$ may decide, at time $t$, to transport the tokens it holds by sending
a message $\mu$, that carries these tokens, to another primary node
$p' \in P$,
through a utility node
$u \in U(p) \cap U(p')$;
the algorithm is designed so that this may happen only if $p'$ is busy at
time $t$ and remains busy (at least) until message $\mu$ is delivered to $p'$.
If
$p.\varTokens^{t} \geq \sigma$,
then $p$ may decide, at time $t$, to eliminate some of the tokens it holds, in
which case, $p$ performs
$\lfloor p.\varTokens^{t} / \sigma \rfloor$
team formations (i.e., as many team formations as possible).
We emphasize that $p$ will never transport tokens as long as
$p.\varTokens \geq \sigma$
and that if $p$ does transport the
($p.\varTokens < \sigma$)
tokens it holds, then all $p.\varTokens$ tokens are transported together and
$p$ becomes non-busy.

If $p$ performs team formation(s) at time $t$ and
$p.\varTokens^{t} \bmod \sigma = k > 0$,
then $p$ is left with $k$ ``remainder tokens''.
To simplify the presentation, we treat these $k$ ``remainder tokens'' as if
they are injected into $p$, as fresh tokens, soon after time $t$ (and strictly
before the next ``original'' activation event);
refer to the injections of such ``remainder tokens'' as \emph{fake
injections}.
In particular, we assume that if $p$ performs team formation(s) at time $t$,
then
$p.\varTokens^{t+} = 0$,
hence
$p \notin \Busy^{t+}$.
Put differently, as long as
$p \in \Busy$,
the variable $p.\varTokens$ is a non-decreasing function of $t$.\footnote{%
Throughout, we omit the superscript $t$ from $\Object^{t}$ when we wish to
address the dynamic nature of the object $\Object$ whose state may vary over
time.}

\Paragraph{Two-Layer Structure.}
To simplify the algorithm's presentation, we divide it, logically, into two
layers:
a lower \emph{channel layer} and an upper \emph{principal layer}.
Rather than sending messages directly, the principal layer in (busy) primary
nodes uses the service of the channel layer for communication.
This service, called \emph{channels}, bears similarities to virtual circuits
such as TCP connections.
As in the case of TCP connections, a channel $\chi$ between two primary nodes
$p, p' \in P$
can be released, and sometimes later, a new channel $\chi'$ between $p$ and
$p'$ may be created.
An important feature of the channels is that a message sent from $p$ to $p'$
as part of $\chi$ does not arrive as part of $\chi'$;
this is formalized, together with several other important assurances of the
channel layer, in \Grnt{}\
\ref{grnt:algorithm:interface:operational-implies-busy}--%
\ref{grnt:algorithm:interface:message-complexity} 
below.

A reader who is familiar with the nuts and bolts of virtual circuits knows
that the task of ensuring such properties is messy, sometimes non-trivial, but
nonetheless possible.
This is the reason we defer the description of the implementation of the
channel layer to \Sect{}~\ref{sec:channel-layer-algorithm}, whereas the
principal layer's implementation is presented in the current section.
For now, let us just say that in the channel layer, each primary node
$p \in P$
updates all the utility nodes in $U(p)$ whenever its status changes from busy
to non-busy or vice versa.
A (non-faulty) utility node
$u \in U$
chooses two primary nodes in $P(u)$ that reported they are busy, and creates a
channel between them.
This channel is released by $u$ when (and only when) $u$ hears that one of the
channel's primary nodes is no longer busy.
We now provide a more formal description of the channels and the interface
between the two layers.

\Paragraph{The Channels.}
The role of the channels is to enable (duplex) communication among (unordered)
pairs of busy primary nodes.
For
$p, p' \in P$,
each
$\{ p, p' \}$-channel
$\chi$ is associated with a utility node
$u \in U(p) \cap U(p')$,
referred to as the channel's \emph{mediator} that relays the messages that $p$
and $p'$ exchange with each other as part of $\chi$, referred to hereafter as
\emph{relayed messages}.
The algorithm is designed so that for every
$u \in U$
and
$p, p' \in P$,
at any given time, the system includes at most one channel mediated by $u$ and
zero or more
$\{ p, p' \}$-channels,
each mediated by its own utility node in
$U(p) \cap U(p')$.

To keep track of the channels, a primary node
$p \in P$
maintains the local variable
$p.\varMediators \subseteq U(p)$
that stores the mediators of the
$\{ p, \cdot \}$-channels;
a utility node
$u \in U$
maintains the local variable
$u.\varChannel \subseteq P(u)$
defined so that
$u.\varChannel = \{ p, p' \}$
if $u$ mediates a
$\{ p, p' \}$-channel,
and
$u.\varChannel = \emptyset$
otherwise.
For a
$\{ p, p' \}$-channel
$\chi$ mediated by $u$, we refer to the time
$t > 0$
at which $u$ sets
$u.\varChannel \gets \{ p, p' \}$
as the \emph{creation time} of $\chi$;
we refer to the earliest time
$\bar{t} > t$
such that
$u.\varChannel^{\bar{t}+} \neq \{ p, p' \}$
as the \emph{release time} of $\chi$.

We emphasize that channel $\chi$ is created (resp., released) once and if the
mediator $u$ creates (resp., releases) a
$\{ p, p' \}$-channel
$\chi'$
at time
$t' \neq t$
(resp.,
$\bar{t}' \neq \bar{t}$),
then $\chi'$ and $\chi$ are considered to be two different channels.
Notice that the placeholders $\chi$ and $\chi'$ are introduced for the sake of
the discussion and we do not assume that the primary nodes $p$ and $p'$ agree
on ``common names'' for the
$\{ p, p' \}$-channels.

Consider some primary node
$p \in P$
and utility node
$u \in U(p)$.
A key feature of the algorithm is that
if
$I \subset \Reals_{> 0}$
is a maximal time interval such that
$u \in p.\varMediators^{t}$
for all
$t \in I$,
then all relayed messages that $p$ sends to or receives from $u$ during $I$
belong to the same
$\{ p, p' \}$-channel
$\chi$ (mediated by $u$) for some primary
node
$p' \in P(u) - \{ p \}$;
we refer to $\chi$ and $p'$ as the \emph{$u$-channel} and \emph{$u$-peer},
respectively, of $p$ during $I$.
Notation-wise, let
$\Channel^{t}(p, u)$
and
$\Peer^{t}(p, u)$
be the operators that return the $u$-channel and $u$-peer, respectively, of
$p$ at time $t$ if
$u \in p.\varMediators^{t}$,
and $\bot$ otherwise.

The algorithm is designed so that if $\chi$ is a
$\{ p, p' \}$-channel
mediated by $u$, then the set
$\{ t \in \Reals_{> 0} : \Channel^{t}(p, u) = \chi \}$
is either empty or forms a single interval of the time axis.
This is in contrast to the set
$\{ t \in \Reals_{> 0} : \Peer^{t}(p, u) = p' \}$
that may form multiple intervals, each corresponding to a different
$\{ p, p' \}$-channel
mediated by $u$.

Notice that
$\Channel^{t}(p, u) = \chi \neq \bot$
does not imply that channel $\chi$ still exists at time $t$ as $\chi$ may
have been released by $u$ at time
$t - 1 \leq t' < t$
(which will be observed by $p$ at time
$t' < t'' \leq t' + 1$).
Moreover, while one may hope that the formula
$\Peer^{t}(p, u) = p'
\Longleftrightarrow
\Peer^{t}(p', u) = p$
is satisfied ``most of the time'', it cannot be satisfied all the time;
indeed, given a
$\{ p, p' \}$-channel
$\chi$,
since the variables $p.\varMediators$ and $p'.\varMediators$ are updated
asynchronously, we cannot expect the aforementioned formula to be satisfied
``shortly after'' (resp., ``shortly before'') $\chi$ is created (resp.,
released).
\Grnt{}~\ref{grnt:algorithm:interface:relayed-messages} ensures that
these inconsistencies do not introduce ``misunderstandings''.

\Paragraph{The Interface between the Layers.}
The channel layer is responsible for maintaining the channels by updating the
$\varMediators$ and $\varChannel$ variables and for handling the delivery of
relayed messages among peers.
The principal layer governs the token gathering process among the busy primary
nodes that communicate with their peers over the channels.
For a primary node
$p \in P$,
the main component in the interface between the two layers is the set
$\Channels^{t}(p) = \{ \Channel^{t}(p, u) : u \in p.\varMediators^{t} \}$
that captures $p$'s channels at time
$t > 0$.
It is assumed that node $p$'s channel layer maintains the set $\Channels(p)$
while hiding its implementation details;
node $p$'s principal layer then uses the channels in $\Channels(p)$ to
exchange relayed messages with the principal layer of $p$'s peers.

A primary node
$p \in P$
is regarded as \emph{operational} at time
$t > 0$
if
$\Channels^{t}(p) \neq \emptyset$;
let
$\OperationalNodes^{t}$
be the set of operational primary nodes at time $t$.
A
$\{ p, p' \}$-channel
$\chi$ is regarded as \emph{operational} at time
$t > 0$
if
$\chi \in \Channels^{t}(p) \cap \Channels^{t}(p')$;
let
$\OperationalChannels^{t}$
be the set of operational channels at time $t$.
These notions facilitate the formulation of the key assurances that the
channel layer provides to the principal layer, stated in \Grnt{}\
\ref{grnt:algorithm:interface:operational-implies-busy}--%
\ref{grnt:algorithm:interface:message-complexity}
(established in \Sect{}~\ref{sec:channel-layer-algorithm}).

\begin{guarantee}
\label{grnt:algorithm:interface:operational-implies-busy}
$\OperationalNodes^{t} \subseteq \Busy^{t}$
for every
$t > 0$.
\end{guarantee}

\begin{guarantee}
\label{grnt:algorithm:interface:operatioanl-channel-generation}
For every
$t_{0} > 0$,
if
$| \{ p \in P : p \in \Busy^{t} \text{ for all } t_{0} < t \leq t_{0} + 4 \} |
\geq
2$,
then there exists
$t_{0} < t \leq t_{0} + 4$
such that
$\OperationalChannels^{t+} \neq \emptyset$.
\end{guarantee}

\begin{guarantee}
\label{grnt:algorithm:interface:operational-channel-continues}
Consider a
$\{ p, p' \}$-channel
$\chi$ and time
$t > 0$,
and assume that
$\chi \in \OperationalChannels^{t}$.
Then,
$\chi \in \OperationalChannels^{t+}$
if and only if
$p ,p' \in \Busy^{t+}$.
Moreover, if
$p \notin \Busy^{t+}$,
then
$\chi \notin \Channels^{(t + 2)+}(p')$.
\end{guarantee}

\begin{guarantee}
\label{grnt:algorithm:interface:relayed-messages}
The following conditions are satisfied for every primary node
$p \in P$,
time
$t > 0$,
and
$\{ p, p' \}$-channel
$\chi \in \Channels^{t}(p)$:
(I)
If the principal layer of $p$ receives a relayed message $\mu$ over
$\chi$ at time $t$, then $\mu$ was sent over $\chi$ by the principal layer of
$p'$ at time
$t - 2 \leq t' < t$
(in other words, $\mu$ could not have originated from some past/future channel
of $p$ with the same mediator).
(II)
If the principal layer of $p$ sends a relayed message $\mu$ over $\chi$ at
time $t$, then either
(II.a)
$\mu$ is received over $\chi$ by the principal layer of $p'$ at time
$t < t' \leq t + 2$;
or
(II.b)
$\mu$ becomes irrelevant because $p'$ turned from busy to non-busy, leading to
the removal of $\chi$ from $\Channels(p')$.
\end{guarantee}

\begin{guarantee}
\label{grnt:algorithm:interface:degree-bound}
$| \Channels^{t}(p) | \leq O (\sqrt{n \log n})$
for every
$p \in P$
and
$t > 0$.
\end{guarantee}

\begin{guarantee}
\label{grnt:algorithm:interface:message-complexity}
The channel layer sends
$O (\sqrt{n \log n})$
messages per token.
\end{guarantee}

\subsubsection{The Principal Layer}
We now turn to describe the implementation of the principal layer, building on
the interface assured by \Grnt{}\
\ref{grnt:algorithm:interface:operational-implies-busy}--%
\ref{grnt:algorithm:interface:message-complexity}.
Pseudocode is provided in \Sect{}~\ref{sec:pseudocode-principal}, consisting of Algorithms \ref{alg:principal-layer}--\ref{alg:procedure-begin-new-phase}, along with Tables \ref{tab:messages-principal} and \ref{tab:variables-principal}.
Recall that this layer is implemented over the (busy) primary nodes
$p \in P$
that exchange relayed messages with their peers over the channels in
$\Channels(p)$ (maintained by the channel layer);
the utility nodes do not participate in this layer and are abstracted away
from its description.
The policy of the principal layer for a non-operational node
$p \in P$
is straightforward:
$p$ performs (as many as possible) team formation operations if
$p.\varTokens \geq \sigma$,
and does nothing otherwise.
The interesting part of the principal layer algorithm addresses the
operational nodes that ``compete'' with their peers over the right to serve as
the local gathering point for the tokens.
To this end, if
$p.\varTokens, p'.\varTokens < \sigma$
for some
$p, p' \in \OperationalNodes$
and
$\Channels(p) \cap \Channels(p') \neq \emptyset$,
then (local) symmetry breaking is needed to determine whether $p$ transports
its tokens to $p'$ or $p'$ transports its tokens to $p$ (or neither).

The principal layer resolves this symmetry breaking challenge by following
a ``star-shaped'' token transportation pattern.
Specifically, an operational node
$p \in \OperationalNodes$
runs consecutive \emph{phases} so that in each phase, $p$ assumes a
\emph{phase type} that can be either \emph{center} or \emph{arm}.
This phase type is determined by $p$ via an (unbiased) coin toss when the phase
begins and is stored in the local variable
$p.\varPhaseType \in \{ \Center, \Arm \}$
that $p$ maintains.
The crux of the (dynamic) classification into phase types is that a center
node attempts to collect the tokens from its arm peers so that a phase of a
center node $p$ is successful if one or more of $p$'s arm peers transport
their tokens to $p$;
a phase of an arm node $p$ is successful if $p$ transports its tokens to
exactly one of its center peers.

Each phase $\phi$ of node $p$ proceeds according to the following
\emph{request-response mechanism}:
When $\phi$ begins, node $p$ sends a \emph{request message} over each channel
in
$\Channels(p)$;
phase $\phi$ ends once $p$ has received a \emph{response message} for each
request message sent at the beginning of $\phi$.
(As discussed soon, $\phi$ may end abruptly if $p$ decides to transport
its tokens to one of its peers, in which case, $p$ becomes non-busy and hence,
also non-operational.)
To implement this request-response mechanism, $p$ maintains the local variable
$p.\varAwaitingResponse(\chi) \in \{ \True, \False \}$
for each
$\chi \in \Channels(p)$,
setting
$p.\varAwaitingResponse(\chi) \gets \True$
when a request message is sent over $\chi$ and resetting
$p.\varAwaitingResponse(\chi) \gets \False$
when a response message is received over $\chi$.

When phase $\phi$ ends, node $p$ checks if the inequality
$p.\varTokens \geq \sigma$
holds and if it does, performs (as many as possible) team formation
operations.
To simplify the presentation, we treat any tokens injected into $p$ during
phase $\phi$ as if they are injected when $\phi$ ends;
we emphasize that the injected tokens are accounted for when checking if the
inequality
$p.\varTokens \geq \sigma$
holds.\footnote{%
To adhere to the formal event driven model defined in
\Sect{}~\ref{sec:preliminaries:model}, that forbids concurrent activation
events, one can move the token injection event ``shortly after'' the end of
phase $\phi$ so that this event triggers the team formation operations (if
any) or the beginning of the subsequent phase (if $p$ is still operational).}

The type of the request messages depends on the phase type of their sender:
the requests of the center nodes are $\msgTokensPlease$ messages, whereas the
requests of the arm nodes are $\msgWaiting$ messages.
The reaction of the operational nodes to an incoming request, including the type
of the corresponding response message, depends on the phase type of the
receiver as well as on the type of the incoming request.
To this end, consider an operational node
$p \in \OperationalNodes$
that receives a request $\mu$ over channel
$\chi \in \Channels(p)$
during phase $\phi$.

Assume first that 
$\mu = \msgTokensPlease$.
If $\phi$ is a center phase, then $p$ reacts by sending a $\msgNoTransport$
response.
If $\phi$ is an arm phase, then $p$ reacts by sending a
$\msgTransport(p.\varTokens)$
response that transports the tokens held by $p$ (excluding those injected
during $\phi$) to the peer of $p$ at the other end of $\chi$.
Consequently, $p$ becomes non-busy and the channels in $\Channels(p)$ are
released (the channel layer takes care of that), making $p$ non-operational;
this is the one exception to the request-response mechanism, where phase
$\phi$ ends abruptly without waiting for all the responses (since
$\Channels(p)$ is emptied,
\Grnt{}~\ref{grnt:algorithm:interface:relayed-messages} ensures that
the ``missing responses'' are dropped and do not interfere with future
channels of $p$).

Now, assume that
$\mu = \msgWaiting$.
If $\phi$ is an arm phase, then $p$ reacts by sending a $\msgGoOn$ response.
If $\phi$ is a center phase, then $p$ delays its response until the current
phase $\phi$ ends and the phase type of the next phase $\phi'$ of $p$ is
revealed.
At this stage, if $\phi'$ is a center phase, then $p$ ``forgets about'' $\mu$
and continues as usual with $\phi'$;
in particular, $p$ sends $\msgTokensPlease$ requests over all channels in
$\Channels(p)$, including $\chi$ --- this $\msgTokensPlease$ request over
$\chi$ plays a key role because it is guaranteed to prompt the peer $p'$ of
$p$ on the other end of $\chi$ to transport its tokens if $p'$ did not do so
already.
If $\phi'$ is an arm phase, then $p$ first sends a $\msgGoOn$ message over
$\chi$ and then, continues as usual with $\phi'$ (sending $\msgWaiting$
requests over all channels in $\Channels(p)$).
To implement this policy, $p$ maintains the local variable
$p.\varDelayingResponse(\chi) \in \{ \True, \False \}$
for each channel
$\chi \in \Channels(p)$,
setting
$p.\varDelayingResponse(\chi) \gets \True$
when $p$ receives a $\msgWaiting$ message over $\chi$ as part of a center
phase $\phi$ and resetting
$p.\varDelayingResponse(\chi) \gets \False$
at the beginning of the next phase $\phi'$, after a $\msgGoOn$ response was
sent over $\chi$ if needed (i.e., if $\phi'$ is an arm phase).
Table~\ref{tab:principal-request-behavior} provides a summary of the
aforementioned reactions.

\begin{table}[ht]
\centering
\begin{tabular}{@{}c|c|c@{}}
\toprule
& \textbf{$p.\varPhaseType = \Arm$} 
& \textbf{$p.\varPhaseType = \Center$} \\
\hline
$\mu = \msgWaiting$ 
& \multicolumn{1}{l|}{send $\msgGoOn$ over $\chi$}
& \multicolumn{1}{l}{set $p.\varDelayingResponse(\chi) \gets \True$} 
\\
\hline
& \multicolumn{1}{l|}{send $\msgTransport(p.\varTokens)$ over $\chi$}
& 
\\
$\mu = \msgTokensPlease$
& \multicolumn{1}{l|}{set $p.\varTokens \gets 0$ (notify channel layer)} 
& \multicolumn{1}{l}{send $\msgNoTransport$ over $\chi$} 
\\
& \multicolumn{1}{l|}{end current phase (abruptly)}
& 
\\
\bottomrule
\end{tabular}
\caption{\label{tab:principal-request-behavior}%
The reaction of a primary node
$p \in P$
upon receiving a request message $\mu$ over channel
$\chi \in \Channels(p)$
depending on
$p.\varPhaseType$.}
\end{table}

We emphasize that the
$p.\varAwaitingResponse(\chi)$
and
$p.\varDelayingResponse(\chi)$
variables of $p$ are maintained only for channels
$\chi \in \Channels(p)$.
In particular, if channel $\chi$ is added to (resp., removed from)
$\Channels(p)$ in the midst of phase $\phi$, then the variables
$p.\varAwaitingResponse(\chi)$
and
$p.\varDelayingResponse(\chi)$
are created (resp., deleted) with it, where both of them are
initialized to $\False$.

\section{Analysis of the Algorithm --- Correctness and Message Load}
\label{sec:analysis}
Throughout this section, we fix some adversarial policy $\Adv$ and analyze the
execution of algorithm \AlgTF{} under $\Adv$.
We start with some useful definitions, followed by \Obs{}\
\ref{obs:analysis:phase-length} and \ref{obs:analysis:transport-in-transit},
that capture basic features of \AlgTF{}, and
\Lem{}~\ref{lem:analysis:operational-channel}, that serves as the cornerstone
of the entire analysis.
\Sect{}~\ref{sec:analysis:safety-liveness} is then dedicated to
establishing the algorithm's correctness while the message load analysis is
presented in \Sect{}~\ref{sec:analysis:message-load}.
The reaction time analysis is deferred to \Sect{}~\ref{sec:reaction-time}
as it builds on a certain implementation feature of the channel layer,
presented in \Sect{}~\ref{sec:channel-layer-algorithm}.
Throughout the analysis, we condition on the event that the assertion of
\Lem{}~\ref{lem:pu-graph} holds.

We say that a primary node
$p \in P$
\emph{retires} at time
$t > 0$
if
$p \in \OperationalNodes^{t} - \Busy^{t+}$,
observing that this holds if and only if a phase $\phi$ of $p$ ends at time
$t$ and either
(1)
$\phi$ is an arm phase that ends abruptly and $p$ sends a $\msgTransport$
message at time $t$;
or
(2)
$p$ performs team formation(s) at time $t$.
We say that a
$\{ p_{1}, p_{2} \}$-channel
$\chi$ \emph{retires} at time
$t > 0$
if
$\chi \in \OperationalChannels^{t} - \OperationalChannels^{t+}$,
observing that by
\Grnt{}~\ref{grnt:algorithm:interface:operational-channel-continues},
this holds if and only if
$\chi \in \OperationalChannels^{t}$
and $p_{i}$ retires at time
$t$ for some (exactly one)
$i \in \{ 1, 2 \}$.
Notice that if a channel $\chi$ becomes operational (resp., retires) at time
$t > 0$
(resp., at time
$t' > t$),
then
$\chi$ is created (resp., released) by its mediator strictly before time $t$
(resp., strictly after time $t'$).

For
$p \in P$
and
$t > 0$,
let
$\Next_{p}(t)$
be the operator that returns the earliest time
$t' > t$
such that a phase of $p$ ends at time $t'$ if
$p \in \OperationalNodes^{t+}$,
and $t$ otherwise (\Obs{}~\ref{obs:analysis:phase-length}
ensures that this operator is well defined).
Notice that if we fix time
$t > 0$,
including all coin tosses up to (excluding) time $t$, then
$\Next_{p}(t)$
is a random variable that depends on the coin tosses from time $t$ onward.
Notice further that while the condition
$p \in \OperationalNodes^{t}
\land
\Next_{p}(t) = t$
does not imply that $p$ retires at time $t$ (it may be the case that $p$ is no
longer operational at time $t+$ although it is still busy), this condition
does imply that every channel
$\chi \in \Channels^{t}(p)$
retires at or before time $t$.

The
$\Next_{p}(\cdot)$
operator is extended inductively as follows:
let
$\Next_{p}^{0}(t) = t$;
for
$i > 0$,
let
$\Next_{p}^{i}(t) = \Next_{p}( \Next_{p}^{i - 1}(t) )$.
This extension gives us a convenient handle for reasoning about events that
occur ``within the next $i$ phases'' of $p$.

\begin{observation}
\label{obs:analysis:phase-length}
A phase that begins at time
$t > 0$
is guaranteed to end at or before time
$t + 8$.
\end{observation}

\begin{proof}
Consider a primary node
$p \in P$
and a phase $\phi$ of $p$ that begins at time $t$.
If $\phi$ is a center phase, then the $\msgTokensPlease$ requests that $p$
sends at time $t$ over the channels in
$\Channels^{t+}(p)$
are replied, with $\msgTransport$ or $\msgNoTransport$ responses, as soon as
they arrive;
by \Grnt{}~\ref{grnt:algorithm:interface:relayed-messages}, these
replies arrive to $p$ at or before time
$t + 4$.
If $\phi$ is an arm phase, then the $\msgWaiting$ request $\mu$ that $p$ sends
over a
$\{ p, p' \}$-channel
$\chi$ is not replied immediately if (and only if) $p'$ is a center node;
rather, $\mu$ is replied (or $\chi$ retires) at the beginning of the
subsequent phase of $p'$.
The assertion follows by
\Grnt{}~\ref{grnt:algorithm:interface:relayed-messages} as we have
already proved that the current phase of $p'$ ends within $4$ time units.
\end{proof}

\begin{observation}
\label{obs:analysis:transport-in-transit}
If a $\msgTransport$ response is in transit from a primary node $p$ to a
primary node $p'$ over a
$\{ p, p' \}$-channel
$\chi$ at time
$t > 0$,
then
$p' \in \Busy^{t}$
and
$\chi \in \Channels^{t}(p')$.
\end{observation}

\begin{proof}
The primary node $p$ sends a $\msgTransport$ response over $\chi$ at time
$\bar{t} > 0$
only if a $\msgTokensPlease$ request $\mu$ is received by $p$ over $\chi$ at
time
$\bar{t}$.
Taking $\phi$ to be the phase during which node $p'$ sent the request $\mu$,
the assertion is satisfied since $\phi$ does not end before a reply for $\mu$
is received over $\chi$.
\end{proof}

\begin{lemma}
\label{lem:analysis:operational-channel}
Fix time
$t_{0} > 0$
and consider a
$\{ p_{1}, p_{2} \}$-channel
$\chi \in \OperationalChannels^{t_{0}}$.
With probability at least
$1 / 16$,
channel $\chi$ retires during the time interval
$[ t_{0}, \min_{i \in \{ 1, 2 \}} \Next_{p_{i}}^{3}(t_{0}) ]$.
\end{lemma}

We first prove Lemma~\ref{lem:analysis:operational-channel-with-center-peers}, which forms the main component of the proof of Lemma~\ref{lem:analysis:operational-channel}, and then proceed to prove the latter.

\begin{lemma}
\label{lem:analysis:operational-channel-with-center-peers}
Fix time
$t_{0} > 0$
and consider a
$\{ p_{1}, p_{2} \}$-channel
$\chi \in \OperationalChannels^{t_{0}}$.
If
$p_{1}.\varPhaseType^{t_{0}} = p_{2}.\varPhaseType^{t_{0}} = \Center$,
then with probability at least
$1 / 4$,
channel $\chi$ retires during the time interval
$[ t_{0}, \min_{i \in \{ 1, 2 \}} \Next_{p_{i}}^{2}(t_{0}) ]$.
\end{lemma}
\begin{proof}%
[Proof of \Lem{}~\ref{lem:analysis:operational-channel-with-center-peers}]
For
$i \in \{ 1, 2 \}$,
let
$t_{i} = \Next_{p_{i}}(t_{0})$
and let
$t'_{i} = \Next_{p_{i}}^{2}(t_{0})$;
we need to prove that with probability at least
$1 / 4$,
channel $\chi$ retires at or before time
$\min \{ t'_{1}, t'_{2} \}$.
If
$t_{i} = t_{0}$
for some
$i \in \{ 1, 2 \}$,
then $\chi$ retires at time $t_{0}$ and the assertion is satisfied, so assume
without loss of generality that
$t_{0} < t_{1} < t_{2}$,
implying that
$\chi \in \OperationalChannels^{t}$
for all
$t_{0} \leq t \leq t_{1}$.
If $\chi$ retires at time $t_{1}$, then $p_{1}$ retires at time $t_{1}$ and
$t'_{1} = t_{1} < t_{2} \leq t'_{2}$,
hence
$t_{1} = \min \{ t'_{1}, t'_{2} \}$
and the assertion is satisfied.
So assume that $\chi$ does not retire at time $t_{1}$, implying that a new
phase $\phi_{1}$ of $p_{1}$ starts at time $t_{1}$ and ends at time
$t'_{1} > t_{1}$.

With probability
$1 / 2$,
phase $\phi_{1}$ is an arm phase --- condition hereafter on this event.
Thus, the following claim holds:
(C1)
at time $t_{1}$, node $p_{1}$ sends a $\msgWaiting$ request $\mu_{1}$ over
$\chi$ and unless $\chi$ retires, phase $\phi_{1}$ cannot end as long as
$p_{1}$ does not receive a ($\msgGoOn$) response for $\mu_{1}$ over $\chi$.
If
$t_{2} > t'_{1}$,
then as
$t'_{2} \geq t_{2} > t'_{1}$,
it suffices to show that $\chi$ retires at time
$t'_{1} = \min \{ t'_{1}, t'_{2} \}$.
Recalling that $p_{2}$ is a center node at least until time $t_{2}$, we
conclude that $p_{2}$ does not respond to $\mu_{1}$ before time
$t_{2} > t'_{1}$,
hence $p_{1}$ does not receive a response for $\mu_{1}$ at or before time
$t'_{1}$.
Claim (C1) implies that $\chi$ retires at time $t'_{1}$, hence the
assertion holds.

In the remainder of this proof, assume that
$t_{1} < t_{2} < t'_{1}$.
If $\chi$ retires at time $t_{2}$, then $p_{2}$ retires at time $t_{2}$ and
$t'_{2} = t_{2} < t'_{1}$,
hence
$t_{2} = \min \{ t'_{1}, t'_{2} \}$
and the assertion is satisfied.
So assume that $\chi$ does not retire at time $t_{2}$, implying that a new
phase $\phi_{2}$ of $p_{2}$ starts at time $t_{2}$ and ends at time
$t'_{2} > t_{2}$.
With probability
$1 / 2$,
phase $\phi_{2}$ is a center phase --- condition hereafter on this event
(refer to Figure~\ref{fig:retire-lemma-timeline-arrows} for an illustration).
Therefore, the following claim holds:
(C2)
at time $t_{2}$, node $p_{2}$ sends a $\msgTokensPlease$ request $\mu_{2}$
over $\chi$ and unless $\chi$ retires, phase $\phi_{2}$ cannot end as long as
$p_{2}$ does not receive a ($\msgTransport$ or $\msgNoTransport$) response for
$\mu_{2}$ over $\chi$.
Recalling that $p_{2}$ is a center also before time $t_{2}$, we deduce that the
following claim holds as well:
(C3)
node $p_{2}$ does not respond to the $\msgWaiting$ request $\mu_{1}$ at least
until time $t'_{2}$ (if at all).

\begin{figure}[ht]
\centering
\includegraphics[width=0.8\linewidth]{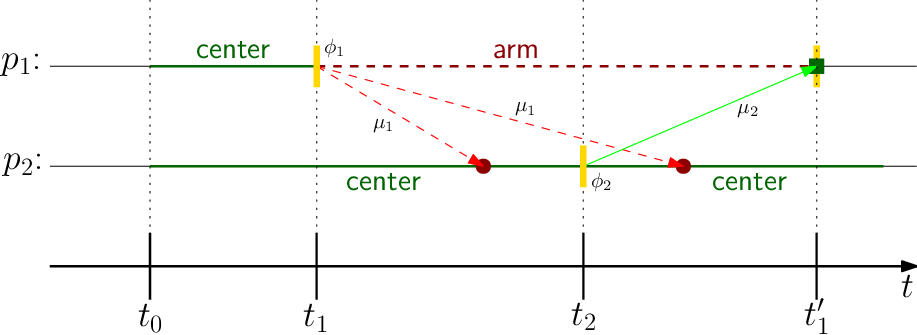}
\caption{\label{fig:retire-lemma-timeline-arrows}%
Timeline for the case where
$t_{1} < t_{2} < t'_{1}$. 
Some of the possible times for $\mu_{1}$ (resp., $\mu_{2}$) to arrive at
$p_{2}$ (resp., $p_{1}$).
Node $p_{2}$ awaits (as a center) a response for $\mu_{2}$ from $p_{1}$ before
$p_{2}$ responds (if at all) to $\mu_{1}$.
Meanwhile, node $p_{1}$ awaits (as an arm) a response for $\mu_{1}$ from
$p_{2}$ or for some node to request $p_{1}$'s tokens.}
\end{figure}

Since $\chi$ retires only when $p_{i}$ retires for some
$i \in \{ 1, 2 \}$
and since $p_{i}$ does not retire before phase $\phi_{i}$ ends, it follows,
by claims (C1) and (C2), that the following claim holds:
(C4)
there exists (at least one choice of)
$i \in \{ 1, 2 \}$
such that $p_{i}$ receives a response to $\mu_{i}$ over $\chi$ during
$\phi_{i}$.
By claim (C3), we deduce that claim (C4) must be satisfied for
$i = 2$,
i.e., node $p_{2}$ receives a response to $\mu_{2}$ over $\chi$ during
$\phi_{2}$;
let
$t_{2} < \bar{t} < t'_{2}$
be the arrival time of $\mu_{2}$ to $p_{1}$.
Node $p_{1}$ does not receive a response for $\mu_{1}$ before time $\bar{t}$
(again, by claim (C3)), thus
$t_{1} < t_{2} < \bar{t} \leq t'_{1}$,
implying that $p_{1}$ receives $\mu_{2}$ during phase $\phi_{1}$.
We conclude that $p_{1}$ receives the $\msgTokensPlease$ request $\mu_{2}$
during the arm phase $\phi_{1}$ and reacts by sending a $\msgTransport$
response over $\chi$ and retiring.
Therefore,
$\bar{t} = t'_{1} < t'_{2}$,
implying that $\chi$ retires at time
$t'_{1} = \min \{ t'_{1}, t'_{2} \}$,
thus completing the proof.
\end{proof}

\begin{proof}%
[Proof of \Lem{}~\ref{lem:analysis:operational-channel}]
For
$i \in \{ 1, 2 \}$,
let
$t_{i} = \Next_{p_{i}}(t_{0})$,
let
$t'_{i} = \Next_{p_{i}}^{2}(t_{0})$,
and let
$t''_{i} = \Next_{p_{i}}^{3}(t_{0})$;
we need to prove that with probability at least
$1 / 16$,
channel $\chi$ retires at or before time
$\min \{ t''_{1}, t''_{2} \}$.
If
$t_{i} = t_{0}$
for some
$i \in \{ 1, 2 \}$,
then $\chi$ retires at time $t_{0}$ and the assertion is satisfied, so assume
without loss of generality that
$t_{0} < t_{1} < t_{2}$,
implying that
$\chi \in \OperationalChannels^{t}$
for all
$t_{0} \leq t \leq t_{1}$.
If $\chi$ retires at time $t_{1}$, then $p_{1}$ retires at time $t_{1}$ and
$t''_{1} = t'_{1} = t_{1} < t_{2} \leq t'_{2}$,
hence
$t_{1} = \min \{ t''_{1}, t''_{2} \}$
and the assertion is satisfied.
So assume that $\chi$ does not retire at time $t_{1}$, implying that a new
phase $\phi_{1}$ of $p_{1}$ starts at time $t_{1}$ and ends at time
$t'_{1} > t_{1}$.

With probability
$1 / 2$,
phase $\phi_{1}$ is a center phase --- condition hereafter on this event.
Therefore, the following claim holds:
(C1)
at time $t_{1}$, node $p_{1}$ sends a $\msgTokensPlease$ request $\mu_{1}$
over $\chi$ and unless $\chi$ retires, phase $\phi_{1}$ cannot end as long as
$p_{1}$ does not receive a ($\msgTransport$ or $\msgNoTransport$) response for
$\mu_{1}$ over $\chi$.
If $p_{2}$ is a center at time $t_{1}$, then we establish the assertion by
applying
\Lem{}~\ref{lem:analysis:operational-channel-with-center-peers}
to time $t_{1}+$, observing that
$\Next_{p_{1}}^{2}(t_{1}+) = \Next_{p_{1}}^{3}(t_{0})$
and
$\Next_{p_{2}}^{2}(t_{1}+) \leq \Next_{p_{2}}^{3}(t_{0})$.
So, assume that $p_{2}$ is an arm at time $t_{1}$.

If the $\msgTokensPlease$ request $\mu_{1}$ arrives to $p_{2}$ (at or) before
time $t_{2}$, then $p_{2}$ receives this request during an arm phase, thus
$p_{2}$ reacts by sending a $\msgTransport$ response over $\chi$ and retiring.
Therefore, $\chi$ retires at time $t_{2}$ and by claim (C1),
$t_{2} < t'_{1} \leq t''_{1}$,
establishing the assertion as
$t_{2} \leq t'_{2} \leq t''_{2}$.
So, assume that $\mu_{1}$ does not arrive to $p_{2}$ (at or) before time
$t_{2}$, implying that
$t_{2} < t'_{1}$.

If $\chi$ retires at time $t_{2}$, then $p_{2}$ retires at time $t_{2}$ and
$t''_{2} = t'_{2} = t_{2} < t'_{1}$,
hence
$t_{2} = \min \{ t''_{1}, t''_{2} \}$
and the assertion is satisfied.
So assume that $\chi$ does not retire at time $t_{2}$, implying that a new
phase $\phi_{2}$ of $p_{2}$ starts at time $t_{2}$ and ends at time
$t'_{2} > t_{2}$.
With probability
$1 / 2$,
phase $\phi_{2}$ is a center phase --- condition hereafter on this event.
The proof is now completed by applying
\Lem{}~\ref{lem:analysis:operational-channel-with-center-peers}
to time $t_{2}+$, observing that
$\Next_{p_{1}}^{2}(t_{2}+) = \Next_{p_{1}}^{3}(t_{0})$
and
$\Next_{p_{2}}^{2}(t_{2}+) = \Next_{p_{2}}^{3}(t_{0})$.
\end{proof}

\subsection{Safety and Liveness}
\label{sec:analysis:safety-liveness}
In this section, we establish the correctness of \AlgTF{}, proving that it
satisfies the safety and liveness conditions.
The former condition holds trivially as tokens are deleted only during team
formation operations, so the remainder of this section is dedicated to the
latter.

Consider time
$t_{0} > 0$
and assume that the system contains at least $\sigma$ tokens at time $t_{0}$.
For time
$t \geq t_{0}$,
node
$p \in P$,
and
$\{ p, p' \}$-channel
$\chi \in \Channels^{t}(p)$,
let
$R^{t}(p, \chi)$
be the number of tokens transported over $\chi$ towards $p$ at time $t$.
Notice that
$R^{t}(p, \chi) = k > 0$
if and only if there is a
$\msgTransport(k)$
message in transit from $p'$ to $p$ over $\chi$ at time $t$.
For time
$t \geq t_{0}$
and node
$p \in P$,
let
$R^{t}(p) = \sum_{\chi \in \Channels^{t}(p)} R^{t}(p, \chi)$
if
$p \in \Busy^{t}$,
and
$R^{t}(p) = 0$
otherwise.

By \Obs{}~\ref{obs:analysis:phase-length}, a primary node
that holds at least $\sigma$ tokens is guaranteed to form a team in
$O (1)$
time.
\Grnt{}~\ref{grnt:algorithm:interface:relayed-messages}
ensures that every $\msgTransport$ message reaches its destination in
$O (1)$
time.
Therefore, if
$p.\varTokens^{t} + R^{t}(p) \geq \sigma$,
then $p$ is guaranteed to form a team by time
$t + O (1)$.

Define the potential function
$\psi : \Reals_{> 0} \rightarrow \Integers_{\geq 0}$
as
\[\textstyle
\psi(t)
\, = \,
\sum_{p \in P} \left( p.\varTokens^{t} + R^{t}(p) - 1 \right)
\]
and notice that if
$\psi(t) \geq (\sigma - 1) n$,
then
$p.\varTokens^{t} + R^{t}(p) \geq \sigma$
for some
$p \in P$,
hence a team is guaranteed to be formed by time
$t + O (1)$.
\Obs{}~\ref{obs:analysis:transport-in-transit} ensures
that if no teams are formed during the time interval
$I = [t_{0}, \cdot] \subset \Reals_{> 0}$,
then the function $\psi(t)$ is non-decreasing in $I$.
Moreover, a primary node
$p \in P$
retires at time
$t \geq t_{0}$
without forming any team if and only if $p$ sends a $\msgTransport$ message at
time $t$, implying that
$\psi(t+) = \psi(t) + 1$.
So, it suffices to show that as long as no teams are formed, the number of
node retirements must increase.

By \Obs{}~\ref{obs:analysis:phase-length} and
\Lem{}~\ref{lem:analysis:operational-channel}, we know that if
$\OperationalChannels^{t} \neq \emptyset$,
then at least one primary node is certain to retire in finite time with
probability $1$.
The liveness proof is completed by observing that if the system contains
$k \geq \sigma$
tokens and there are no tokens in transit, then either
(I)
all $k$ tokens are held by a single busy node that forms a team in
$O (1)$
time;
or
(II)
the $k$ tokens are distributed over multiple busy nodes, in which case, an
operational channel is certain to be generated in
$O (1)$
time by
\Grnt{}~\ref{grnt:algorithm:interface:operatioanl-channel-generation}.

\begin{theorem}
\label{thm:analysis:safety-and-liveness}
\AlgTF{} satisfies the safety and liveness conditions for any TF instance.
\end{theorem}

\subsection{Message Load}
\label{sec:analysis:message-load}
We now turn to analyze the message load of \AlgTF{} in expectation and whp.
To this end, fix some
$\ell \in \Integers_{> 0}$
and assume that the adversarial policy $\Adv$ injects at most $\ell$ tokens
throughout the execution
(i.e., $\Adv \in \A_{\ell}$
in the language of \Sect{}~\ref{sec:preliminaries:model}).
The analysis starts with bounding the number of relayed messages sent over
non-operational channels.

\begin{lemma}
\label{lem:analysis:message-load:non-operational-channels}
The total number of relayed messages sent throughout the execution over
channels $\chi$ while
$\chi \notin \OperationalChannels$
is
$O (\ell \cdot \sqrt{n \log n})$.
\end{lemma}

\begin{proof}
Consider some channel $\chi$ and observe that at most one relayed message (a
request) is sent over $\chi$ before it is operational (if at all) and at most
two relayed messages (a response and a request) are sent over $\chi$ after it
is no longer operational.
Since the channel layer sends
$\Theta (1)$
messages for the purpose of creating and releasing $\chi$, we can charge the
(at most $3$) relayed messages sent over $\chi$ while
$\chi \notin \OperationalChannels$
to the channel layer.
The assertion follows by
\Grnt{}~\ref{grnt:algorithm:interface:message-complexity}.
\end{proof}

Combined with
\Grnt{}~\ref{grnt:algorithm:interface:message-complexity}, we conclude
that with the exception of relayed messages sent over operational channels,
the total number of messages sent by \AlgTF{} is
$O (\ell \cdot \sqrt{n \log n})$,
so it remains to bound the number of former messages.
To this end, we establish \Lem{}~\ref{lem:analysis:message-load:blaming-tokens}
and \Lem{}~\ref{lem:analysis:message-load:operational-channels} (whose proof follows almost directly from
\Lem{}~\ref{lem:analysis:operational-channel}).

\begin{lemma}
\label{lem:analysis:message-load:blaming-tokens}
$\left| \bigcup_{t > 0} \OperationalChannels^{t} \right|
\leq
O (\ell \cdot \sqrt{n \log n})$.
\end{lemma}
\begin{proof}
Let $\T$ (resp., $\T'$) be the set of tokens injected into the system
throughout the execution, including (resp., excluding) fakely injected tokens.
Since each fake injection into a primary node
$p \in P$
includes
$k < \sigma$
tokens and can be (injectively) attributed to a team formation operation of
$p$ that includes $\sigma$ tokens, it follows that
$| \T | < 2 | \T' |$.
Recalling that
$| \T | \leq \ell$,
we conclude that
$| \T | \leq O (\ell)$.

The assertion is established by introducing a ``blame'' function
$B : \bigcup_{t > 0} \OperationalChannels^{t} \rightarrow \T$
and proving that $B$ maps at most
$O (\sqrt{n \log n})$
channels to each token in $\T$.
To this end, consider a
$\{ p_{1}, p_{2} \}$-channel
$\chi \in \bigcup_{t > 0} \OperationalChannels^{t}$,
let
$\bar{t} > 0$
be the retirement time of $\chi$,
and let
$i \in \{ 1, 2 \}$
be the (unique) index for which $p_{i}$ retires at time $\bar{t}$.
\Obs{}~\ref{obs:analysis:transport-in-transit} implies
that at least one of the tokens held by $p_{i}$ at time $\bar{t}$ was
injected into $p_{i}$ itself (before time $\bar{t}$) --- the function $B$ maps
$\chi$ to such a token $\tau$ (breaking ties arbitrarily).
This construction ensures that only the channels in
$\Channels^{\bar{t}}(p_{i})$ are mapped under $B$ to $\tau$, thus completing
the proof due to \Grnt{}~\ref{grnt:algorithm:interface:degree-bound}.
\end{proof}

\begin{lemma}
\label{lem:analysis:message-load:operational-channels}
Consider a
$\{ p, p' \}$-channel
$\chi$ and let $Y_{\chi}$ be a random variable that counts the number of
relayed messages sent over $\chi$ while
$\chi \in \OperationalChannels$.
There exist constants
$\alpha, \beta > 0$
such that
$\Pr (Y_{\chi} > y) < \beta e^{-\alpha y}$
for every
$y \in \Integers_{\geq 0}$.
\end{lemma}

\begin{proof}
Charge every response message to the corresponding request message, so it
suffices to up-bound the latter.
The assertion follows from
\Lem{}~\ref{lem:analysis:operational-channel} by recalling that
while
$\chi \in \OperationalChannels$,
every phase of either $p$ or $p'$ accounts for one request message over $\chi$.
\end{proof}

For a channel $\chi$, let $Y_{\chi}$ be the random variable in
\Lem{}~\ref{lem:analysis:message-load:operational-channels}.
Recalling that the expected value of an exponentially decaying random variable
is bounded by
$O (1)$,
we can combine \Lem{}\
\ref{lem:analysis:message-load:blaming-tokens} and
\ref{lem:analysis:message-load:operational-channels} to deduce,
by the linearity of expectation, that \AlgTF{} sends
$O (\ell \cdot \sqrt{n \log n})$
messages in expectation.

At this stage, one may be tempted to advance towards a whp bound by assuming
that the random variables $Y_{\chi}$ of different channels $\chi$ are
independent.
Unfortunately, this assumption is wrong:
the number of relayed messages sent over different channels in
$\OperationalChannels$ may be strongly correlated.
To resolve this difficulty, we use the following lemma, stating that the
average of finitely many exponentially decaying random variables is also
exponentially decaying even if the original random variables exhibit complex
dependencies;
we will not be surprised to learn that this lemma is known from the existing
literature, however, we could not find it anywhere and therefore, provide a
standalone proof.

\begin{lemma}
\label{lem:average-exponentially-decaying-rvs}
Fix some
$\alpha, \beta > 0$
and let
$X_{1}, \dots X_{m}$
be (not necessarily independent) random variables over
$\Reals_{> 0}$
such that
$\Pr (X_{i} > x) < \beta e^{-\alpha x}$
for every
$i \in [m]$
and
$x > 0$.
For every
$\epsilon > 0$,
there exists
$\beta' = \beta'(\beta, \epsilon) > 0$
such that
$\Pr (\bar{X} > x) < \beta' e^{-\frac{\alpha x}{1 + \epsilon}}$
for every
$x > 0$,
where
$\bar{X} = \frac{1}{m} \sum_{i \in [m]} X_{i}$.
\end{lemma}

\begin{proof}
Let
$\alpha' = \frac{\alpha}{1 + \epsilon}$.
We establish the lemma by identifying some
$\beta' = \beta'(\beta, \epsilon) > 0$
that satisfies
$\Pr (\bar{X} > x / \alpha') < \beta' e^{-x}$
for every
$x > 0$.
To this end, develop
\begin{align*}
\Pr \left( \bar{X} > x / \alpha' \right)
\, = \, &
\Pr \left( \alpha' \bar{X} > x \right)
\\
= \, &
\Pr \left( e^{\alpha' \bar{X}} > e^{x} \right)
\\
< \, &
\frac{\Ex \left( e^{\alpha' \bar{X}} \right)}{e^{x}}
\\
= \, &
e^{-x}
\cdot
\Ex \left(
\exp \left( \frac{1}{m} \sum_{i \in [m]} \alpha' X_{i} \right)
\right)
\\
= \, &
e^{-x}
\cdot
\Ex \left( \sqrt[m]{\prod_{i \in [m]} e^{\alpha' X_{i}}} \right)
\\
\leq \, &
e^{-x}
\cdot
\Ex \left( \frac{1}{m} \sum_{i \in [m]} e^{\alpha' X_{i}} \right)
\\
= \, &
e^{-x}
\cdot
\frac{1}{m} \sum_{i \in [m]} \Ex \left( e^{\alpha' X_{i}} \right)
\, ,
\end{align*}
where
the third transition is due to Markov's inequality
and
the penultimate transition is due to the AM-GM inequality.
We complete the proof by choosing
$\beta' = \beta'(\beta, \epsilon) > 0$
such that
$\Ex \left( e^{\alpha' X_{i}} \right) \leq \beta'$
for every
$i \in [m]$.
To identify such $\beta'$, we develop
\begin{align*}
\Ex \left( e^{\alpha' X_{i}} \right)
\, = \, &
\int_{0}^{\infty} \Pr \left( e^{\alpha' X_{i}} > z \right) d z
\\
\leq \, &
1 + \int_{1}^{\infty} \Pr \left( e^{\alpha' X_{i}} > z \right) d z
\\
= \, &
1 + \int_{1}^{\infty} \Pr \left( \alpha' X_{i} > \ln z \right) d z
\\
= \, &
1 + \int_{1}^{\infty} \Pr \left( X_{i} > \frac{1}{\alpha'} \ln z \right) d z
\\
< \, &
1
+
\int_{1}^{\infty}
\beta \cdot \exp \left( -\frac{\alpha}{\alpha'} \ln z \right) d z
\\
= \, &
1 + \beta \int_{1}^{\infty} z^{-(1 + \epsilon)} d z
\\
= \, &
1 + \beta / \epsilon
\end{align*}
and conclude that it suffices to choose
$\beta' = 1 + \beta / \epsilon$.
\end{proof}

Let
$m = O (\ell \cdot \sqrt{n \log n})$
be the bound promised in \Lem{}~\ref{lem:analysis:message-load:blaming-tokens}.
By applying \Lem{}~\ref{lem:average-exponentially-decaying-rvs} to the (at most)
$m$ random variables in
$\{ Y_{\chi} \}_{\chi \in \bigcup_{t > 0} \OperationalChannels^{t}}$,
we conclude that
$\frac{1}{m} \sum_{\chi \in \bigcup_{t > 0} \OperationalChannels^{t}} Y_{\chi}
\leq
O (\log n)$
whp, hence \AlgTF{} sends
$O (\ell \cdot \sqrt{n \log n} \cdot \log n)$
messages whp.

\begin{theorem}
\label{thm:analysis:message-load}
The message load of \AlgTF{} is
$O (\sqrt{n \log n})$
in expectation and
$O (\sqrt{n \log n} \cdot \log n)$
whp.
\end{theorem}

\section{The Channel Layer}
\label{sec:channel-layer-algorithm}
In this section, we describe how the channel layer is implemented;
the service assurances that this layer provides to the principal
layer, as stated in \Grnt{}\
\ref{grnt:algorithm:interface:operational-implies-busy}--%
\ref{grnt:algorithm:interface:message-complexity},
are then established in \Sect{}~\ref{sec:channel-layer-algorithm:analysis}.
A pseudocode description of the channel layer is presented in
\Sect{}~\ref{sec:pseudocode-channel}, consisting of Algorithms
\ref{alg:channel-layer-primary}, \ref{alg:channel-layer-utility}, and
\ref{alg:procedure-create-channel}, and accompanied by Tables \ref{tab:messages-channel}
and \ref{tab:variables-channel};
below, we present a textual description.

The starting point of the channel layer is that a primary node
$p \in P$
updates all the utility nodes in $U(p)$ whenever $p$ becomes busy by sending
them $\msgBusy$ messages.
A utility node
$u \in U$
maintains the local variable
$u.\varBusyTokens(p) \in \Integers_{\geq 0} \cup \{ \bot \}$
for each
$p \in P(u)$,
initiated so that
$u.\varBusyTokens(p) = \bot$.\footnote{%
\label{footnote:varBusyTokens}%
Strictly speaking, since $u$ is not assumed to know the set $U(p)$, the
algorithm actually maintains the variable $u.\varBusyTokens(p)$ for all
primary nodes
$p \in P$,
where $u.\varBusyTokens(p)$ may get an integral value (i.e., a value other
than $\bot$) only if
$p \in P(u)$.}
When $u$ receives a $\msgBusy$ message from $p$, node $u$ sets
$u.\varBusyTokens(p) \gets 0$
and responds with a $\msgBusyAck$ message;
the arrival of this $\msgBusyAck$ message to $p$ prompts $p$ to record that
``$u$ knows that $p$ is busy'' by adding $u$ to the local variable
$p.\varBusyAcked \subseteq U(p)$
that $p$ maintains.
Following that, $p$ updates, using $\msgTokensUpdate$ messages, the utility
nodes in $p.\varBusyAcked$ on any increase in $p.\varTokens$ and waits for
these utility nodes to create
$\{ p, \cdot \}$-channels.
Upon arrival of a $\msgTokensUpdate(k)$ message from a primary node
$p \in P(u)$,
the utility node $u$ records the corresponding update by setting
$u.\varBusyTokens(p) \gets k$.

Fix some utility node
$u \in U$
and let
$S = \{ p \in P(u) : u.\varBusyTokens(p) \in \Integers_{> 0} \}$.
If
$u.\varChannel = \emptyset$
and
$|S| \geq 2$,
then $u$ creates a
$\{ p_{1}, p_{2} \}$-channel
$\chi$ for some
$p_{1}, p_{2} \in S$
by setting
$u.\varChannel \gets \{ p_{1}, p_{2} \}$,
giving preference to the nodes
$p \in S$
with larger values of
$u.\varBusyTokens(p)$;
following that, $u$ sends a $\msgChannel$ message to $p_{i}$ for
$i \in \{ 1, 2 \}$.
Upon arrival of this $\msgChannel$ message to $p_{i}$, if $p_{i}$ is (still)
busy, then $p_{i}$ adds $u$ to its list of channel mediators by setting
$p_{i}.\varMediators \gets p_{i}.\varMediators \cup \{ u \}$.

As long as
$u \in p_{1}.\varMediators \cap p_{2}.\varMediators$,
the primary nodes $p_{1}$ and $p_{2}$ can use channel $\chi$ to exchange
relayed messages.
This lasts until $p_{i}$ becomes non-busy for some
$i \in \{ 1, 2 \}$,
at which point $p_{i}$ sends a $\msgNotBusy$ message to all channel mediators
in
$p_{i}.\varMediators$,
including $u$, and following that, resets
$p_{i}.\varMediators \gets \emptyset$.
When $u$ receives this $\msgNotBusy$ message from $p_{i}$, it resets
$u.\varBusyTokens(p) \gets \bot$;
moreover, if
$u.\varChannel = \{ p_{1}, p_{2} \}$
(still) holds, then $u$ releases $\chi$, resetting
$u.\varChannel \gets \emptyset$,
and sends a $\msgNoChannel$ message to
$p_{3 - i}$.
Upon arrival of this $\msgNoChannel$ to
$p_{3 - i}$,
the primary node
$p_{3 - i}$
removes $u$ from its list of channel mediators by setting
$p_{3 - i}.\varMediators \gets p_{3 - i}.\varMediators - \{ u \}$.

On top of the above, a primary node
$p \in P$
that receives a $\msgChannel$ message from a utility node
$u \in U(p)$
responds ``automatically'' with a $\msgChannelAck$ message, even if $p$ is no
longer busy.
The utility node $u$ maintains the local variable
$u.\varDiff(p) \in \Integers_{\geq 0}$
for each
$p \in P(u)$,
incrementing
$u.\varDiff(p) \gets u.\varDiff(p) + 1$
on every $\msgChannel$ message that $u$ sends to $p$
and decrementing
$u.\varDiff(p) \gets u.\varDiff(p) - 1$
on every $\msgChannelAck$ message that $u$ receives from $p$.\footnote{%
Similarly to the discussion in footnote~\ref{footnote:varBusyTokens} with
regards to the variable $u.\varBusyTokens(p)$, the algorithm  actually
maintains the variable $u.\varDiff(p)$ for all primary nodes
$p \in P$,
where $u.\varDiff(p)$ may get a positive value only if
$p \in P(u)$.}
Assuming that
$u.\varChannel = \{ p, p' \}$,
this mechanism allows $u$ to determine whether an incoming relayed message
$\mu$ from $p$ was sent as part of the ``current channel'' (rather than an
outdated session of $p$ that has already been released);
indeed, $\mu$ is relayed to $p'$ if and only if
$u.\varDiff(p) = 0$.

\subsection{Analysis}
\label{sec:channel-layer-algorithm:analysis}
The analysis of the channel layer is facilitated by
Table~\ref{tab:legal-configurations} that depicts the possible configurations
of the link between a primary node
$p \in P$
and a utility node
$u \in U(p)$,
as well as the values of the relevant local variables of $p$ and $u$.
Table~\ref{tab:legal-configurations-incidence-matrix} depicts the incidence
matrix defined over the configurations of Table~\ref{tab:legal-configurations}
and the reactions of $p$ and $u$ to the relevant local events, covering all
possible combinations and therefore, providing a rigorous proof that all, and
only, the configurations in Table~\ref{tab:legal-configurations} are reachable
from the initial configuration in which
$p.\varTokens = 0$,
$p.\varBusyAcked = \emptyset$,
$p.\varMediators = \emptyset$,
$u.\varBusyTokens(p) = \bot$,
$u.\varChannel = \emptyset$,
and there are no messages in transit between $p$ and $u$.

\begin{table}[ht]
\centering
\begin{tabular}{@{}c|ccccc|cc@{}}
\toprule
\rotatebox{90}{configuration} &
\rotatebox{90}{$p.\varTokens$} &
\rotatebox{90}{$p.\varBusyAcked \ni u$} &
\rotatebox{90}{$p.\varMediators \ni u$} &
\rotatebox{90}{$u.\varBusyTokens(p)$} &
\rotatebox{90}{$u.\varChannel \ni p$} &
\rotatebox{90}{$(p, u)$-link} &
\rotatebox{90}{$(u, p)$-link}
\\
\midrule
\RowNumber{}\label{row:0-x-x-bot-x} &
$0$ &
\xmark &
\xmark &
$\bot$ &
\xmark &
$\langle \msgShortB^{*} \rangle$ &
$\langle
\msgShortNC^{?} ( \msgShortC \cdot \msgShortNC )^{*} \msgShortC^{?}
\rangle$
\\
\RowNumber{}\label{row:0-x-x-0-x-BR} &
$0$ &
\xmark &
\xmark &
$0$ &
\xmark &
$\langle \msgShortB^{*} \rangle$ &
$\langle
\msgShortBA
\cdot
\msgShortNC^{?} ( \msgShortC \cdot \msgShortNC )^{*} \msgShortC^{?}
\rangle$
\\
\RowNumber{}\label{row:0-x-x-0-x-noBR} &
$0$ &
\xmark &
\xmark &
$0$ &
\xmark &
$\langle
\msgShortB^{*} \cdot \msgShortNB \cdot \msgShortTU^{*} \cdot \msgShortB^{*}
\rangle$ &
$\langle \rangle$
\\
\RowNumber{}\label{row:0-x-x-g-x} &
$0$ &
\xmark &
\xmark &
$> 0$ &
\xmark &
$\langle \msgShortB^{*} \cdot \msgShortNB \cdot \msgShortTU^{*} \rangle$ &
$\langle ( \msgShortNC \cdot \msgShortC )^{*} \msgShortNC^{?} \rangle$
\\
\RowNumber{}\label{row:0-x-x-g-v} &
$0$ &
\xmark &
\xmark &
$> 0$ &
\vmark &
$\langle \msgShortB^{*} \cdot \msgShortNB \cdot \msgShortTU^{*} \rangle$ &
$\langle ( \msgShortC \cdot \msgShortNC )^{*} \msgShortC^{?} \rangle$
\\
\RowNumber{}\label{row:g-x-x-bot-x} &
$> 0$ &
\xmark &
\xmark &
$\bot$ &
\xmark &
$\langle \msgShortB^{+} \rangle$ &
$\langle
\msgShortNC^{?} ( \msgShortC \cdot \msgShortNC )^{*} \msgShortC^{?}
\rangle$
\\
\RowNumber{}\label{row:g-x-x-0-x-BR} &
$> 0$ &
\xmark &
\xmark &
$0$ &
\xmark &
$\langle \msgShortB^{*} \rangle$ &
$\langle
\msgShortBA
\cdot
\msgShortNC^{?} ( \msgShortC \cdot \msgShortNC )^{*} \msgShortC^{?}
\rangle$
\\
\RowNumber{}\label{row:g-x-x-0-x-noBR} &
$> 0$ &
\xmark &
\xmark &
$0$ &
\xmark &
$\langle
\msgShortB^{+} \cdot \msgShortNB \cdot \msgShortTU^{*} \cdot \msgShortB^{*}
\rangle$ &
$\langle \rangle$
\\
\RowNumber{}\label{row:g-x-x-g-x} &
$> 0$ &
\xmark &
\xmark &
$> 0$ &
\xmark &
$\langle \msgShortB^{+} \cdot \msgShortNB \cdot \msgShortTU^{*} \rangle$ &
$\langle ( \msgShortNC \cdot \msgShortC )^{*} \msgShortNC^{?} \rangle$
\\
\RowNumber{}\label{row:g-x-x-g-v} &
$> 0$ &
\xmark &
\xmark &
$> 0$ &
\vmark &
$\langle \msgShortB^{+} \cdot \msgShortNB \cdot \msgShortTU^{*} \rangle$ &
$\langle ( \msgShortC \cdot \msgShortNC )^{*} \msgShortC^{?} \rangle$
\\
\RowNumber{}\label{row:g-v-x-0-x} &
$> 0$ &
\vmark &
\xmark &
$0$ &
\xmark &
$\langle \msgShortTU^{+} \cdot \msgShortB^{*} \rangle$ &
$\langle \rangle$
\\
\RowNumber{}\label{row:g-v-x-g-x} &
$> 0$ &
\vmark &
\xmark &
$> 0$ &
\xmark &
$\langle \msgShortTU^{*} \rangle$ &
$\langle ( \msgShortNC \cdot \msgShortC )^{*} \rangle$
\\
\RowNumber{}\label{row:g-v-x-g-v} &
$> 0$ &
\vmark &
\xmark &
$> 0$ &
\vmark &
$\langle \msgShortTU^{*} \rangle$ &
$\langle \msgShortC ( \msgShortNC \cdot \msgShortC )^{*} \rangle$
\\
\RowNumber{}\label{row:g-v-v-g-x} &
$> 0$ &
\vmark &
\vmark &
$> 0$ &
\xmark &
$\langle \msgShortTU^{*} \rangle$ &
$\langle \msgShortNC ( \msgShortC \cdot \msgShortNC )^{*} \rangle$
\\
\RowNumber{}\label{row:g-v-v-g-v} &
$> 0$ &
\vmark &
\vmark &
$> 0$ &
\vmark &
$\langle \msgShortTU^{*} \rangle$ &
$\langle ( \msgShortC \cdot \msgShortNC )^{*} \rangle$
\\
\bottomrule
\end{tabular}
\caption{\label{tab:legal-configurations}%
Reachable configurations of the local states of a primary node
$p \in P$,
a utility node
$u \in U(p)$,
and the edge between $p$ and $u$ (represented as two anti-parallel directed
links) under \AlgTF{}'s channel layer.
The messages
$\msgBusy$,
$\msgTokensUpdate$,
$\msgNotBusy$,
$\msgBusyAck$,
$\msgChannel$,
and
$\msgNoChannel$
are abbreviated by
$\msgShortB$,
$\msgShortTU$,
$\msgShortNB$,
$\msgShortBA$,
$\msgShortC$,
and
$\msgShortNC$,
respectively.
The link content is encoded (using standard regular expression notation) as a
string of messages, ordered from left to right in decreasing sending time
(so that the left-most message in the string is the newest one).
For clarity of the exposition, variable $u.\varDiff(p)$ and the
$\msgChannelAck$ messages are omitted here.
}
\end{table}

\begin{table}[ht]
\centering
\begin{tabular}{@{}c|ccc|cccccc@{}}
\toprule
&
\rotatebox{90}{$p.\varTokens$ increases within $\Integers_{> 0}$} &
\rotatebox{90}{$p.\varTokens = 0$ is toggled} &
\rotatebox{90}{'$u$ mediates $p$' is toggled} &
\rotatebox{90}{$u$ receives $\msgBusy$} &
\rotatebox{90}{$u$ receives $\msgTokensUpdate$} &
\rotatebox{90}{$u$ receives $\msgNotBusy$} &
\rotatebox{90}{$p$ receives $\msgBusyAck$} &
\rotatebox{90}{$p$ receives $\msgChannel$} &
\rotatebox{90}{$p$ receives $\msgNoChannel$}
\\
\midrule
(\ref{row:0-x-x-bot-x}) &
$-$ &
(\ref{row:g-x-x-bot-x}) &
$-$ &
(\ref{row:0-x-x-0-x-BR}) &
$-$ &
$-$ &
$-$ &
(\ref{row:0-x-x-bot-x}) &
(\ref{row:0-x-x-bot-x})
\\
(\ref{row:0-x-x-0-x-BR}) &
$-$ &
(\ref{row:g-x-x-0-x-BR}) &
$-$ &
(\ref{row:0-x-x-0-x-BR}) &
$-$ &
$-$ &
(\ref{row:0-x-x-0-x-noBR}) &
(\ref{row:0-x-x-0-x-BR}) &
(\ref{row:0-x-x-0-x-BR})
\\
(\ref{row:0-x-x-0-x-noBR}) &
$-$ &
(\ref{row:g-x-x-0-x-noBR}) &
$-$ &
(\ref{row:0-x-x-0-x-noBR}) &
(\ref{row:0-x-x-g-x}) &
(\ref{row:0-x-x-bot-x}) &
$-$ &
$-$ &
$-$
\\
(\ref{row:0-x-x-g-x}) &
$-$ &
(\ref{row:g-x-x-g-x}) &
(\ref{row:0-x-x-g-v}) &
$-$ &
(\ref{row:0-x-x-g-x}) &
(\ref{row:0-x-x-bot-x}) &
$-$ &
(\ref{row:0-x-x-g-x}) &
(\ref{row:0-x-x-g-x})
\\
(\ref{row:0-x-x-g-v}) &
$-$ &
(\ref{row:g-x-x-g-v}) &
(\ref{row:0-x-x-g-x}) &
$-$ &
(\ref{row:0-x-x-g-v}) &
(\ref{row:0-x-x-bot-x}) &
$-$ &
(\ref{row:0-x-x-g-v}) &
(\ref{row:0-x-x-g-v})
\\
(\ref{row:g-x-x-bot-x}) &
(\ref{row:g-x-x-bot-x}) &
(\ref{row:0-x-x-bot-x}) &
$-$ &
(\ref{row:g-x-x-0-x-BR}) &
$-$ &
$-$ &
$-$ &
(\ref{row:g-x-x-bot-x}) &
(\ref{row:g-x-x-bot-x})
\\
(\ref{row:g-x-x-0-x-BR}) &
(\ref{row:g-x-x-0-x-BR}) &
(\ref{row:0-x-x-0-x-BR}) &
$-$ &
(\ref{row:g-x-x-0-x-BR}) &
$-$ &
$-$ &
(\ref{row:g-v-x-0-x}) &
(\ref{row:g-x-x-0-x-BR}) &
(\ref{row:g-x-x-0-x-BR})
\\
(\ref{row:g-x-x-0-x-noBR}) &
(\ref{row:g-x-x-0-x-noBR}) &
(\ref{row:0-x-x-0-x-noBR}) &
$-$ &
(\ref{row:g-x-x-0-x-noBR}) &
(\ref{row:g-x-x-g-x}) &
(\ref{row:g-x-x-bot-x}) &
$-$ &
(\ref{row:g-x-x-0-x-noBR}) &
(\ref{row:g-x-x-0-x-noBR})
\\
(\ref{row:g-x-x-g-x}) &
(\ref{row:g-x-x-g-x}) &
(\ref{row:0-x-x-g-x}) &
(\ref{row:g-x-x-g-v}) &
$-$ &
(\ref{row:g-x-x-g-x}) &
(\ref{row:g-x-x-bot-x}) &
$-$ &
(\ref{row:g-x-x-g-x}) &
(\ref{row:g-x-x-g-x})
\\
(\ref{row:g-x-x-g-v}) &
(\ref{row:g-x-x-g-v}) &
(\ref{row:0-x-x-g-v}) &
(\ref{row:g-x-x-g-x}) &
$-$ &
(\ref{row:g-x-x-g-v}) &
(\ref{row:g-x-x-bot-x}) &
$-$ &
(\ref{row:g-x-x-g-v}) &
(\ref{row:g-x-x-g-v})
\\
(\ref{row:g-v-x-0-x}) &
(\ref{row:g-v-x-0-x}) &
(\ref{row:0-x-x-0-x-noBR}) &
$-$ &
(\ref{row:g-v-x-0-x}) &
(\ref{row:g-v-x-g-x}) &
$-$ &
$-$ &
$-$ &
$-$
\\
(\ref{row:g-v-x-g-x}) &
(\ref{row:g-v-x-g-x}) &
(\ref{row:0-x-x-g-x}) &
(\ref{row:g-v-x-g-v}) &
$-$ &
(\ref{row:g-v-x-g-x}) &
$-$ &
$-$ &
(\ref{row:g-v-v-g-x}) &
$-$
\\
(\ref{row:g-v-x-g-v}) &
(\ref{row:g-v-x-g-v}) &
(\ref{row:0-x-x-g-v}) &
(\ref{row:g-v-x-g-x}) &
$-$ &
(\ref{row:g-v-x-g-v}) &
$-$ &
$-$ &
(\ref{row:g-v-v-g-v}) &
$-$
\\
(\ref{row:g-v-v-g-x}) &
(\ref{row:g-v-v-g-x}) &
(\ref{row:0-x-x-g-x}) &
(\ref{row:g-v-v-g-v}) &
$-$ &
(\ref{row:g-v-v-g-x}) &
$-$ &
$-$ &
$-$ &
(\ref{row:g-v-x-g-x})
\\
(\ref{row:g-v-v-g-v}) &
(\ref{row:g-v-v-g-v}) &
(\ref{row:0-x-x-g-v}) &
(\ref{row:g-v-v-g-x}) &
$-$ &
(\ref{row:g-v-v-g-v}) &
$-$ &
$-$ &
$-$ &
(\ref{row:g-v-x-g-v})
\\
\bottomrule
\end{tabular}
\caption{\label{tab:legal-configurations-incidence-matrix}%
The incidence matrix of (the logical digraph defined over) the configurations
of Table~\ref{tab:legal-configurations}.
Each column corresponds to a local event and the cell in row $i$ and column
$j$ presents the configuration obtained from configuration $i$ upon event
$j$ (a '$-$' sign indicates that event $j$ is inapplicable for configuration
$i$).
For example, if the current configuration of the local states of a primary
node
$p \in P$,
a utility node
$u \in U$,
and the edge between $p$ and $u$ is (\ref{row:0-x-x-bot-x}) and
$p.\varTokens$ changes from
$p.\varTokens = 0$
to
$p.\varTokens > 0$
(regarded as a local event that ``triggers'' $p$ to act),
then the resulting configuration is (\ref{row:g-x-x-bot-x}).
On the other hand, if the current configuration is (\ref{row:g-x-x-g-v}) and
$u$ receives $\msgNotBusy$, then the resulting configuration is
(\ref{row:g-x-x-bot-x}).
Notice that in some scenarios, multiple events may occur in succession.
}
\end{table}

Formally, \Grnt{}~\ref{grnt:algorithm:interface:operational-implies-busy} follows
immediately from Table~\ref{tab:legal-configurations}. Informally, this
guarantee says that a primary node can have an active channel only while it is
busy. Intuitively, this is exactly the channel-layer policy: when a primary
node becomes non-busy, it removes all utility nodes from
$p.\varMediators$,
so it cannot remain operational after that point.

Formally, \Grnt{}~\ref{grnt:algorithm:interface:operatioanl-channel-generation} is
established by an inspection of
Tables~\ref{tab:legal-configurations} and
\ref{tab:legal-configurations-incidence-matrix}, together with the facts that
every message reaches its destination within $1$ time unit and that
\Lem{}~\ref{lem:pu-graph} ensures that for every two primary nodes
$p, p' \in P$,
the intersection
$U(p) \cap U(p')$
contains at least one non-fragile utility node.
Informally, this guarantee says that if at least two primary nodes stay busy
for $4$ time units, then the channel layer creates some operational channel by
that time.
Intuitively, each busy primary node informs its utility neighbors that it is
busy, receives their acknowledgment, and then updates them with its current
token count; once some non-fragile utility node knows of two busy primaries
with positive token counts, it creates a channel between two such primaries.
This is essentially the short primary-utility handshake:
$\msgBusy \rightarrow \msgBusyAck \rightarrow \msgTokensUpdate \rightarrow \msgChannel$.

Formally, \Grnt{}~\ref{grnt:algorithm:interface:operational-channel-continues} is
established by an inspection of
Tables~\ref{tab:legal-configurations} and
\ref{tab:legal-configurations-incidence-matrix}.
Informally, this guarantee says that a channel remains operational exactly
while both of its endpoints remain busy, and that once one endpoint becomes
non-busy, the other endpoint removes the channel within $2$ time units.
Intuitively, this is exactly the channel-layer policy:
when one endpoint stops being busy, it sends
$\msgNotBusy$
to the mediator, which releases the channel and notifies the other endpoint by
$\msgNoChannel$;
since each of these two messages is delivered within $1$ time unit, the
channel disappears from the other endpoint within $2$ time units.

Formally, \Grnt{}~\ref{grnt:algorithm:interface:relayed-messages} is
established by an inspection of
Tables~\ref{tab:legal-configurations} and
\ref{tab:legal-configurations-incidence-matrix}, together with the fact that a
utility node $u$ relays a message received from a primary node $p$ only if
$u.\varDiff(p)=0$, and the assumption that each directed link delivers
messages in FIFO order.
Informally, this guarantee says that relayed messages are tied to the correct
current channel: a message received over a channel really originated from the
current peer of that channel, and a message sent over a channel either reaches
that peer soon or becomes irrelevant because the peer stopped being busy and
the channel is being dismantled.
Intuitively, the role of
$u.\varDiff(p)$
is to prevent ``cross-talk'' between different incarnations of channels
mediated by the same utility node, while FIFO delivery ensures that the
control messages defining a channel's lifetime appear in the correct order
relative to the relayed messages themselves.
Whenever $u$ sends
$\msgChannel$
to $p$, it increments
$u.\varDiff(p)$,
and only after receiving the matching
$\msgChannelAck$
does it allow relayed messages from $p$ to pass through.
Because messages on the
$(p,u)$-link
are delivered in FIFO order, any stale relayed message sent by $p$ before it
learned about the current channel must arrive at $u$ before the corresponding
$\msgChannelAck$,
and is therefore screened out.
Dually, when the channel is released, the relevant
$\msgNoChannel$
/
$\msgNotBusy$
control messages also delimit the end of that channel in FIFO order, so a
relayed message cannot slip past them and be interpreted as belonging to some
later incarnation of the same connection.
Thus, the control messages act as clean boundaries for the lifetime of a
channel, and any relayed message is either delivered within those boundaries
to the current peer or discarded as irrelevant.

\Grnt{}~\ref{grnt:algorithm:interface:degree-bound} follows immediately from
\Lem{}~\ref{lem:pu-graph}.
To prove \Grnt{}~\ref{grnt:algorithm:interface:message-complexity}, let
$\ell$ (resp, $\ell'$) be the total number of tokens injected into the system
throughout the execution, including (resp., excluding) fakely injected ones;
recall (see the proof of \Lem{}~\ref{lem:analysis:message-load:blaming-tokens})
that
$\ell \leq 2 \ell'$.
Consider
a primary node
$p \in P$
and a utility node
$u \in U(p)$
and observe that
(1)
each $\msgBusyAck$ message sent from $u$ to $p$ can be (injectively)
attributed to a $\msgBusy$ message sent from $p$ to $u$;
(2)
each $\msgNotBusy$ message sent from $p$ to $u$ can be (injectively)
attributed to a $\msgBusy$ message sent from $p$ to $u$;
(3)
each $\msgNoChannel$ message sent from $u$ to $p$ can be (injectively)
attributed to a $\msgChannel$ message sent from $u$ to $p$;
and
(4)
each $\msgChannelAck$ message sent from $p$ to $u$ can be (injectively)
attributed to a $\msgChannel$ message sent from $u$ to $p$.
Therefore, it suffices to bound the number of $\msgBusy$, $\msgTokensUpdate$,
and $\msgChannel$ messages with respect to $\ell$;
we do so for each message type separately based on the following two
observations.

\begin{observation}
\label{obs:channel-layer-algorithm:non-busy-to-busy}
$\sum_{p \in P}
\left| \left\{ t > 0 : p \in \Busy^{t+} - \Busy^{t} \right\} \right|
\leq
\ell$.
\end{observation}
\begin{proof}
Follows from the definition of busy nodes by
\Obs{}~\ref{obs:analysis:transport-in-transit}.
\end{proof}

\begin{observation}
\label{obs:channel-layer-algorithm:busy-to-non-busy}
$\sum_{p \in P}
\left| \left\{ t > 0 : p \in \Busy^{t} - \Busy^{t+} \right\} \right|
\leq
\ell$.
\end{observation}
\begin{proof}
Follows immediately from
\Obs{}~\ref{obs:channel-layer-algorithm:non-busy-to-busy}.
\end{proof}

Bounding the number of $\msgBusy$ messages is now straightforward:
Since a primary node
$p \in P$
sends $\msgBusy$ messages only if it becomes busy, in which case, $p$ sends
$\msgBusy$ messages to, and only to, the utility nodes in $U(p)$, we conclude,
by \Lem{}~\ref{lem:pu-graph} and
\Obs{}~\ref{obs:channel-layer-algorithm:non-busy-to-busy} that at most
$O (\ell \cdot \sqrt{n \log n})$
$\msgBusy$ messages are sent throughout the execution.

To bound the number of $\msgTokensUpdate$ messages, notice that these messages
are sent by a primary node
$p \in P$
only when either
(1)
tokens (at least one) are injected into $p$;
or
(2)
$p$ receives a $\msgTransport$ message $\mu$ from a peer
$p' \in P$
over a
$\{ p, p' \}$-channel
(relayed by the channel's mediator).
In both cases, $p$ sends $\msgTokensUpdate$ messages to, and only to, the
$O (\sqrt{n \log n})$
utility nodes in $U(p)$, so it remains to show the number of $\msgTransport$
messages sent throughout the execution is at most $\ell$.
This follows from
\Obs{}~\ref{obs:channel-layer-algorithm:busy-to-non-busy} since node $p'$
can send a $\msgTransport$ message at time
$t > 0$
only if $p$ becomes non-busy at time $t$.

Finally, for $\msgChannel$ messages, observe that a utility node
$u \in U$
sends such messages in pairs and that this happens at time
$t > 0$
only if $u$ receives either a $\msgTokensUpdate$ message or a $\msgNotBusy$
message at time $t$.
Therefore, the bound on the number of $\msgChannel$ messages sent throughout
the execution follows from the bounds on the numbers of $\msgTokensUpdate$ and
$\msgNotBusy$ messages sent throughout the execution that we have already
established.

\section{Analysis of the Algorithm --- Reaction Time}
\label{sec:reaction-time}
Recall our proof of the liveness condition in
\Sect{}~\ref{sec:analysis:safety-liveness}.
This proof relies on a potential function
$\psi : \Reals_{> 0} \rightarrow \Integers_{\geq 0}$,
arguing that if the system contains at least $\sigma$ tokens, then as long as
there are no team formation operations, the following conditions are
satisfied:
(1)
$\psi(t)$ is non-decreasing;
(2)
$\psi(t)$ cannot increase beyond
$O (\sigma n)$;
and
(3)
$\psi(t)$ increases in positive rate.
In the current section, we devise a different potential function
$\hat{\psi} : \Reals_{> 0} \rightarrow \Integers_{\geq 0}$
and argue that it still satisfies condition (1), while replacing conditions
(2) and (3) with the following conditions:
(2')
$\hat{\psi}(t)$ cannot increase beyond
$O (\sigma)$;
and
(3')
$\hat{\psi}(t)$ increases in a constant rate.

Consider time
$t_{0} > 0$
and assume that the system contains at least $\sigma$ tokens at time $t_{0}$.
For time
$t \geq t_{0}$
and primary node
$p \in P$,
define the operator
$R^{t}(p)$
as in \Sect{}~\ref{sec:analysis:safety-liveness}.
If
$| \Busy^{t} | = \{ p \}$,
then
$p.\varTokens^{t} + R^{t}(p) \geq \sigma$,
hence a team is guaranteed to be formed by time
$t + O (1)$.
So, assume that
$| \Busy^{t} | \geq 2$
and let
$
M_{1}^{t}
\, = \,
\MaxOne_{p \in P} \, p.\varTokens^{t} + R^{t}(p)$
and
$M_{2}^{t}
\, = \,
\MaxTwo_{p \in P} \, p.\varTokens^{t} + R^{t}(p)$,
where the $\MaxTwo$ operator returns the second largest value in the value
multiset (in particular, if the maximum is realized by multiple items, then
$\MaxTwo = \MaxOne$).
The aforementioned notation allows us to define the potential function
$\hat{\psi} : \Reals_{> 0} \rightarrow \Integers_{\geq 0}$
as
$\hat{\psi}(t)
\, = \,
M_{1}^{t} + M_{2}^{t}$.

Notice that if
$\hat{\psi}(t) \geq 2 \sigma$,
then
$p.\varTokens^{t} + R^{t}(p) \geq \sigma$
for some
$p \in P$,
hence a team is guaranteed to be formed by time
$t + O (1)$.
As in \Sect{}~\ref{sec:analysis:safety-liveness},
\Obs{}~\ref{obs:analysis:transport-in-transit} ensures that if no teams
are formed during the time interval
$I = [t_{0}, \cdot] \subset \Reals_{> 0}$,
then the function $\hat{\psi}(t)$ is non-decreasing in $I$.

It remains to prove that $\hat{\psi}(t)$ increases in a constant rate.
To this end, we say that a primary node
$p \in P$
is a \emph{maximizer} at time
$t \geq t_{0}$
if $p$ realizes $M_{i}^{t}$ for some
$i \in \{ 1, 2 \}$,
denoting the set of maximizers at time $t$ by $\M^{t}$.

\begin{observation}
\label{observation:reaction-time:no-longer-maximizer}
For every
$p \in P$
and
$t \geq t_{0}$,
if
$p \in \M^{t} - \M^{t+}$
and no team is formed at time $t$, then
$\hat{\psi}(t+) = \hat{\psi}(t) + 1$.
\end{observation}
\begin{proof}
Follows directly from the definition of $\hat{\psi}(\cdot)$.
\end{proof}

\begin{observation}
\label{observation:reaction-time:channel-between-maximizers}
For every
$p, p' \in P$,
$u \in U(p) \cap U(p')$,
and
$t \geq t_{0}$,
if
$p, p' \in \M^{t}$
and $u$ creates a
$\{ p, p' \}$-channel
at time $t$ then with probability at least
$1 / 16$,
if no team is formed by time
$t + 24$,
then
$\hat{\psi}(t + 24) > \hat{\psi}(t)$.
\end{observation}
\begin{proof}
\Obs{}~\ref{obs:analysis:phase-length} and
\Lem{}~\ref{lem:analysis:operational-channel} ensure that with probability at
least
$1 / 16$,
at least one of $p$ and $p'$ retires by time
$t + 24$.
The assertion follows by
\Obs{}~\ref{observation:reaction-time:channel-between-maximizers}.
\end{proof}

Fix time
$\bar{t} \geq t_{0}$
and assume that no teams have been formed during the time interval
$[t_{0}, \bar{t}]$
and that
$| \Busy^{\bar{t}} | \geq 2$
which implies that
$| \M^{\bar{t}} | \geq 2$.
By \Lem{}~\ref{lem:pu-graph}, there exists a non-fragile
utility node
$u \in U$
such that
$| P(u) \cap \M^{\bar{t}} | \geq 2$.
By \Obs{}~\ref{observation:reaction-time:no-longer-maximizer},
if no team is formed by time
$t > \bar{t}$
and
$| P(u) \cap \M^{\bar{t}} | < 2$,
then
$\hat{\psi}(t) > \hat{\psi}(\bar{t})$.

For time
$t \geq \bar{t} + 5$,
let
$Q^{t}
=
\left\{ p \in P : u.\varBusyTokens^{t}(p) \in \Integers_{> 0} \right\}$.
Assuming that no team is formed by time $t$ and that
$\hat{\psi}(t) = \hat{\psi}(\bar{t})$,
the design of the channel layer ensures that 
$| Q^{t} \cap \M^{t} | \geq 2$
and that
$p.\varTokens^{t} = u.\varBusyTokens^{t}(p)$
for every
$p \in Q^{t} \cap \M^{t}$
(here, we exploit the fact that every $\msgTransport$ message reaches its
destination in $2$ time units).
\Obs{}~\ref{obs:analysis:phase-length} and
\Lem{}~\ref{lem:analysis:operational-channel} ensure that if
$t \geq \bar{t} + 5 + (3 \cdot 8) = \bar{t} + 29$,
then with probability at least
$1 / 16$,
node $u$ creates a
$\{ p_{1}, p_{2} \}$-channel
at time
$\bar{t} + 5 \leq t' \leq \bar{t} + 29$
for some
$p_{1}, p_{2} \in \M^{t'}$.
Combined with
\Obs{}~\ref{observation:reaction-time:channel-between-maximizers}, we
obtain the following lemma.

\begin{lemma}
\label{lem:reaction-time:constant-rate}
For every
$\bar{t} \geq t_{0}$,
with probability at least
$1 / 16^{2} = 1 / 256$,
if no team is formed by time
$\bar{t} + 29 + 24 = \bar{t} + 53$,
then
$\hat{\psi}(\bar{t} + 53) > \hat{\psi}(\bar{t})$.
\end{lemma}

The following theorem follows by standard probabilistic arguments.

\begin{theorem}
\label{thm:reaction-time}
The reaction time of \AlgTF{} is
$O (\sigma + \log n)$
whp.\footnote{%
\label{footnote:asynchronous-time-complexity}
We note that some researchers take an alternative approach to the definition
of asynchronous time complexity in message passing, one that differs from the
definition introduced in \Sect{}~\ref{sec:preliminaries:model}.
Specifically, under the definition adopted in the current paper, it is assumed
that an activated node $v$ can send a message $\mu$ to a neighbor $v'$ at
(activation) time
$t > 0$
regardless of the state of the
$(v, v')$-link
at time $t$ and that message $\mu$ is guaranteed to be delivered to $v'$
(assuming that $v'$ is not faulty) at or before time
$t + 1$.
In contrast, under the alternative approach, message $\mu$ ``waits'' at $v$
until the first time
$t' \geq t$
at which the
$(v, v')$-link
is clear and only then starts its journey to $v'$ so that it is guaranteed to
be delivered to $v'$ at time
$t' + 1$.
The upper bound in this theorem may or may not hold under the alternative
approach.
}
\end{theorem}

\section{Lower Bound}
\label{sec:lower-bound}
In this section, we prove that any algorithm that solves the TF problem whp has message complexity
$\Omega\left(\max\left\{\sqrt{n \log n},\, \sqrt{n \sigma}\right\}\right)$,
thereby establishing the lower bound stated informally in \Thm{}~\ref{thm:lowerbound-informal}.
The bound holds for arbitrary 
$\sigma \geq 2$
(the case
$\sigma = 1$
is trivial, requiring no message exchange), where $\sigma$ may depend on $n$, subject to the restriction that
$\sigma \leq n / 2$.

Specifically, we analyze the case where exactly
$\sigma$
tokens are injected into the system, and an algorithm is considered successful if it forms a team of
$\sigma$
tokens.
The tokens are injected simultaneously at the start of the execution, and each token is placed at a distinct node.
Our lower bound is established in the synchronous LOCAL model, which also implies the same bound for the asynchronous CONGEST model.
We assume that nodes are initially unaware of their neighbors' identities and can distinguish among their
$n - 1$
neighbors only via local port numbers.
Additionally, we assume that nodes have access to private unbiased random coins, rather than a shared source of randomness.

Our bounds are formally stated in the following propositions.
Let $A$ be an algorithm for the TF problem that succeeds whp and sends at most $f(n)$ messages in expectation.

\begin{proposition}
\label{prp:lowerbound-log}
If
$2 \leq \sigma \leq O(\sqrt{n \log n})$, then
$f(n) \geq \Omega(\sqrt{n \log n})$.
\end{proposition}

\begin{proposition}
\label{prp:lowerbound-sigma}
If
$2 \leq \sigma \leq  n/2$,
then
$f(n) \geq \Omega(\sqrt{n \sigma})$.
\end{proposition}

\begin{corollary}
\label{cor:lowerbound-max}
If
$2 \leq \sigma \leq n/2$,
then
\[
f(n) \geq \Omega\left(\max\left\{\sqrt{n \log n},\, \sqrt{n \sigma}\right\}\right).
\]
\end{corollary}

A lower bound of 
$\Omega(\sqrt{n})$
messages for the leader-election problem (for algorithms that succeed with constant probability) was established in~\cite{DBLP:journals/tcs/KuttenPP0T15}.
When
$\sigma$
is logarithmic in $n$, this bound applies directly to our setting via a straightforward reduction from leader election (see \Sect{}~\ref{sec:intro:applications}).
For a broader range of 
$\sigma$
values—and for stronger bounds within that range—additional work is required.
We build on the approach of~\cite{DBLP:journals/tcs/KuttenPP0T15}, adapting and
generalizing it to the TF problem and to algorithms that succeed whp, to
obtain the bounds stated in \Prop{}\ \ref{prp:lowerbound-log} and
\ref{prp:lowerbound-sigma}.

We consider the \emph{influence graph}—the undirected subgraph of the complete network induced by those edges over which at least one message has been sent so far during the execution.
A connected component of the influence graph that contains no tokens is called \emph{trivial}.
For all $\sigma$ tokens to be gathered at a single node (and hence for a team to form),
the influence graph eventually must have exactly one non-trivial connected component:
if the tokens reside in different connected components, then some of them never communicate with the others, making it impossible to collect all $\sigma$ tokens at one node.
We formalize this in the following observation.

\begin{observation}
\label{obs:lowerbound-necessary-condition}
In order for algorithm $A$ to succeed, the influence graph must contain exactly one non-trivial connected component that includes all tokens by the end of the execution.
\end{observation}

Both proofs of \Prop{}\ \ref{prp:lowerbound-log} and \ref{prp:lowerbound-sigma} rely on \Obs{}~\ref{obs:lowerbound-necessary-condition}, but use different means of analyzing the probability of this necessary condition for success.
The remainder of this section is dedicated to laying the common groundwork for
both proofs, with the specific arguments completed in \Sect{}\
\ref{subsec:lowerbound-log-proof} and \ref{subsec:lowerbound-sigma-proof}.

Formally, assume that the communication graph
$G = (V,E)$
is complete and that, for each node $v$, an adversary independently assigns its port numbers to neighbors according to a uniformly random permutation.
Let
$\sigma \colon \mathbb{N} \to \mathbb{N}$
be a function such that 
$
2 \leq \sigma(n) \leq n/2
$
for sufficiently large $n$
(for the proof of \Prop{}~\ref{prp:lowerbound-log} we will further restrict this to
$\sigma \leq O(\sqrt{n \log n})$).
For clarity, we typically omit the explicit dependence on $n$ and simply write
$\sigma$, treating it as an integer.
Let
$A = A(\sigma)$
be an algorithm that  
(1) sends at most $f(n)$ messages in expectation; and  
(2) whp eventually succeeds in forming a team after exactly $\sigma$ tokens are injected into the system.

In the proofs of our lower bounds, we assume that
$2 \leq \sigma \leq O(\sqrt{n \log n})$
for \Prop{}~\ref{prp:lowerbound-log}
(resp., $2 \leq \sigma \leq n/2$ for \Prop{}~\ref{prp:lowerbound-sigma}) and assume by contradiction that
$f(n) \leq o(\sqrt{n \log n})$ (resp., $f(n) \leq o(\sqrt{n \sigma})$).
These assumptions imply the following observation, which we use throughout the remainder of the section.

\begin{observation}
\label{obs:lowerbound-small-sigma-and-fn}
For sufficiently large $n$,
we have that
$2 \leq \sigma \leq n/2$
and
$f(n) \leq o(n)$.
\end{observation}

Let $M$ be the event that $A$ sends at most $2f(n)$ messages.  
By Markov's inequality, we have
$\Pr[M] \ge 1/2$.  
Throughout the remainder of this section we assume that $M$ occurs. 
Conditioned on $M$, only finitely many messages are sent, so there is a round after which no further communication takes place. 
Let $S$ denote the event that $A$ succeeds, i.e., forms a team; we assume $S$ occurs whp, meaning
$\Pr[S] \ge 1 - n^{-c}$
for some constant $c > 0$.

Although the original setting involves a fully distributed algorithm, we give the algorithm designer even more power by introducing a \emph{central entity} (henceforth, CE).
The CE observes a global view of the system (described in detail below) and, in every round, decides for each node whether to send a message and, if so, through which port.
Because we (soon) define the CE to have at least as much information as any individual node, any lower bound established for this stronger model also applies to the original distributed setting.

We now formalize the global view available to the CE.
For every node $v \in V$, the CE knows the entire local state of $v$, including (but not limited to) whether $v$ currently holds a token and the identities of all ports that $v$ has already identified.
This alone guarantees that the CE has at least as much information as any individual node in the baseline model.
In addition, whenever a message is sent over an edge
$e=(u,v)\in E$,
the CE immediately learns, for each endpoint
$x \in \{u,v\}$,
which port $p$ of $x$ leads to the other endpoint in
$\{u,v\} - \{x\}$.
When this happens, we say that the CE \emph{exposed} port $p$, i.e., by sending a message over edge $e$, the CE identified where $p$ leads.

The edges that have carried at least one message induce an undirected subgraph, called the \emph{influence graph}.
Each connected component $C$ of the influence graph is referred to as an \emph{influence component}.
To strengthen the CE's knowledge even further, we let it learn, for every node
$u \in C$ and every port $p$ of $u$,
whether $p$ is \emph{internal} to $C$ (i.e., the port leads to another node
$v \in C$) and, if so, exactly which node $p$ connects $u$ to.
Informally, we observe that the influence components coincide with the equivalence classes of the relation
$u \sim v$ if and only if $u$ and~$v$ are connected in the influence graph.

Because the CE knows, for each influence component, exactly which ports are internal, we can regard these ports as belonging to the component itself, rather than to the individual nodes.
An influence component whose nodes hold no token is called \emph{trivial}.
Subject to the message constraints guaranteed by event~$M$, the goal of the CE is to \emph{merge} influence components (by exposing their ports that hopefully lead to different components) until a single non-trivial component contains all $\sigma$ tokens.
By \Obs{}~\ref{obs:lowerbound-necessary-condition}, if the CE fails to do so, then algorithm $A$ must fail and the success event~$S$ cannot occur.

Without loss of generality, we assume that in every round the CE sends exactly one message and that its sole decision is which influence component is \emph{selected} to originate that message.
This is justified because, once a message is sent, the CE gains additional information; future decisions can therefore be based on this richer view, making it never disadvantageous to send messages sequentially.
We also assume that messages are never sent over internal ports of a component, as the CE already possesses this information (potentially gained by using an internal port) about every node within such a component.

Already‐exposed ports of an influence component are internal by definition.
Ports that have not yet been exposed are not internal and may either lead to a different trivial or non-trivial component.
We say that round~$r$ is \emph{successful} if the selected (to expose a port in that round) component $C$ \emph{hits} another component~$C'$, where both $C$ and $C'$ are non-trivial.
Define random variable $H_r$ as an indicator to the event of a successful round $r$.

Let $N_r$ denote the sum of the sizes of all non-trivial influence components at the start of round~$r$.
Consider some round~$r$ and some non-trivial component $C$ selected by the CE during round $r$.
There are $N_r - |C|$ nodes belonging to other non-trivial components.
Since we condition on $M$, at most 
$2f(n)$
ports have been exposed in total, so component~$C$ has at least
$n - 2f(n)$
remaining unidentified ports.
As each node's port assignment is an independent uniform permutation, the probability that $C$ hits another non-trivial component, an event denoted by
$H_r^C$,
is at most 
$\frac{N_r - |C|}{n - 2f(n)}$.
Due to \Obs{}~\ref{obs:lowerbound-small-sigma-and-fn}, we have that
$n - 2f(n) \geq 2n/3$
for sufficiently large $n$, and therefore the probability for
$H_r^C$ is at most
$\frac{3N_r - 3|C|}{2n}$.

Note that if round~$r$ is successful, then
$N_{r+1} = N_r$.
If round ~$r$ is not successful and a non-trivial component is selected (in round $r$) by the CE, then a port leading to some trivial component $C''$ was exposed, and 
$N_{r+1} = N_r + |C''|$.
If round ~$r$ is not successful, but a trivial component $C''$ is selected, then
$N_{r+1} = N_r + |C''|$
if the newly-exposed port leads to a non-trivial component, and otherwise
$N_{r+1} = N_r$.
Even if $C''$ is later merged into a non-trivial component by another message, the overall increase to $N_r$ is only by the number of messages ``invested'' in $C''$ plus one.
Therefore, we assume without loss of generality that trivial components $C''$ are never selected (therefore, trivial components are always of size $1$), as doing so does not increase the probability of a successful (current or later) round and provides no advantage to the CE.

However, since $N_r$ still depends on previous rounds, the probability of a hit in a given round depends on the outcomes of previous rounds.
In order to ensure that
$N_{r+1} = N_r + 1$
always holds and keep the
$H_r$ indicators independent, 
we adopt the following convention:
in the case of a successful round where two non-trivial components $C$ and $C'$ merge, 
we additionally merge an arbitrary trivial component $C''$ into the result (which only strengthens the CE).
Such a trivial component $C''$ always exists because, conditioned on $M$, we have 
$N_r \leq 2f(n) + \sigma \leq 2n/3$ 
for sufficiently large $n$, due to \Obs{}~\ref{obs:lowerbound-small-sigma-and-fn}.

We further upper bound the probability of 
$H_r^C$
by
$p := p(n) = \frac{3f(n) + 3\sigma/2}{n}$.
This bound holds since we are conditioning on $M$, 
which implies that
$N_r \leq 2f(n) + \sigma$.
Note that
$p \leq o(1) + 3/4$
for sufficiently large $n$ (see \Obs{}~\ref{obs:lowerbound-small-sigma-and-fn}), therefore,
$p  \in (0,1)$.
Therefore, the $H_r$ indicators are independent and each occurs with probability at most~$p$.

Let
$l$
be the number of messages sent during execution and let
$H := \sum_{r=1} ^{l} H_r$.
For
$j \in \{1, \sigma - 1\}$,
let $S_j$ denote the event that
$H \geq j$,
i.e., that there were at least $j$ successful rounds.
Note that $S_j$ is a necessary condition for the algorithm to produce a single non-trivial component that contains all tokens.
Therefore, due to \Obs{}~\ref{obs:lowerbound-necessary-condition}, event $S$ implies the occurrence of each event $S_j$.

Recall that event $M$ (resp., $S$) occurs with probability at least
$\frac{1}{2}$
(resp., $1 - n^{-c}$).
It follows that
$n^{-c} \geq \mathbb{P}[\neg S] \geq \mathbb{P}[\neg S_j \wedge M]
\geq \mathbb{P}[\neg S_j \mid M] \cdot \mathbb{P}[M]$.
We thus obtain the following observation:

\begin{observation}
\label{obs:lowerbound-Sj-bound}
For each
$j \in \{1, \sigma - 1\}$,
it holds that
$\mathbb{P}[\neg S_j \mid M] \leq 2n^{-c}$.
\end{observation}

Intuitively, our challenge now reduces to analyzing a probabilistic process rather than reasoning about algorithmic behavior,
as a consequence of the structural assumptions and simplifications made thus far.
In particular, we further increase the power of the CE by assuming that each trial $H_r$ succeeds with probability exactly~$p$ (rather than at most),
and that there are exactly $2f(n)$ such trials (rather than $l \leq 2f(n)$).
We are now ready to show each of the individual bounds.
For each of the $S_j$ events, we show that, informally, since algorithm~$A$ sends only a small number of messages, the event 
$\neg S_j$
occurs with non-negligible probability and the algorithm must fail.

\subsection{Proof of \Prop{}~\ref{prp:lowerbound-log}}
\label{subsec:lowerbound-log-proof}

Let
$2 \leq \sigma \leq O(\sqrt{n \log n})$ and assume, toward a contradiction, that $f(n) \leq o(\sqrt{n \log n})$.
Note that \Obs{}~\ref{obs:lowerbound-small-sigma-and-fn} applies.
The following lemma states that event
$\neg S_1$
(i.e., no successful round occurs) happens with non-negligible probability.

\begin{lemma}
\label{lem:lowerbound-disjoint-all-components}
For sufficiently large $n$, it holds that $\mathbb{P}[\neg S_1 \mid M] \geq (1 - p)^{2f(n)}$.
\end{lemma}

\begin{proof}
By definition, and conditioned on $M$,
event
$\neg S_1$
is equivalent to
$H = \sum_{r=1}^{2f(n)} H_r = 0$.
The $H_r$ indicators are independent Bernoulli trials with success probability $p$. Thus, the probability that all trials fail is as claimed.
\end{proof}

We are now ready to derive a contradiction.

\begin{proof}[Proof of \Prop{}~\ref{prp:lowerbound-log}]

Recalling \Lem{}~\ref{lem:lowerbound-disjoint-all-components} and \Obs{}~\ref{obs:lowerbound-Sj-bound}, we obtain 
$-c \cdot \log n 
\geq
-\log 2 + 2f(n) \cdot \log\left(1 - p\right)$.  
Since
$f(n) \leq o(\sqrt{n \log n})$
and
$\sigma \leq O(\sqrt{n \log n})$,
it holds that
$0 < p \leq o(1)$
and
$\frac{4f(n)\cdot p}{\log n} \leq o(1)$.
For sufficiently large $n$, we can therefore apply the inequality
$\log(1 - x) \geq -2x$
for
$x \in [0, 0.5]$,
yielding  
$c \leq \frac{\log 2}{\log n} + \frac{4f(n) \cdot p}{\log n} \leq o(1)$, 
which contradicts the assumption that $c > 0$ is a constant.
\end{proof}

\subsection{Proof of \Prop{}~\ref{prp:lowerbound-sigma}}
\label{subsec:lowerbound-sigma-proof}

Let
$2 \leq \sigma \leq n/2$
and assume, toward a contradiction, that
$f(n) \leq o(\sqrt{\sigma n})$.
Note that \Obs{}~\ref{obs:lowerbound-small-sigma-and-fn} applies.
By linearity of expectation, we have that
$\mu := \mathbb{E}[X \mid M] = 2f(n) \cdot p$.
Since
$f(n) \leq o(\sqrt{n \sigma})$
and 
$\sigma \leq n/2$,
we have that
$\mu \leq o(\sigma)$.
Therefore, the following holds.

\begin{observation}
\label{obs:lowerbound-mu-small}
For sufficiently large $n$, it holds that
$\mu \leq \frac{\sigma}{8}$.
\end{observation}

The following lemma states that event
$\neg S_{\sigma-1}$
(i.e., at most
$\sigma-2$
rounds are successful) happens with non-negligible probability.

\begin{lemma}
\label{lem:lowerbound-chernoff-tail}
For sufficiently large $n$, it holds that
$
\mathbb{P}[\neg S_{\sigma-1} \mid M] \geq 1-\exp\left(-\frac{\sigma}{24}\right)
$.
\end{lemma}

\begin{proof}
Let
$t := \sigma - 1 - \mu$.
From \Obs{}~\ref{obs:lowerbound-mu-small} and the assumption that
$\sigma \geq 2$,
we have
$
t \geq \frac{7\sigma}{8} -1 \geq \frac{\sigma}{8}
$.
Define
$\delta := t / \mu$.
Since
$t \geq \sigma / 8$
and $\mu \leq \sigma / 8$, it follows that $\delta \geq 1$.
Note that
$
(1 + \delta)\mu = \mu + t = \sigma - 1
$.
By Chernoff's bound, we have
\begin{align*}
\mathbb{P}[H \geq \sigma - 1 \mid M] 
&= \mathbb{P}[H \geq (1 + \delta)\mu \mid M]\\
&\leq \exp\left(-\frac{\delta^2 \mu}{2 + \delta}\right)\\
&= \exp\left(-\frac{\delta}{2 + \delta} \cdot t\right).
\end{align*}
Since
$\delta \geq 1$,
we have
$\frac{\delta}{2 + \delta} \geq \frac{1}{3}$,
and thus
\[\mathbb{P}[H \geq \sigma - 1 \mid M]
\leq \exp\left(-\frac{t}{3}\right)
\leq \exp\left(-\frac{\sigma}{24}\right).\]
\end{proof}

We are now ready to derive a contradiction.

\begin{proof}[Proof of \Prop{}~\ref{prp:lowerbound-sigma}]

Recalling \Lem{}~\ref{lem:lowerbound-chernoff-tail} and \Obs{}~\ref{obs:lowerbound-Sj-bound}, we have that
$-\frac{\sigma}{24} \geq \log\left(1 - 2n^{-c}\right)$.
Since
$0 < 2n^{-c} \leq o(1)$,
for sufficiently large $n$,
we can apply
$\log(1 - x) \geq -2x$
for
$x \in [0, 0.5]$.
It follows that
$
\log\left(1 - 2n^{-c}\right) > -4n^{-c}
$,
and hence
$
-\frac{\sigma}{24} > -4n^{-c}
$.
We conclude that
$
\sigma < 96n^{-c} \leq o(1)
$,
which contradicts the assumption that 
$\sigma \geq 2$.
\end{proof}

\section{The Trace Tree Mechanism in Detail}
\label{sec:trace-tree-mechanism-details}
%
\def\Port{pr}
For completeness, let us discuss informally the maintenance of the distributed data structure that allows algorithms (that use the TF algorithm) to use the trace tree. 
(Recall \ref{sec:intro:extra-features} the tree in the temporal graph defined by the trips of tokens in a team from their injection nodes to the node that forms the team; recall also that we refer to a temporal graph, since the trip in the network graph may not be simple.)
Maintaining this data structure is conceptually easy.\footnote{But using a lot of writing...} 
First, as an example, when would one want to use such a feature? 
Intuitively, the environment that injects tokens into a node $v$ may consist of ``users'' (e.g., Dungeons and Dragons players), such that tokens are their requests to become parts of teams. Such a user may want to be notified by $v$ at the time that their token joined a team. (The user may then feel free to head for the forest where the actual game takes place, knowing they have a ``ticket''.)

Consider the time when the environment injects a token into some node $v$. Even without assuming unique token names, we can assume that $v$ has some way to identify the token locally, to distinguish it from the other tokens held by $v$ at that time. To make the discussion more concrete, suppose that each token $T$ that is held by $v$ (either when being injected or when being received from another node) occupies a location $v.L(T)$ in $v$'s memory. Node $v$ will store some additional information at $v.L(T)$. The first piece of such information is whether $T$ is currently held by $v$ (initially ``yes'', which becomes a ``no'' if $v$ sends $T$ over one of $v$'s ports). 

We describe the implementation of the required property without using the services of the channels mechanism, but rather the primary and utility nodes directly. Each node $v$ keeps two counters for each of its ports - one counter for incoming tokens and one for outgoing ones.
For simplicity, consider first that a single token is sent in a message.
When $v$ sends a token $T$ over a port $v.\Port$, the message carrying the token also carries the value of the outgoing counter of $v.\Port$. This value is also recorded in $v.L(T)$ together with the number of the port $v.\Port$. The outgoing counter of $v.\Port$ in $v$ is then incremented.

Eventually, message carrying $T$ arrives at its destination, say node $u$, over one of $u$'s ports, say $u.\Port '$. Then $u$ places $T$ in its own memory (in $u.L(T)$  created for that purpose at that time; note the $T$ does not have an ID that is unique in the networks, however, $u$ knows it as the token just received over $u.\Port'$ and stored in the location that $u$ has just chosen as $u.L(T)$ for that $T$). Node $u$ also stores in
$u.L(T)$ 
the value of the outgoing counter of $v.\Port$ (received in the same message that has brought $T$) and the value of the incoming counter of $u.\Port '$ of $u$. (It is easy to see that those 
two counter values are in fact, equal.)
These actions are repeated whenever $T$ is sent. If a node $w$ holding $T$ forms a team that contains $T$, this too is recorded at $L(T)$ at $w$. Moreover, $w$ then starts the procedure (described below) of notifying $v$ (where $T$ was originally injected) and the environment at $v$.

In the case that multiple tokens are sent in a message, the above counters are incremented by the number of tokens in the message. Note that it is then easy for both the sender and the receiver to simulate the counter per token in the set.

\begin{observation}
Consider any time between the injection of $T$ into $v$ and the time that $T$ is included in a team formation. Then $T$ is either (a) held by some node $w$ or (b) is carried in a message to $w$.
\begin{itemize}
    
\item{Case (a)}
There exists a path $v_1, v_2, ..., v_k$ for some $k$ with the following properties.\\ 
(1) $v_1$=v.\\
(2) $T$ is held by $v_k$.\\ 
(3) For each $k>i>1$ (if the path is longer than one), node $v_i$ stores in $v_i.L(T)$ of $v_i$ (i) the port $v_i.\Port '$ over which $T$ was received at $v_i$ from $v_{i-1}$, (ii) the value of the incoming counter of $v_i.p'$ at that time, (iii) the value of the outgoing counter value of $v_{i-1}$ at the time that $v_{i-1}$ sent $T$ to $v_{i}$ in the $i-1$th sending event of $T$. 

\item{Case B} there exists such a path except that $T$ is not in $v_k$. Instead, it was sent by $v_k$ to some node $v_{k+1}$ and has not arrived yet. Moreover, the message carrying $T$ also carries the value of the outgoing counter of the port over which $v_k$ last sent $T$.
\item
The value of the outgoing counter of $v_{i-1}$ stored in $T$ is the same as that of the value of the incoming counter of $v_i$ stored with $T$.

\item
None of the nodes $v_1, ... v_k $ (or $ v_{k+1}$) is faulty.
\end{itemize}
\end{observation}

The path mentioned in the observation may not be simple in the network graph, but it is simple if considered in the temporal graph (that is, where each node has many copies, each for every time that an event happened.) The collection of the paths of the tokens that participate in the same team formation is called the trace tree of that formation.

The proof of the observation is obvious by induction on the order of the sending events of $T$, considering the actions that the sender and receiver perform for the counters and the locations $L(T)$. The last part of the observation follows from the definition of faulty nodes and quiescent times. 

Given the observation, the procedure to report each single token separately to $v$ is obvious. If $v_k$ includes $T$ in a team formation, then $v_k$ finds $v_{-1}$ in $L(T)$ in $v_k$ together with the port over which $T$ was received and the incoming counter value of the port. Thus, $v_k$ can send a notification message to $v_{k-1}$ together with the counter values. That way, $v_{k-1}$ can identify the relevant token location and continue the process similarly. Eventually, $v$ is notified and can identify $v.L(T)$. Note that all the nodes on the way of this chain of messages cannot be faulty by the fact that they served to forward $T$ earlier, and that no quiescent time starts at least until all the notifications are done.

To save messages, if $v_i$ needs to send $v_{i-1}$ multiple reports for one team formation, it can again simulate the above using a single message.

As described so far, the storage for tokens is unbounded, and so are the
values of the counters. It is easy to bound the spaces and the values. When a report for $T$ is sent from a node $v_i$ to $v_{i-1}$, the space for $T$ can be recycled. Moreover, the values of the counters can be recomputed as if $T$ had never existed.

\section{Applications --- Additional Information}
\label{sec:applications-additional-information}
%
\def\ImpL{LE-I}
\def\Cand{Cand}
\def\ExpL{Expl}
\def\Term{Term}
In this section, we provide further technical details of some of the
applications discussed in \Sect{}~\ref{sec:intro:applications}.

\subsection{Leader Election}
 \label{sec:leader}
 \subsubsection{Implicit leader election}
  \label{sec:implicit}
  Implicit LE is defined in \Sect{}~\ref{sec:intro:applications}.
We first derive a simpler algorithm $\ImpL 1$ that has no termination detection for some nodes and add the termination detection later.

Informally, it is assumed that there are 
$( 1/2 + \epsilon ) n$
non-fragile nodes and each of them wakes up and starts the algorithm. Any fragile node may also be set by the adversary to be non-faulty at that point and start the algorithm, too. 
\begin{enumerate}
    \item 
Each node starting the algorithm flips a coin with probability $\frac{c\log n}{n}$ (for some constant $c$ that is ``large enough''; the selection of $c$ uses rather standard considerations, so we skip it here). Let $\Cand$ be the set of nodes that flipped $1$.
\item
Each node not in $\Cand$ sets its status to ``not-leader.''
\item
Each node in $\Cand$ generates a token and starts the algorithm for TF
with $\sigma= (1-\frac{\epsilon}{2})(\frac{1}{2}+\epsilon)c\log n$ (rounded up, let us ignore that for simplicity of notation).
\item
A node that deleted $\sigma$ tokens in TF sets its status to ``leader.''
\end{enumerate}

The following observation follows directly from \Thm{}~\ref{thm:main-informal}
and the algorithm. (For the second part, note that $\sigma=O(\log n)$.)
\begin{observation}
\label{obs:ImpL1}
Algorithm $\ImpL 1$ satisfies the following two properties:
\begin{itemize}
  
\item 
A node that is not the elected leader may not know that it is not the leader.
(We shall fix this later.)

\item 
The message complexity of Algorithm $\ImpL 1$ is
$O (\sqrt{n \log n} \cdot \log n)$
in expectation and
$O (\sqrt{n \log n} \cdot \log^2 n)$
whp.
The time complexity is
$O (\log n)$.\footnote{%
In spite of what footnote~\ref{footnote:asynchronous-time-complexity} says about the time complexity of TF, it is not hard to see that the time complexity 
claimed here for Leader Election holds for both interpretations of the definition of time.
(Recall that the use of TF for LE is a special case: the number of tokens is
only
$O (\log n)$,
the problem is not long-lived, and only a single team is formed.)
}
\end{itemize}
\end{observation}

\begin{lemma}
\label{lem:unique}    
 Algorithm $\ImpL 1$ elects a unique leader whp.   
\end{lemma}

\begin{proof}
Let us analyze the relation between $\sigma$ and the number $T$ of tokens generated by nodes in the third line of the algorithm. The number $n'$ of nodes that flipped coins may include any number  (between zero and  $(1/2-\epsilon)n)$ of fragile nodes;
hence,
$(1/2+\epsilon)n\leq n' \le n$.
Thus, for the expected number $\frac{c n' \log n}{n}$ of tokens generated, we know that
$(1/2 + \epsilon)c \log n \leq \frac{c n' \log n}{n} \leq c\log n$.
Applying Chernoff's bound, we find for the random variable $T$ (the number of tokens) that whp, 
\[
(1-\frac{\epsilon}{2})(\frac{1}{2}+\epsilon)c\log n
\leq
T
\leq (1+\frac{\epsilon}{2})c\log n.
\]
By the left inequality, $\sigma \leq T$ whp. 
Hence, at least one set of tokens can be deleted in TF, which means that at least one node is elected in step 4 of the algorithm.

It remains to show that $T<2\sigma$, so no more than one deletion event occurs (whp) and thus, no more than one leader is elected. 
By the above bounds on the number of tokens,
$T/\sigma
=\frac{T}{(1-\epsilon/2)(1/2+\epsilon)c\log n}
\leq
\frac{(1+\epsilon/2)c\log n}
{(1-\epsilon/2)(1/2+\epsilon)c\log n}
\leq 
\frac{(1+\epsilon/2)}
{(1-\epsilon/2)(1/2+\epsilon)}
$.
This is strictly smaller than 2. 
\end{proof}

\Paragraph{Termination Detection for Non-Leaders.}
Let us now show how to enhance algorithm $\ImpL 1$ such that every node that is \emph{not} elected will set its status to ``not-leader'' eventually, as is sometimes required.
That is, the promised 
$\ImpL$ algorithm will first run Algorithm $\ImpL 1$
(that runs
the TF algorithm).
After that, the elected leader will invoke
a termination detection module
$\Term$.

We came up with several methods of implementing $\Term$.
Let us describe two different implementations.
The first one uses the accumulation property of the algorithm (see
\Sect{}~\ref{sec:intro:extra-features}).
 The second one does not use the accumulation property and is given here to show that a reduction from LE to TF is possible without using additional features.
For both implementations, mark every token created in the third step of $\ImpL 1$ by ``LE token''. The implementations use an additional kind of token, called TERM (described below). 
The TF algorithm (run as a subroutine of the LE algorithm) does not distinguish between the two kinds of tokens. However, the $\Term$ part of the LE notices when a TERM token participates in a formation (second implementation) or is received (first implementation)
and makes use of this knowledge.

\Subparagraph{First Implementation.}
The leader $r$ (the node who performed the team formation) creates one TERM token.
The underlying TF algorithm treats the TERM token exactly like any other token. However, any non-leader node $v\neq r$ performing $\ImpL 1$ that receives a TERM token, note that it ($v$) is not a leader.

Recall that by the accumulation feature of the algorithm, all the tokens that remain in the networks after the team formation eventually are held (forever) by a single node $v$. If $v\neq r$, then $v$ knows it is not the leader. Let $w$ be any node other than $v$ and $r$. Either $w$ did not join $\Cand$, or it did join but at some point sent its token to another node. In both cases, $w$ knows that it is not the leader. It is easy to show that this implementation does not increase the order of the complexities of the LE election algorithm beyond that of $\ImpL 1$.

\Subparagraph{Second Implementation.}
In this implementation of $\Term$,
when the leader
node deletes $\sigma$ tokens, it generates $\sigma-1$ TERM tokens and continues to perform TF.
Any other node that deletes $\sigma$ tokens that include some 
$x$ TERM tokens and some $\sigma-x$ LE ones generates $x$ TERM ones.
Any node that
is \emph{not} the leader but ever sees a TERM token, changes its status to ``not-leader.'' The arguments for correctness are similar to those for the first implementation. Note that introducing those new tokens increases the complexities of the algorithm by a logarithmic factor.

\long\def\comment#1{}\long\def\comment#1{}
\comment{
\begin{observation}
\label{obs:blackbox}
\begin{enumerate}
    \item 

Algorithm $\ImpL 1$ plus the second implementation of $\Term$ 
elects a unique leader. Eventually, the leader sets its status to ``leader'' and every other non-faulty node sets its own status to ``not-leader.''

\item
The message complexity of the LE algorithm with the second implementation of $\Term$ is
   $O (\sqrt{n \log n}\cdot\log ^2n)$ in expectation and
$O (\sqrt{n \log^3 n} \cdot \log^2 n)$
messages whp. The time complexity is
$O (\log^2 n)$.
\end{enumerate}

\end{observation}

The unique leader follows from \Lem{}~\ref{lem:unique}. The time complexity follows from observing that the total number of tokens (LE and TERM) is $O(\sigma^2)$. This allows $(\sigma)$ deletion events, each with time complexity $(\log n +\sigma)$
(by
\Thm{}~\ref{thm:main-informal}, recall also that $\sigma=O(\log n)$.)
Similarly, for the message complexity, multiply the TF message complexity per token by $\sigma^2$ tokens.  
}
  
\subsubsection{Explicit LE}
 \label{sec:explicit}
So far, a non-leader node may not know which of its edges leads to the leader. To make the algorithm solve explicit leader election, the elected leader $r$ sends each of its neighbors one message, notifying the neighbor that $r$ is the leader. This adds $n$ messages and a single time unit to the complexities of the election.

\subsection{vTF and the Assembly of Complex Teams}
\label{sec:vTF}
Consider a team that requires 
$\sigma_1,\sigma_2,...,\sigma_k$
tokens of colors
$v_1,v_2,...,v_k$
correspondingly.
The vTF algorithm uses FT multiple times. First, it
runs in parallel $k$ versions of TF, each version $i$ deals only with tokens with color $c_i$ and ignores all the other kinds of tokens. When version $i$ deletes $\sigma_i$ tokens, the vTF algorithm generates a new token that we term super-token of color $c_i$.

We are left with the task of ``super-combining,'' that is, combining one super-token of each color into the final team. This is a much easier task than TF, since the different colors can be used for breaking the symmetry. 

Nevertheless, to demonstrate yet another use of TF, let us tailor TF slightly to accomplish super-combining too. 
We use $\sigma=2$
and change TF slightly. In TFdiff, the slightly tailored TF, tokens can be of one of two colors, and TFdiff only creates and deletes pairs consisting of one token of each color. The required change to the TF algorithm is very minor, and we leave the details to the reader.
(Hint: a node holding a token of some color $c$ is considered for transfer only to another node with a token of the other color we couple with $c$.)

 For simplicity of the description, let us assume that $k$ is a power of $2$. Generalizing to other values is trivial. 
Let us generalize TFdiff to a task 
(called Multi TFdiff) 
that the super combining will use $\log k$ times. 
The Multi TFdiff algorithm
 is given tokens of some $k^j$ different colors ($1\leq j < \log k$)
$c^j_1, c^j_2, ... c_{k^j}$
and combines them into super-tokens of $k^{j+1}=k^{k}/2$ new colors $c^{j+1}_1=c^j_{1,2},c^{j+1}_2 c^j_{3,4}, ... c^{j+1}_{k^j/2}= c^j_{k^j-1, k^j}$. It is easy to see that Multi TFdiff can be accomplished just by $k^j/2$ activations of the TFdiff algorithm.
Super-combining is now accomplished by $\log k-1$ activations of Multi TFdiff, where $k_1=k$, the tokens for the first application are the original super-tokens and the color $c_{i}$ of the $l+1$th activation $1<l\leq \log k-1$ results from colors $c_{2i-1},c_{2i}$ activation.

It is easy to see that the complexity of TFdiff is the same as that of TF.
The message complexity follows from the fact that a token may participate
in an activation of TF, incurring the message complexity of TF. Then \emph{two} tokens yield a token that participates in another activation, adding half the message complexity of TF per token, and so forth.

For the time complexity,
note that the vTF algorithm uses $O(k)$ activations of either the TF or vTF algorithms. Hence, the time for an activation is not larger than the time shown for TF multiplied by $O(k)$. However, only the first $k$ are started as a first batch at time zero  of the vTF algorithm (ending after $O(k)$ times the time of TF). An activation of the second batch is started when two specific activations of the first batch delete tokens; hence, the second batch starts at time zero + $O(k)$ times the time of TF and so forth for the next batches. There are $\log k$ batches, implying the claimed time complexity.

\subsection{Agreement with Failures}
As implicit agreement is directly reducible to the aforementioned
implicit leader election problem, we obtain an \emph{asynchronous} implicit
agreement algorithm that withstands
$n (1 / 2 - \epsilon)$
initial failures, whose message and time complexities match the (synchronous)
state-of-the-art of up to logarithmic factors.
In fact, our message complexity upper bound matches, up to logarithmic
factors, the lower bound of \cite{DBLP:conf/podc/AugustineMP18} (that holds already
in the synchronous case).

  \def\AS{A}
\subsection{Tokens Associated with Additional Information}

Let us assume that the environment keeps the invariant that the $\AS(T)$ is unique as long as $T$ stays in the system (does not participate in a team formation). Had we sent $\AS(T)$ together with $T$, this could have increased the message (and the time) complexity since it would not have been possible to send many tokens in a message of size $O(\log n)$ bits. (If $\sigma$ is large, then sending $\sigma$
times the size of an $\AS(T)$ cannot be sent in one message).

Let us describe informally an algorithm (to be called) TFA that treats associated information efficiently, nevertheless.
It uses the trace tree additional feature of the TF algorithm. 

TFA will use the TF algorithm, as well as the trace tree 
(of \Sect{}~\ref{sec:trace-tree-mechanism-details}). That is, consider a token $T$ that is injected into a node $v$. Only the token is sent in the TF algorithm, while $\AS(T)$ is kept in $v.L(T)$
(see \Sect{}~\ref{sec:trace-tree-mechanism-details}) where $v$ keeps track of its tokens. When $T$ participates in a team formation in some node $w$, node $v$ learns of this event using the reporting mechanism of \Sect{}~\ref{sec:trace-tree-mechanism-details}. It is left to transfer $\AS(T)$ to $w$ efficiently.

Recall that the TF algorithm has the property that a node does not send tokens to a node that holds no tokens. Hence, at the time of the team formation, $w$ is holding, among other tokens, a token $T'$ that was injected into $w$. Recall $\AS(T')$ is unique. TFA uses $\AS(T')$ temporarily as a unique ID of $w$. First, it is easy to change the mechanism of \ref{sec:trace-tree-mechanism-details} slightly, so that $v$ receives not only the information that $T$ participated in a team formation, but also 
$\AS(T')$. Node $v$ then tells all the utility nodes in $U(v)$ (those connected to $v$ (in the primary-utility graph) that it wishes to speak to the node whose temporary ID is $\AS(T')$. To complement that, node $w$ tells all the utility nodes in $U(w)$ that it is the one holding $\AS(T')$. Hence, some utility nodes can offer $v$ channels to $w$. Node $v$ selects one of these channels and sends $\AS(T)$ to $w$. 

Note that the message complexity of TFA is that of TF, plus that of the reporting (which is of the same order), plus $O((\sigma+1) |U(w)|)$ to notify
the utility nodes regarding $\AS(T')$
plus $O(\sigma)$ to send the associated information itself.
Similarly, the time complexity grows only by a constant factor (for the reporting) plus $O(\sigma)$.

\section{Discussion}
\label{sec:discussion}
We introduce and study the TF problem and
demonstrate its usefulness. We consider an asynchronous message-passing system assuming a complete network topology. 
It would be interesting to study TF
in other models of computation assuming, for example,
a more general network topology, synchronous systems, a situation where the nodes and/or the tokens are not anonymous,
assume that the nodes know the identities of their neighbors 
(i.e., the $KT_1$ model \cite{awerbuch1990trade}), changing the
assumption about the message size, etc.

Even for complete networks, interesting problems remain. Can the lower bound be improved? The TF solution helped us improve the known results for some applications.
Does it mean that the complexities of these applications are inherently similar to that of TF? or are there upper bounds for some of the applications better than the lower bound for TF?
We addressed the case of a small team size $\sigma$ that seems more practical and challenging; how about the case of a large $\sigma$?  The complexity per token in the case of a large $\sigma$ may be much lower. For example, if $\sigma=n+1$, no team formation is possible unless some node gets at least $2$ tokens injected. Hence, it makes sense that an algorithm at a node will do nothing as long as it holds a smaller number of tokens
(until approached by another node).

We generalized the TF problem to address a team represented by vectors such as $x_1$ tokens of color red \textit{and} $x_2$ tokens of color blue (e.g., $x_1$ sorcerers and $x_2$ warriors). How about using other logical operations when composing teams? For example, $x_1$ sorcerers \textit{and}  either $x_2$ warriors \textit{or} $x_3$ thieves?

Our algorithms have several novel properties. These properties may be useful for studying other problems, not just TF. One of the properties is, in the context of long-lived problems, resilience to failures \textit{stronger} than initial failures. This might be useful for other long-lived problems. It would be interesting to compare the power of the model that allows such failures to the power of other failure models.

Another important notion we introduced in the context of randomized algorithms is that of a \textit{forgetful} algorithm.
While, for deterministic algorithms, this notion is the same as that of a memoryless algorithm, 
for randomized algorithms, this notion is not as strong as one
may expect, as we assume that a node will never forget its initial coin tosses. We use this slightly limited version to ensure our algorithms succeed with high probability, even for infinite executions. 
Had we sufficed with high probability for each team formation separately, 
then we would have been able to remove this limitation and ensure that a node does
not remember anything (not even its initial coin tosses) in the TF algorithm.\footnote{When allowing a node to fully forget everything,
in an infinite number of team formation where each succeeds whp
separately, some formations might not succeed -- our slightly limited notion of forgetfulness prevents this situation from happening.}

Requiring that nodes forget their past behavior in some cases
usually leads to simple solutions, saves the cost of maintaining data structures during quiescent times, and eases the recovery in case of catastrophes, since no data structure needs to be recovered.
We emphasize again that 
in long-lived problems (like mutual exclusion and team formation), a naive solution where a leader is elected and forever coordinates all activities is not forgetful (or memoryless), even in fault-free models.

Let us note that various related notions appeared in the literature.
For example, notions of ``not keeping information about past behavior'' are promoted in practice, e.g., \cite{stateless}, and theory, e.g., \cite{DBLP:conf/icdcn/JayantiJJ20}, because they are considered more scalable and easier to recover. 

An interesting property of our TF algorithm seems to be related to privacy and requires further discussion (we have not analyzed this property formally). Even in anonymous networks, nodes in many algorithms (but not the algorithm of this paper) choose identities that are unique whp. Suppose that node $u$ sends message $M_1$ to node $v$ and messages $M_2$ to node $w$. Nodes $v$ and $w$ may now compare $M_1$ and $M_2$ knowing that they both originated from $u$. This may harm $u$'s privacy. For example, $u$ may not want other nodes to know that its token participated in a team since this team has some private purpose. Let us say that the algorithm that has this property is \textit{inconspicuous}. It would be interesting to define it formally,
prove formally that our algorithm satisfies this property, 
and study the power of models that allow only such algorithms.

Our generalization of initial failures is not weaker than crash failures. When there are only crash failures, in an infinite execution, a node either never crashes or it is alive only for a finite time (the finite prefix before it crashes). This claim is not valid with our ``fragile failures'' since there is no limit on how many times the status of a node can be toggled.
Another observation regarding crash failures is that TF is not solvable, even with randomization, in the presence of a single crash failure. This observation follows from the fact that a node may crash while holding a token or just before a token is sent to it.
(This impossibility holds even with the restriction that nodes with tokens may not crash.)
It is possible to modify the liveness condition in the definition of TF so that the above argument does not apply.

It is interesting to compare our adversarial model to that of \cite{DBLP:conf/atal/DatarNA23}, which distinguishes between weakly adaptive and strongly adaptive adversaries. In their synchronous setting, a strongly adaptive adversary can exploit the random bits that determine the nodes' next moves in that round. A natural asynchronous analogue is an adversary that, at any time, knows the random bits that would be used in each possible next activation event; in their terminology, our adversary is weakly adaptive.

Our current algorithm does not withstand this stronger adversary: it can violate the intersection property of \Lem{}~\ref{lem:pu-graph} by seeing in advance how the primary-utility neighborhoods of two primary nodes intersect, making their common utility nodes fragile (and faulty when the execution starts), and then injecting $\sigma/2$ tokens into each of the two primary nodes. As a result, these primary nodes may fail to discover each other, preventing team formation despite the presence of $\sigma$ tokens in the system. It would be interesting to determine whether one can adapt our algorithm to withstand this stronger adversary, or whether an entirely different approach is required.

Finally, 
we point out an extensively studied biological phenomenon similar to the one captured by TF.
\textit{Quorum sensing} is a process in which bacteria can sense that
the number of bacteria ready to release
toxins (which cause disease in plants, animals, and humans) 
has reached a certain threshold.
This enables them to release the toxins at approximately the same time
\cite{BB2019,FWG1994}.
Several insects, like ants and honey bees, have been shown to also use quorum sensing in a process that resembles collective decision-making \cite{Pratt2005,SV2004}.
Unlike the trigger counting problem (mentioned in the Introduction), here, most participants must know
when the threshold is reached.


\section{Pseudocode}
\label{sec:pseudocode}

This section presents the pseudocode of \AlgTF{}. \Sect{}~\ref{sec:pseudocode-principal} gives the pseudocode of the principal layer (Algorithms~\ref{alg:principal-layer}--\ref{alg:procedure-begin-new-phase}), together with its message and variable tables (Tables~\ref{tab:messages-principal} and~\ref{tab:variables-principal}).
\Sect{}~\ref{sec:pseudocode-channel} does the same for the channel layer (Algorithms~\ref{alg:channel-layer-primary}--\ref{alg:procedure-create-channel} and Tables~\ref{tab:messages-channel} and~\ref{tab:variables-channel}).

\subsection{Principal Layer}
\label{sec:pseudocode-principal}

\begin{table}[ht]
\centering
\begin{tabular}{@{}c|ccl@{}}
\toprule
& message  & sender phase type & semantics
\\
\hline

\multirow{2}{*}{request} &
$\msgTokensPlease$ &
$\Center$ & 
requests tokens
\\
&
$\msgWaiting$ &
$\Arm$ &
asks whether to wait before phase ends
\\

\hline

\multirow{3}{*}{response} &
$\msgNoTransport$ &
$\Center$ &
declines token request
\\
&
$\msgTransport(k)$ &
$\Arm$ &
transports $k$ tokens in response
\\
&
$\msgGoOn$ &
$\Arm$ &
tells peer not to wait
\\

\bottomrule
\end{tabular}
\caption{\label{tab:messages-principal}%
The messages used by \AlgTF{}'s principal layer.
The sender (resp., receiver) is a primary node
$p \in P$
(resp.,
$p' \in P$)
such that $p$ and $p'$ share a channel
$\chi \in \Channels(p) \cap \Channels(p')$.}
\end{table}

\begin{table}[ht]
\centering
\begin{tabular}{@{}ccp{0.5\textwidth}@{}}
\toprule
variable & initial value & semantics
\\
\midrule

$p.\varTokens$ &
$0$ &
the number of tokens held by $p \in P$
\\

$p.\varPhaseType$ &
$\bot$ &
current phase type of
$p \in P$
\\

$p.\varAwaitingResponse(\chi)$ &
$\False$ &
whether
$p \in P$
awaits response over 
$\chi \in \Channels(p)$
before ending current phase
\\

$p.\varDelayingResponse(\chi)$ &
$\False$ &
whether
$p \in P$
owes a response over
$\chi \in \Channels(p)$
to a past
$\msgWaiting$
request received over $\chi$
\\

\bottomrule
\end{tabular}
\caption{\label{tab:variables-principal}%
The local variables accessed by \AlgTF{}'s principal layer.}
\end{table}

\FloatBarrier

\begin{algorithm}[ht]
\begin{algorithmic}[1]

\Event{$k > 0$ tokens are injected into $p$ while non-operational ($\Channels(p) = \emptyset$)}
  \State{$\varTokens \gets \varTokens + k$}
  \State{notify channel layer that token count increased} 
  \If{$\varTokens \geq \sigma$}
    \State{call $\ProcFormTeams()$}
  \EndIf
\EndEvent

\Event{the channel layer notifies about a channel $\chi$ that is added to $\Channels(p)$}
  \State{$\varAwaitingResponse(\chi) \gets \False$}
  \State{$\varDelayingResponse(\chi) \gets \False$}
  \If{$\Channels(p) - \{\chi\} = \emptyset$} 
    \State{call $\ProcBeginNewPhase()$} 
  \EndIf
\EndEvent

\Event{the channel layer notifies about a channel $\chi$ that is removed from $\Channels(p)$}
  \State{delete $\varAwaitingResponse(\chi)$ and $\varDelayingResponse(\chi)$}
  \State{call $\ProcCheckEndPhase()$}
\EndEvent

\Event{$p$ receives request message 
$\msgTokensPlease$
over channel $\chi \in \Channels(p)$}
  \If{$\varPhaseType = \Center$}
    \State{send $\msgNoTransport$ over $\chi$}
  \ElsIf{$\varPhaseType = \Arm$}
    \State{send $\msgTransport(\varTokens)$ over $\chi$}
    \State{$\varTokens \gets 0$}
    \State{notify channel layer of non-busy status}
  \EndIf
\EndEvent

\Event{$p$ receives request message $\msgWaiting$ over channel $\chi \in \Channels(p)$}
  \If{$\varPhaseType = \Center$}
    \State{$\varDelayingResponse(\chi) \gets \True$}
  \ElsIf{$\varPhaseType = \Arm$}
    \State{send $\msgGoOn$ over $\chi$}
  \EndIf
\EndEvent

\Event{$p$ receives response message
$\mu \in \left\{ \msgGoOn,\, \msgNoTransport,\, \msgTransport(k) \right\}$
over channel
$\chi \in \Channels(p)$}

    \If{$\mu = \msgTransport(k)$}
        \State{$\varTokens \gets \varTokens + k$}
        \State{notify channel layer that token count increased}
    \EndIf
    
    \If{($\varPhaseType = \Arm$ and $\mu = \msgGoOn$) or ($\varPhaseType = \Center$ and $\mu \neq \msgGoOn$)}
      \State{$\varAwaitingResponse(\chi) \gets \False$}
      \State{call $\ProcCheckEndPhase()$}
    \EndIf
\EndEvent
\end{algorithmic}
\caption{\label{alg:principal-layer}%
$\AlgTF{}$'s principal layer from the perspective of a primary node
$p \in P$.}
\end{algorithm}

\begin{algorithm}[ht]
\begin{algorithmic}[1]

\If{$\varPhaseType \neq \bot$ and 
$\varAwaitingResponse(\chi) = \False$ for all
$\chi \in \Channels(p)$}
    \If{$k>0$ tokens were injected during the phase}
      \State{$\varTokens \gets \varTokens + k$}
      \State{notify channel layer that token count increased} 
    \EndIf
    
    \If{$\varTokens \geq \sigma$}
      \State{call $\ProcFormTeams()$}
    \ElsIf{$\Channels(p) = \emptyset$}
      \State{$\varPhaseType \gets \bot$}
    \Else
      \State{call $\ProcBeginNewPhase()$}
    \EndIf
\EndIf
\end{algorithmic}
\caption{\label{alg:procedure-check-end-phase}%
$\ProcCheckEndPhase$ from the perspective of a primary node
$p \in P$.}
\end{algorithm}

\begin{algorithm}[ht]
\begin{algorithmic}[1]

\State{form $\lfloor \varTokens / \sigma \rfloor$ teams}
\State{$\varPhaseType \gets \bot$}
\State{$\varTokens \gets 0$} \State{notify channel layer of non-busy status}

\If{there were $r>0$ remainder tokens}
    \State{treat remainder tokens as a (fake) injection of $r$ tokens shortly after this event}
\EndIf
\end{algorithmic}
\caption{\label{alg:procedure-form-teams}%
$\ProcFormTeams$ from the perspective of a primary node
$p \in P$.}
\end{algorithm}

\begin{algorithm}[ht]
\begin{algorithmic}[1]

\State{pick $\varPhaseType \in \{ \Center, \Arm \}$ uniformly at random}
\For{each $\chi \in \Channels(p)$}
  \If{$\varPhaseType = \Center$}
    \State{$\mu \gets \msgTokensPlease$}
  \Else 
    \State{$\mu \gets \msgWaiting$}
    \If{$\varDelayingResponse(\chi) = \True$}
      \State{send $\msgGoOn$ over $\chi$}
    \EndIf
  \EndIf
  \State{$\varDelayingResponse(\chi) \gets \False$}
  \State{$\varAwaitingResponse(\chi) \gets \True$}
  \State{send $\mu$ over $\chi$}
\EndFor
\end{algorithmic}
\caption{\label{alg:procedure-begin-new-phase}%
$\ProcBeginNewPhase$ from the perspective of a primary node
$p \in P$.}
\end{algorithm}

\FloatBarrier

\subsection{Channel Layer}
\label{sec:pseudocode-channel}

\begin{table}[ht]
\centering
\begin{tabular}{@{}cccl@{}}
\toprule
message & sender & receiver & semantics
\\
\midrule
$\msgBusy$ &
$p \in P$ &
$u \in U(p)$ &
$p$ is busy
\\
$\msgBusyAck$ &
$u \in U$ &
$p \in P(u)$ &
$u$ acknowledges that $p$ is busy
\\
$\msgTokensUpdate(k)$ &
$p \in P$ &
$u \in U(p)$ &
$p$ (currently) holds $k$ tokens
\\
$\msgNotBusy$ &
$p \in P$ &
$u \in U(p)$ &
$p$ is no longer busy
\\
$\msgChannel$ &
$u \in U$ &
$p \in P(u)$ &
$u$ mediates a channel of $p$
\\
$\msgNoChannel$ &
$u \in U$ &
$p \in P(u)$ &
$u$ no longer mediates a channel of $p$
\\
$\msgChannelAck$ &
$p \in P$ &
$u \in U(p)$ &
automatic response to incoming $\msgChannel$
\\
\bottomrule
\end{tabular}
\caption{\label{tab:messages-channel}%
The messages used by \AlgTF{}'s channel layer.}
\end{table}

\begin{table}[ht]
\centering
\begin{tabular}{@{}ccp{0.65\textwidth}@{}}
\toprule
variable & initial value & semantics
\\
\midrule
$p.\varTokens$ &
$0$ &
the number of tokens held by
$p \in P$
\\
$p.\varBusyAcked$ &
$\emptyset$ &
the set of nodes
$u \in U(p)$
that acknowledged that
$p \in P$
is busy
\\
$p.\varMediators$ &
$\emptyset$ &
the set of nodes
$u \in U(p)$
that mediate channels of
$p \in P$
\\
$u.\varBusyTokens(p)$ &
$\bot$ &
the perspective of
$u \in U$
on the number of tokens held by busy node
$p \in P(u)$
\\
$u.\varChannel$ &
$\emptyset$ &
the nodes in $P(u)$ that are mediated by
$u \in U$
\\
$u.\varDiff(p)$ &
$0$ &
the number of $\msgChannelAck$ messages that
$u \in U$
``expects'' to receive from
$p \in P$
\\
\bottomrule
\end{tabular}
\caption{\label{tab:variables-channel}%
The local variables accessed by \AlgTF{}'s channel layer.}
\end{table}

\FloatBarrier

\begin{algorithm}[ht]
\begin{algorithmic}[1]
\Event{%
the principal layer notifies about a change of $\varTokens$ from
$\varTokens = 0$
to
$\varTokens > 0$
}
  \State{%
send $\msgBusy$ to each
$u \in U(p)$}
\EndEvent
\Event{%
$p$ receives $\msgBusyAck$ from
$u \in U(p)$}
  \If{%
$\varTokens > 0$}
    \State{%
$\varBusyAcked \gets \varBusyAcked \cup \{ u \}$}
    \State{%
send $\msgTokensUpdate(\varTokens)$ to $u$}
  \Else
    \State{%
send $\msgNotBusy$ to $u$}
  \EndIf
\EndEvent
\Event{%
the principal layer notifies about a change of $\varTokens$ from
$0 < \varTokens = k$
to
$\varTokens > k$}
  \State{%
send $\msgTokensUpdate(\varTokens)$ to each
$u \in \varBusyAcked$}
\EndEvent
\Event{%
the principal layer notifies about a change of $\varTokens$ from
$\varTokens > 0$
to
$\varTokens = 0$
}
  \State{send $\msgNotBusy$ to each
$u \in \varBusyAcked - \varMediators$}
  \State{$\varBusyAcked \gets \varMediators$}

  \For{each node $u \in \varMediators$}
    \State{$\varMediators \gets \varMediators - \{ u \}$}
    \State{send $\msgNotBusy$ to $u$}
    \State{%
    $\varBusyAcked \gets \varBusyAcked - \{ u \}$}
    \State{notify principal layer that channel $\chi$ that $u$ mediated for $p$ was removed from $\Channels(p)$}
  \EndFor

\EndEvent
\Event{%
$p$ receives $\msgChannel$ from
$u \in U(p)$}
  \State{%
send $\msgChannelAck$ to $u$}
  \If{%
$\varTokens > 0$
and
$u \in \varBusyAcked$}
    \State{%
      $\varMediators \gets \varMediators \cup \{ u \}$}
    \State{notify principal layer that $u$ mediates a new channel $\chi$ for $p$ that was added to $\Channels(p)$}
  \EndIf
\EndEvent
\Event{%
$p$ receives $\msgNoChannel$ from
$u \in U(p)$}
  \State{%
$\varMediators \gets \varMediators - \{ u \}$}
  \State{notify principal layer that channel $\chi$ that $u$ mediated for $p$ was removed from $\Channels(p)$}
\EndEvent
\Event{%
$p$ receives a (principal layer) relayed message $\mu$ from
$u \in U(p)$}
  \If{%
$u \notin \varMediators$}
    \State{%
screen $\mu$}\Comment{%
$\mu$ is not passed to the principal layer}
  \EndIf
\EndEvent
\end{algorithmic}
\caption{\label{alg:channel-layer-primary}%
\AlgTF{}'s channel layer from the perspective of a primary node
$p \in P$.}
\end{algorithm}

\begin{algorithm}[ht]
\begin{algorithmic}[1]
\Event{%
$u$ receives $\msgBusy$ from
$p \in P$}
  \If{%
$\varBusyTokens(p) = \bot$}
    \State{%
$\varBusyTokens(p) \gets 0$}
    \State{%
send $\msgBusyAck$ to $p$}
  \EndIf
\EndEvent
\Event{%
$u$ receives $\msgTokensUpdate(k)$ from
$p \in P$}
  \State{%
$\varBusyTokens(p) \gets k$}
  \State{%
call $\ProcCreateChannel()$}
\EndEvent
\Event{%
$u$ receives $\msgNotBusy()$ from
$p \in P$}
  \State{%
$\varBusyTokens(p) \gets \bot$}
  \If{%
$p \in \varChannel$
}
    \State{%
$\{ p' \} \gets \varChannel - \{ p \}$}
    \State{%
send $\msgNoChannel$ to $p'$}
    \State{%
$\varChannel \gets \emptyset$}
    \State{%
call $\ProcCreateChannel()$}
  \EndIf
\EndEvent
\Event{%
$u$ receives $\msgChannelAck$ from
$p \in P$}
  \State{%
$\varDiff(p) \gets \varDiff(p) - 1$}
\EndEvent
\Event{%
$u$ receives a (principal layer) relayed message $\mu$ from 
$p \in \varChannel$}
  \If{%
$\varDiff(p) = 0$}
    \State{%
$\{ p' \} \gets \varChannel - \{ p \}$}
    \State{%
relay $\mu$ to $p'$}
  \Else
    \State{%
screen $\mu$}
  \EndIf
\EndEvent
\end{algorithmic}
\caption{\label{alg:channel-layer-utility}%
\AlgTF{}'s channel layer from the perspective of a utility node
$u \in U$.}
\end{algorithm}

\begin{algorithm}[ht]
\begin{algorithmic}[1]
\State{%
$Q \gets \left\{
p \in P
\, : \,
\varBusyTokens(p) \neq \bot \, \land \, \varBusyTokens(p) > 0 \right\} $}
\If{%
$\varChannel = \emptyset$
and
$|Q| \geq 2$}
  \State{%
$p_{1} \gets \argmax_{p \in Q} \varBusyTokens(p)$}\Comment{%
apply first-come-first-served tie breaking}
  \State{%
$p_{2} \gets
\argmax_{p \in Q - \{ p_{1} \}} \varBusyTokens(p)$}\Comment{%
apply first-come-first-served tie breaking}
  \For{%
$i \in \{ 1, 2 \}$}
    \State{%
$\varChannel \gets \varChannel \cup \{ p_{i} \}$}
    \State{%
send $\msgChannel$ to $p_{i}$}
    \State{%
$\varDiff(p_{i}) \gets \varDiff(p_{i}) + 1$}
  \EndFor
\EndIf
\end{algorithmic}
\caption{\label{alg:procedure-create-channel}%
$\ProcCreateChannel$ from the perspective of a utility node
$u \in U$.}
\end{algorithm}



\clearpage

\bibliography{references.bib}

\end{document}